\documentclass[11pt]{article}
\usepackage{fullpage}

\usepackage[T1]{fontenc}
\usepackage[english]{babel}

\usepackage{amsmath,amssymb,amsthm,mathtools,esint}
\usepackage{tikz,graphicx}
\usetikzlibrary{arrows.meta}
\usetikzlibrary{decorations.markings}
\usetikzlibrary{calc}
\usepackage{graphicx,subcaption}
\usepackage{enumerate}
\usepackage{dsfont}
\usepackage{verbatim}
\usepackage{hyperref}

\theoremstyle{plain}
\newtheorem*{theorem*}{Theorem}
\newtheorem{theorem}{Theorem}[section] 
\newtheorem{lemma}[theorem]{Lemma}
\newtheorem{proposition}[theorem]{Proposition}
\newtheorem{corollary}[theorem]{Corollary}

\newtheorem*{assumption*}{Assumption}
\newtheorem{example}[theorem]{Example}
\theoremstyle{definition}
\newtheorem{definition}[theorem]{Definition}
\newtheorem{remark}[theorem]{Remark}

\numberwithin{equation}{section}

\mathtoolsset{showonlyrefs}

\newcommand{\eps}{\varepsilon}

\newcommand{\ZZ}{\mathbb{Z}}
\newcommand{\RR}{\mathbb{R}}
\newcommand{\CC}{\mathbb{C}}

\newcommand{\Ordo}{\mathcal{O}}

\newcommand{\EE}{\mathbb{E}}
\newcommand{\PP}{\mathbb{P}}

\DeclareMathOperator{\re}{Re}

\DeclareMathOperator{\sgn}{sgn}
\renewcommand{\d}{\,\mathrm{d}}
\renewcommand{\i}{\mathrm{i}}
\newcommand{\e}{\mathrm{e}}
\newcommand{\map}{\Psi}

  \tikzset{
  compass/.pic = {
    \foreach[count=\i,evaluate={\m=div(\i-1,4);\a=90*\i-45*(\m+1)}] \d in {NE,NW,SW,SE,E,N,W,S}{
      \filldraw[pic actions,rotate=\a,scale=.7+.3*\m] (0,0) -- (45:1)--(0:3) node[scale=3, transform shape,rotate=-90,above]{\d};
      \filldraw[pic actions,fill=white,rotate=\a,scale=.7+.3*\m] (0,0) -- (-45:1)--(0:3)--cycle;
    };
  }
}
  \tikzset{
  compassOp/.pic = {
    \foreach[count=\i,evaluate={\m=div(\i-1,4);\a=90*\i-45*(\m+1)}] \d in {NW,NE,SE,SW,W,N,E,S}{
      \filldraw[pic actions,rotate=\a,scale=.7+.3*\m] (0,0) -- (45:1)--(0:3) node[scale=3, transform shape,rotate=-90,above]{\d};
      \filldraw[pic actions,fill=white,rotate=\a,scale=.7+.3*\m] (0,0) -- (-45:1)--(0:3)--cycle;
    };
  }
}

\DeclareMathOperator{\Log}{Log}

\DeclareMathOperator{\adj}{adj}
\DeclareMathOperator{\one}{\mathds{1}}

\DeclareMathOperator{\convHull}{ConvexHull}

\graphicspath{ {Images/} }

\title{Crystallization of the Aztec diamond}

\author{Tomas Berggren\footnote{Department of Mathematics, Royal Institute of Technology (KTH), Lindstedsvägen 25., SE-100 44 Stockholm, Sweden. E-mail: tobergg@kth.se}
	\and Alexei Borodin\footnote{Department of Mathematics, Massachusetts Institute of Technology, 77 Massachusetts Ave., Cambridge, MA 02139, USA. E-mail: borodin@math.mit.edu}}

\date{}

\begin{document}

\maketitle

\begin{abstract} 
We consider dimer models on growing Aztec diamonds, which are certain domains in the square lattice, 
with edge weights of the form~$\nu(\,\cdot\,)^\beta$, where~$\nu(\,\cdot\,)$ is a doubly periodic function on the edges of the lattice and~$\beta$ is an inverse temperature parameter. 

We prove that in the zero-temperature ($\beta\to\infty$) limit, and for generic values of~$\nu(\,\cdot\,)$, these dimer models undergo crystallization: The limit shape converges to a piecewise linear function called the tropical limit shape, and the local fluctuations are governed by the Gibbs measures with the slope dictated by the tropical limit shape for high enough values of~$\beta$. 

We also show that the tropical limit shape and the tropical arctic curve (consisting of ridges of the crystal) are described in terms of a tropical curve and a tropical action function on that curve, which are the tropical analogs of the spectral curve and the action function that describe the finite-temperature models. The tropical curve is explicit in terms of the edge weights, and the tropical action function is a solution of Kirchhoff's problem on the tropical curve. 
\end{abstract}

\tableofcontents

\section{Introduction}

\subsection{Preface} 
Planar dimer models, a subject of active research since the work of Kasteleyn~\cite{Kas61} and Temperley--Fisher~\cite{TF61} in the early 1960's, have recently seen significant advancements in understanding models with periodic edge weights, particularly for the Aztec diamond. The present paper introduces a temperature parameter into the models and investigates the asymptotic behavior of the Aztec diamond dimer covers in the zero-temperature limit. 

Our motivation came for the 1970's and 1980's physics literature, where it was suggested that planar dimer models or, more generally, Solid-On-Solid (SOS) models were relevant for describing the experimental phenomenon of \emph{roughening transition} in equilibrium crystals, see, \emph{e.g.}, Nienhuis--Hilhorst--Bl\"{o}te~\cite{NHB84}, Rottman--Wortis~\cite{RW84}, and references therein. Crystals are known to have a smooth boundary consisting of facets at low temperatures. As the temperature increases, the facets gradually turn into a curved boundary that is more rough. See, \emph{e.g.}, Balibar--Alles--Parshin~\cite{BAP05} for a detailed exposition in the case of helium. 

Our initial goal was to see how the roughening transition manifested itself in dimer models. One might \emph{a priori} expect that the temperature parameter should induce a transition from a frozen state (the delta measure on a single dimer cover with the highest weight) in the zero-temperature limit, to complete randomness (the uniform measure) in the infinite temperature limit. Previously, a detailed description of such a transition was provided in two~$1$-parameter families of lozenge tilings of a hexagon by Charlier--Duits--Kuijlaars--Lenells in~\cite{CDKL19} and by Charlier in~\cite{Cha20a}. Such settings may also offer interesting scaling limits. In particular, Mason~\cite{Mas22}, see also Chhita~\cite{Chh12} and Berestycki--Haunschmid-Sibitz~\cite{BHS22}, argued that the free fermion sine-Gordon field arises in a scaling limit that can be viewed as sending the temperature to infinity. Another  scaling limit that corresponds to the temperature being sent to zero allows one to access the frozen--smooth boundary; a transition of this kind was analyzed in a two-periodic Aztec diamond by Johansson--Mason~\cite{JM23}.

In our recent work~\cite{BB23}, we were able to analyze the dimer models on the Aztec diamond (a specific type of domains in the square lattice) for arbitrary (generic) doubly periodic edge weights. Our analysis substantially relied on an algebraic geometric formalism that was originally developed in the groundbreaking works of Kenyon--Okounkov--Sheffield~\cite{KOS06} and Kenyon--Okounkov~\cite{KO06}. In particular, they established a correspondence between dimer models on weighted periodic graphs and their \emph{spectral data} consisting of a \emph{spectral curve} as well as a point on its Jacobian. Notably, they showed that these spectral curves are always \emph{Harnack} -- a class of curves that had been previously identified and studied by Mikhalkin~\cite{Mik00}. In~\cite{BB23}, for the Aztec diamond we were able to describe the phase separating \emph{arctic} curve, the limit shape, and the local fluctuations of the dimer model in terms of the spectral curve and a new element of the geometric data termed the \emph{action function}. Very recently, some of our results were extended, by different methods, to quasi-periodic edge weights by Boutillier--de Tili\`ere~\cite{BdT24} and Bobenko--Bobenko--Suris~\cite{BBS24, BB24}, and the work by Bobenko--Bobenko~\cite{BB24} also included the case of the lozenge tilings of a hexagon.

The main goal of the present paper is to describe the zero-temperature limit of the Aztec diamond dimer model with arbitrary generic doubly periodic edge weights. One can view this limit as \emph{crystallization} -- the limit shape becomes piecewise linear and its curved part disappears, cf. Figure~\ref{fig:varying_temp}. The arctic curve and the limit shape are again described via geometric data, but now the algebraic curve is replaced by a \emph{tropical curve} and the action function by its tropical analog, which is a harmonic function on the tropical curve with certain boundary conditions. That is, tropical geometry assumes the role of algebraic geometry. To us, the similarities between the use of the algebraic geometry in the analysis of the finite-temperature models and the use of the tropical geometry in the analysis of their zero-temperature limit are striking.

\begin{figure}[t]
    \begin{center}
     \includegraphics[angle=90, scale=.23]{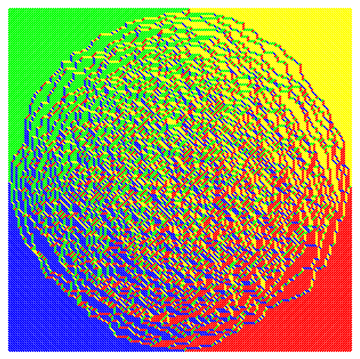}
     \quad
     \includegraphics[angle=90, scale=.23]{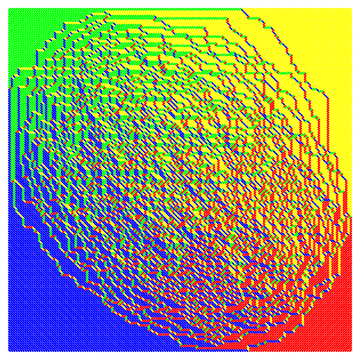}
     \quad
     \includegraphics[angle=90, scale=.23]{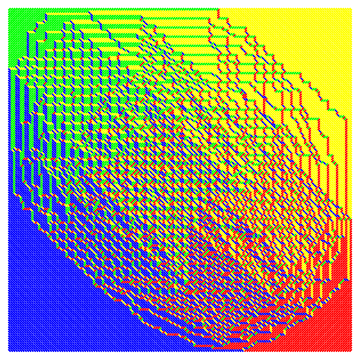}
     \quad
     \includegraphics[angle=90, scale=.23]{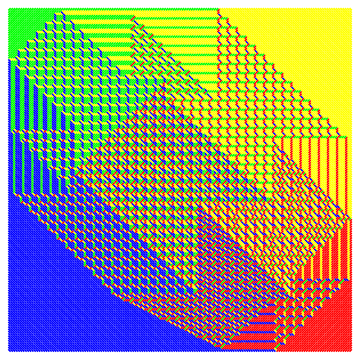}
    \end{center}
   \caption{Random samplings of the Aztec diamond of size~$N=150$ and with edge weights of periodicity~$k=4$ and~$\ell=4$. The temperature parameter~$\beta$ is, from left to right:~$\beta=0$,~$\beta=1$,~$\beta=2$,~$\beta=40$.
   \label{fig:varying_temp}}
   \end{figure}

We also consider the zero-temperature limit of the ergodic translation-invariant Gibbs measures that describe local fluctuations.  
For generic edge weights, local configurations will freeze in this limit. However, it may happen that the limit shape still tends to a piecewise linear limit, yet the local Gibbs measures have a nontrivial limiting behavior. We provide an explicit formula for the correlations of the Gibbs measures in the zero-temperature limit. 

Let us now describe our results in more detail. 

\subsection{The dimer model and its zero-temperature limit}
In this paper, we investigate the zero-temperature limit of the periodically weighted Aztec diamond dimer model. The \emph{Aztec diamond graph}~$G_\text{Az}$ is a subgraph of the square lattice as illustrated in Figure~\ref{fig:height_function}. A precise definition is provided in Section~\ref{sec:measures}. A \emph{dimer cover}, also known as \emph{perfect matching} of a graph is a subset of its edges, with elements called \emph{dimers}, such that each vertex is covered by exactly one dimer. A dimer model is a probability measure on the set of all possible dimer covers of the graph. See Kenyon~\cite{Ken09} and Gorin~\cite{Gor21} for surveys of the planar dimer models.  

To define our probability measures on the dimer covers of Aztec diamonds, we let~$\nu$ be a doubly periodic edge weight function, that is, a positive periodic function of the edges of~$G_\text{Az}$. Let us denote its vertical and horizontal periods by~$k$ and~$\ell$, respectively. Further, we introduce an inverse temperature parameter~$\beta>0$. The probability measures of interest to us have the form
\begin{equation}\label{eq:intro:model_beta}
\PP_{\text{Az},\beta}(\mathcal D_\text{Az})=\frac{1}{Z_\beta}\prod_{e\in \mathcal D_\text{Az}}\nu(e)^\beta=\frac{1}{Z_\beta}\,\e^{\beta \sum_{e\in \mathcal D_\text{Az}}\log \nu(e)},
\end{equation}
where the product is over all edges~$e$ of a dimer cover~$\mathcal D_\text{Az}$ of the Aztec diamond and~$Z_\beta$ is the \emph{partition function}. Figure~\ref{fig:varying_temp} shows samples from such probability measures with the same weight function but different values of~$\beta$. For~$\beta=1$, and, hence, for all finite~$\beta>0$, this model was asymptotically analyzed in~\cite{BB23} in the limit when the size of the Aztec diamond tends to infinity. The limiting object was described in terms of an associated \emph{spectral curve} together with an \emph{action function}. In this paper, we focus on the zero-temperature limit,~$\beta \to \infty$. Analogously to the finite~$\beta$ case, we will describe the limiting object using an associated \emph{tropical curve} and a \emph{tropical action function}.

 \begin{figure}[t]
 \begin{center}
  \begin{subfigure}[c]{0.22\textwidth}
\includegraphics[scale=.05]{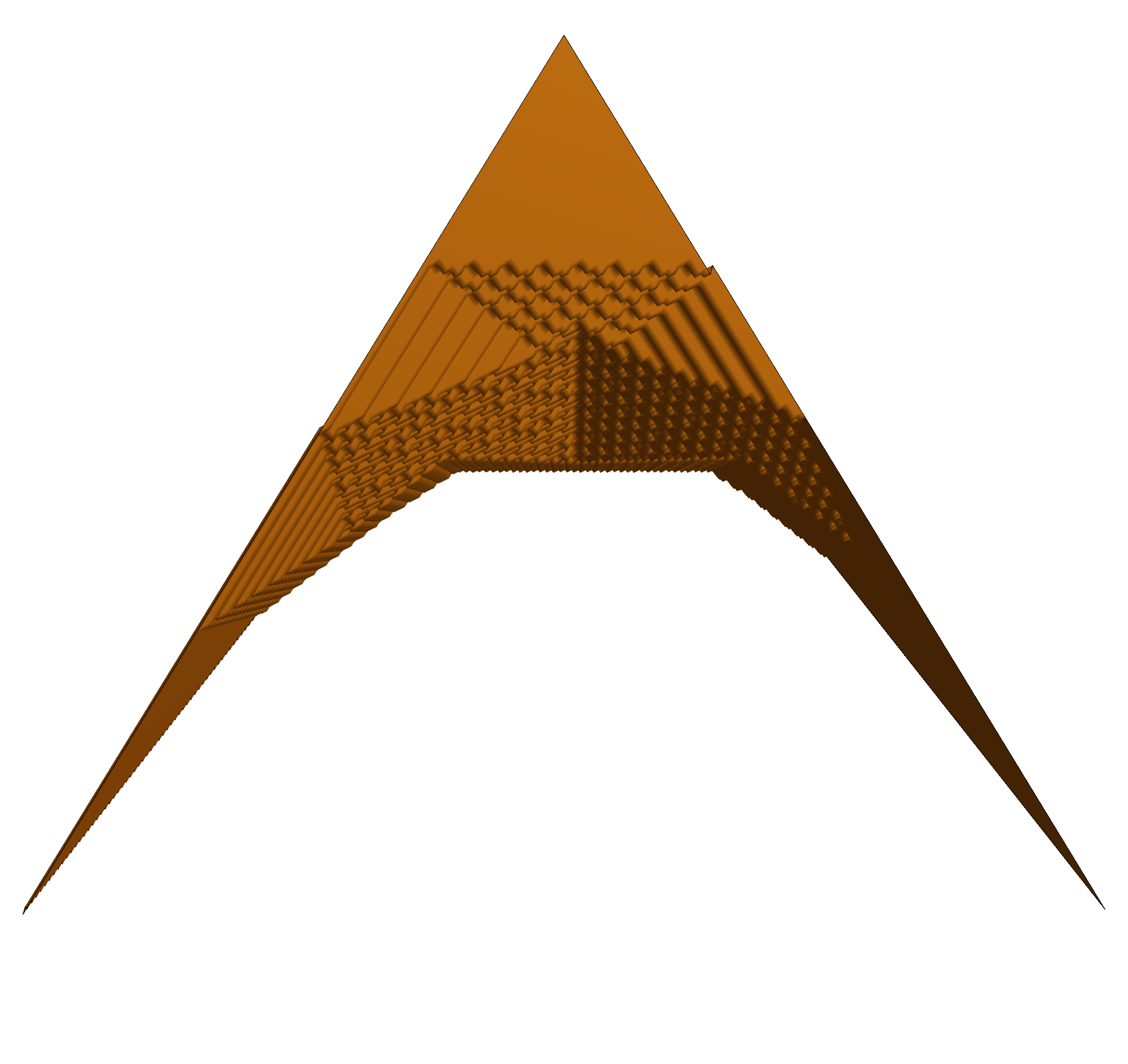}
\end{subfigure}
\quad
  \begin{subfigure}[c]{0.22\textwidth}
\includegraphics[scale=.05]{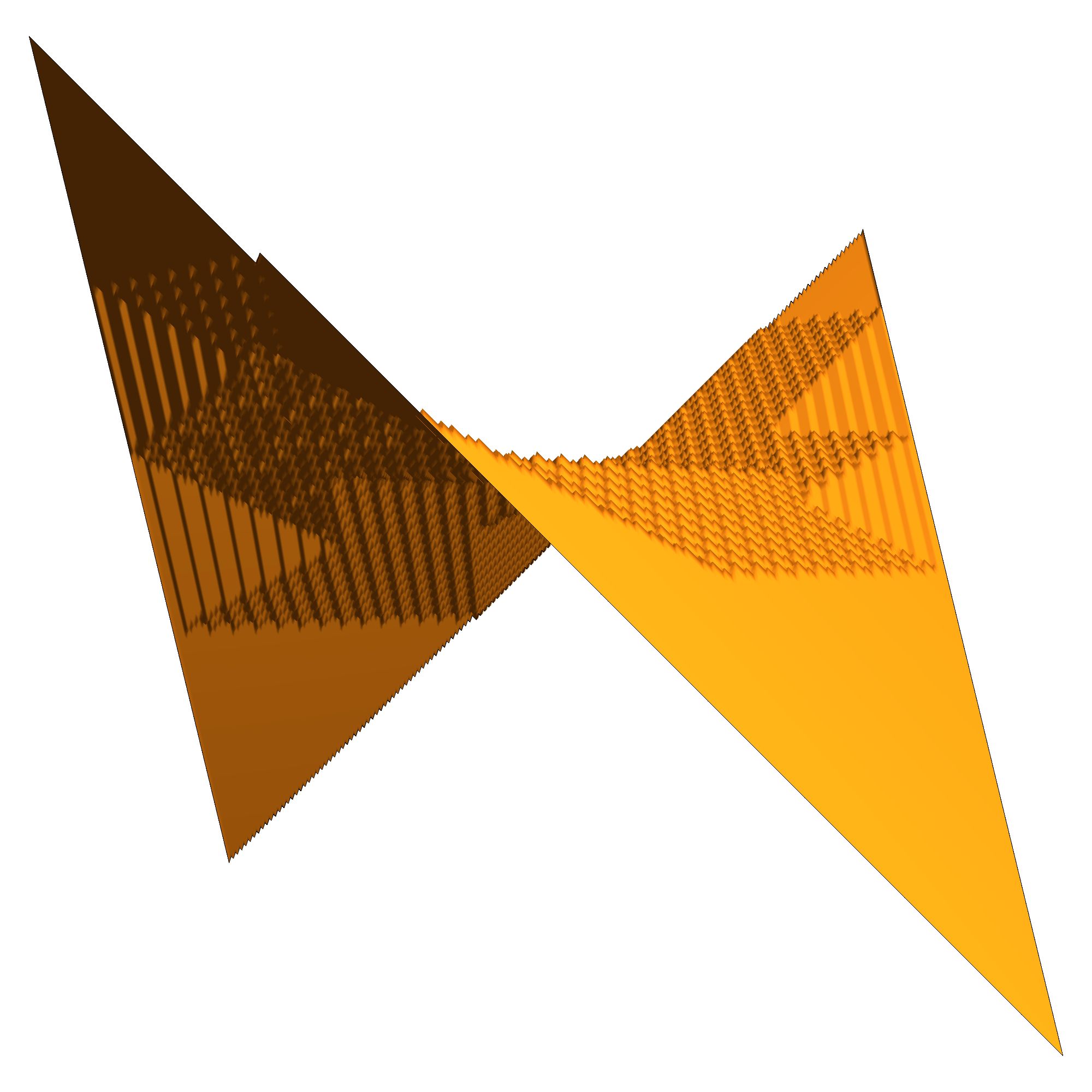}
\end{subfigure}
\quad
  \begin{subfigure}[c]{0.22\textwidth}
\includegraphics[scale=.05]{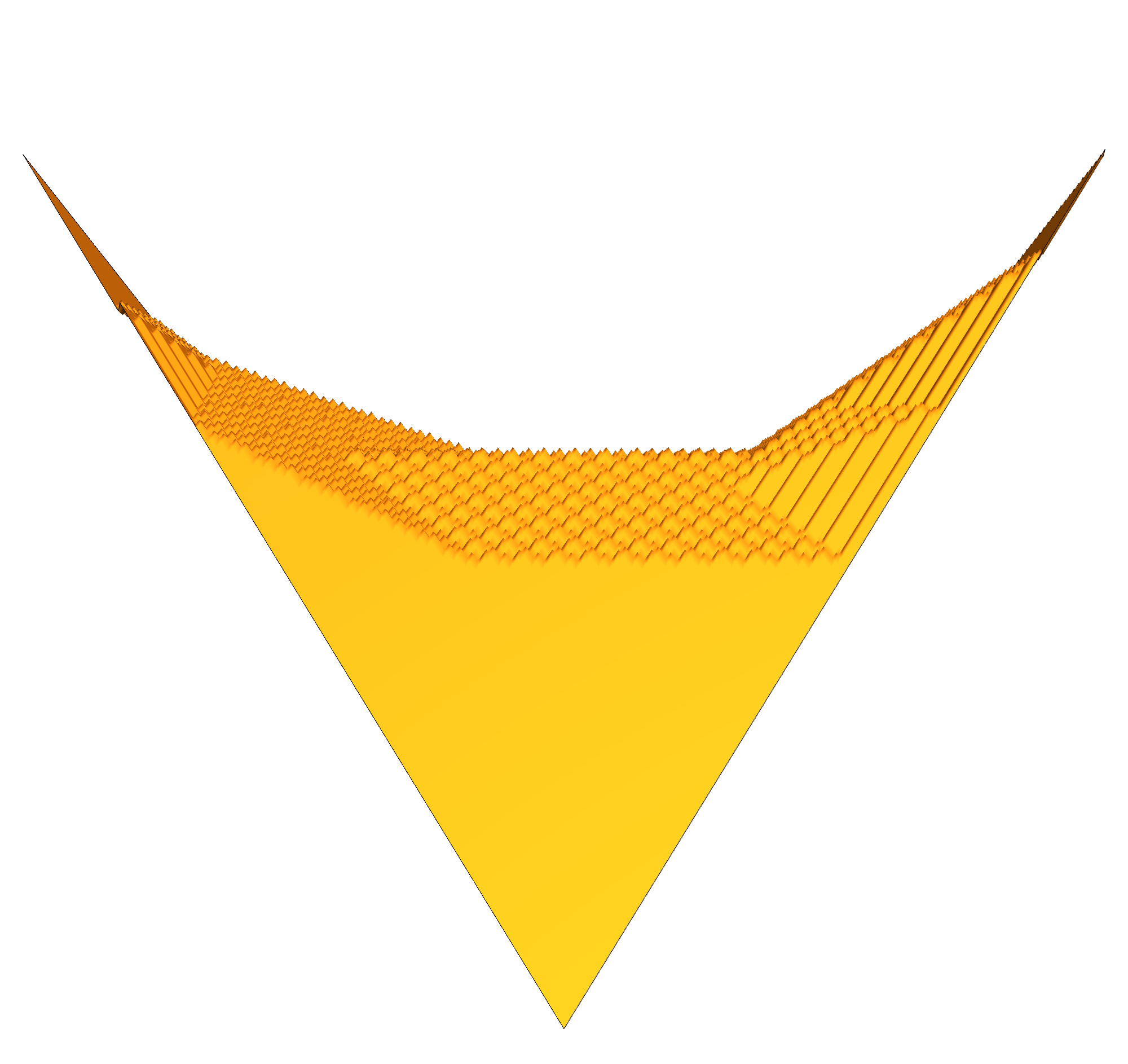}
\end{subfigure}
\quad
  \begin{subfigure}[c]{0.22\textwidth}
\includegraphics[scale=.05]{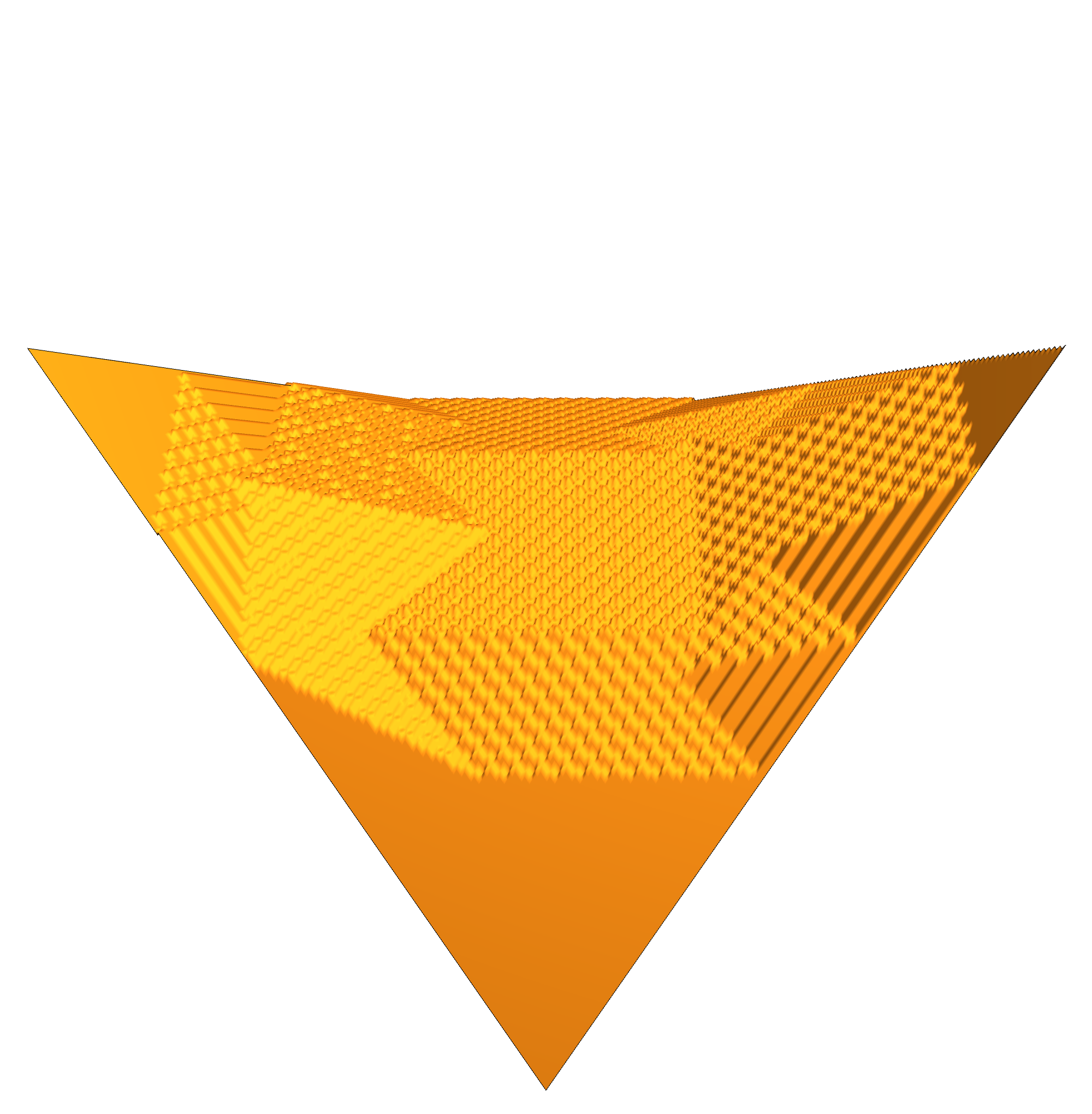}
\end{subfigure}
 \end{center}
\caption{The height function associated with a dimer cover of the Aztec diamond from four different perspectives. The weight function is the same as in Figure~\ref{fig:varying_temp} and~$\beta=40$.
\label{fig:height_function_orange}}
\end{figure}

More concretely, for finite values of~$\beta$, the \emph{limit shape}~$\bar h_\beta$, which is the large-size (deterministic) limit of the random height function associated with dimer covers distributed according to~\eqref{eq:intro:model_beta}, was explicitly described in~\cite{BB23} as follows. It was given in terms of the action function~$F_\beta$, or rather its differential~$\d F_\beta$ on the spectral curve. For~$(u,v)\in D_\text{Az}$, where~$D_\text{Az}$ is the limiting rectangular region of appropriately scaled Aztec diamonds in~$\mathbb{R}^2$, we proved that 
\begin{equation}\label{eq:intro:limit_shape}
\bar h_\beta(u,v)=\frac{1}{k\ell}\frac{1}{2\pi \i}\int_{\gamma_{u,v}}\d F_\beta+1,
\end{equation}  
where~$\gamma_{u,v}$ is a symmetric (with respect to conjugation) simple contour in the spectral curve, depending on the point~$(u,v)\in D_\text{Az}$. See Section~\ref{sec:finite_beta_results} and Figure~\ref{fig:curves_height_function} below for more details.

Following the convention of~\cite{BB23}, in our situation the \emph{Newton polygon}~$N(P)\subset \mathbb{R}^2$ is the rectangle with vertices~$(0,0)$, $(-\ell,0)$,~$(-\ell,k)$ and~$(0,k)$. Set~$\mathcal N=N(P)\cap \ZZ^2$. Then,~$\nabla \bar h_\beta(u,v)+(0,k)\in N(P)$ for all~$(u,v)\in D_\text{Az}$, and~$\bar h_\beta(u,v)+(0,k)\in \mathcal N$ if and only if~$(u,v)$ is in a smooth or frozen region; these types of regions correspond to the facets of the limit shape. For~$\mu\in \mathcal N$, we denote the region consisting of the points~$(u,v)\in D_\text{Az}$ with~$\nabla \bar h_t(u,v)+(0,k)=\mu$ by~$R_{\beta,\mu}$.

We are interested in the crystallization of the Aztec diamond that arises as the temperature tends to zero (or~$\beta\to\infty$), see the right-most image in Figure~\ref{fig:varying_temp} and Figure~\ref{fig:height_function_orange}. As~$\beta\to \infty$, the spectral curve, or rather its \emph{amoeba}~$\mathcal A_\beta$, tends to a \emph{tropical curve} that we denote by~$\mathcal A_t$. The tropical curve can be viewed as a graph embedded into~$\RR^2$ so that its edges are represented by line segments of rational slopes, and we discuss its structure in Section~\ref{sec:intro:tropical_arctic} below. See Figure~\ref{fig:amoeba_tropical_tentacles} for an example of the amoeba~$\mathcal A_\beta$ and the limiting tropical curve~$\mathcal A_t$. The differential~$\d F_\beta$ is an \emph{imaginary normalized differential}, cf. Krichever~\cite{Kri14}, and the tropical limit of such differentials were obtained by Lang~\cite{Lan20}, see also Grushevsky-Krichever-Norton~\cite{GKN19}. We denote its~$\beta\to\infty$ limit as~$\d F_t$ and call the underlying function~$F_t$ on the tropical curve the \emph{tropical action function}. This function also depends on~$(u,v)$ as parameters, and for each~$(u,v)\in D_\text{Az}$ its differential is a~$1$-form on~$\mathcal A_t$ that looks as follows:
\begin{equation}\label{eq:intro:action_function_tropical}
\d F_t(\,\cdot\,;u,v)=k(1+\ell u)\d y(\,\cdot\,)-\ell(1+kv)\d x(\,\cdot\,)-\d f_t(\,\cdot\,),
\end{equation}
where~$\d f_t$ is a~$1$-form on~$\mathcal A_t$ that does not depend on~$(u,v)$. This~$1$-form can be characterized as the unique exact~$1$-form on~$\mathcal A_t$ with (fairly simple) boundary conditions determined from the corresponding periods of the pre-limit differentials~$\d F_\beta$. See Section~\ref{sec:tropical_action} below for details.

 \begin{figure}[t]
 \begin{center}
     \begin{tikzpicture}[scale=1]
    \draw (0,0) node {\includegraphics[angle=180, scale=.4]{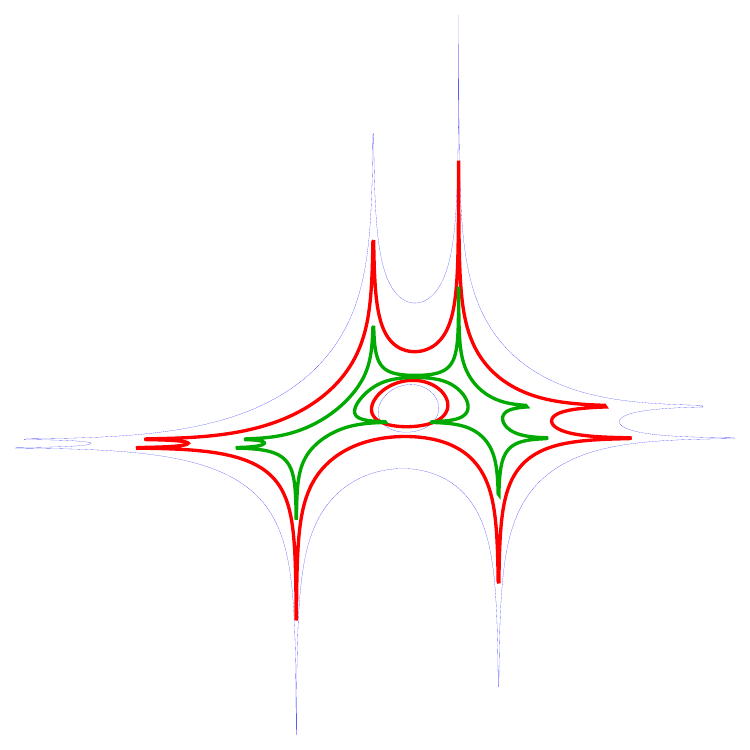}};
    \draw (-.138,.26) node {\includegraphics[scale=.285]{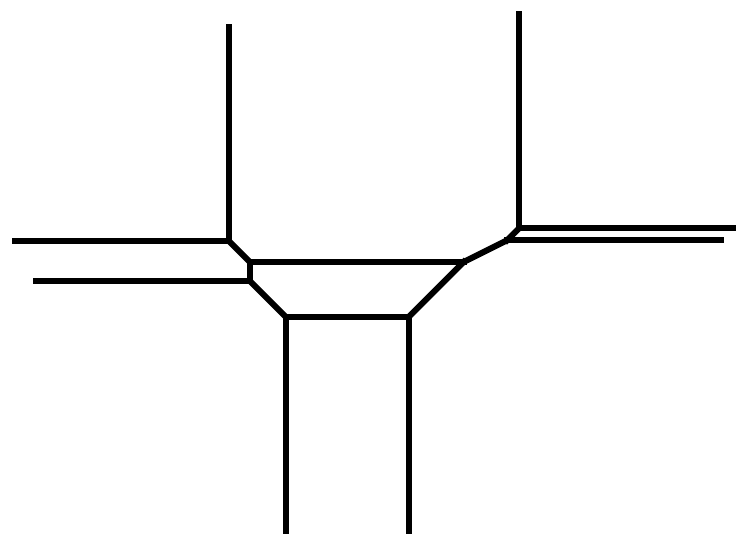}};
  \end{tikzpicture}
 \end{center}
\caption{The amoebas~$\mathcal A_\beta$, with~$\beta=1$,~$\beta=1.5$, and~$\beta=3$, with the corresponding tropical curve~$\mathcal A_t$ on top. The asymptotes of the tentacles in the amoeba are the unbounded parts of the tropical curve. 
\label{fig:amoeba_tropical_tentacles}}
\end{figure}

Given~$\mu\in \mathcal N$, we define~$R_\mu\subset D_\text{Az}$ as the interior of the set of points~$(u,v)$ for which there exists~$\beta_0=\beta_0(u,v)$ such that~$(u,v)\in R_{\beta,\mu}$ for all~$\beta>\beta_0$. In Definition~\ref{def:regions} below, the sets~$R_\mu$ are instead defined through zeros of~$\d F_t$, and the equivalence of these definitions follows from Theorem~\ref{thm:tropical_limit_arctic_curve}. For~$(u,v)\in R_\mu$, we define, cf.~\eqref{eq:intro:limit_shape},
\begin{equation}\label{eq:intro:tropical_limit_shape}
\bar h_t(u,v)=\frac{1}{k\ell}\sum_{e\in \Gamma_\mu}\d F_t(\eta(e);u,v)+1,
\end{equation}
where~$\Gamma_\mu$ is a subset of the edges of~$\mathcal A_t$ described in Definition~\ref{def:discrete_curve}, see also Figure~\ref{fig:discrete_curve}, and~$\eta(e)$ is the primitive (with coprime coordinates) integer vector parallel to~$e$.

Before stating our first result, proved in Corollary~\ref{cor:tropical_limit_shape} below, we recall that the tropical curve~$\mathcal A_t$ is said to be \emph{smooth} if~$\mathcal A_t$ has~$2k\ell$ vertices, all of degree three (see Definition~\ref{def:t-curve}). This property is discussed in detail in Section~\ref{sec:subdivision}, where we prove that in the~$k\ell$-dimensional space of all possible edge weights~$\log \nu (\,\cdot\,)$,~$\mathcal A_t$ is a smooth tropical curve outside of a (subset of a) finite number of hyperplanes. 
\begin{theorem}\label{thm:intro:tropical_limit_shape}
If~$\mathcal A_t$ is a smooth tropical curve, and~$(u,v)\in R_\mu$ for some~$\mu\in \mathcal N$, then
\begin{equation}
\lim_{\beta\to\infty} \bar h_\beta(u,v)=\bar h_t(u,v).
\end{equation}
\end{theorem}
If~$\mathcal A_t$ is a smooth tropical curve, then~$\cup_{\mu\in \mathcal N} \overline{R_\mu}=D_\text{Az}$, which means that~\eqref{eq:intro:tropical_limit_shape} defines a continuous piecewise-linear function on all of~$D_\text{Az}$. Furthermore, the \emph{tropical arctic curve}, the union of polygonal boundaries of~$R_\mu$ over all~$\mu\in \mathcal N$, can be described in terms of the tropical action function~$F_t$, as we are about to describe. 

Note that the zero-temperature limit cannot produce a piecewise linear shape for all possible values of edge weights as, for example, in the uniform case of all weights equal to 1 the model is simply independent of~$\beta$. It would be very interesting to understand what can happen to the limit shape in the~$\beta\to\infty$ limit if~$\mathcal A_t$ is not a smooth tropical curve, and we leave this question for future research. 

In~\cite{BB23}, the rough region and the arctic curve came with a homeomorphism between the closure of the rough region and the amoeba. In the zero-temperature limit, the rough region disappears, as does the interior of the amoeba, so such a map does not make literal sense. Instead, the map~$\map_t$ defined below plays a similar role in the description of the tropical arctic curve. 

The tropical action function~$F_t$ is a continuous piecewise linear function on~$\mathcal A_t$ such that, for each vertex~$(x,y)\in \mathcal A_t\subset \RR^2$, there exists a plane~$\Pi=\Pi((x,y);u,v)$ in~$\mathbb{R}^3$ that contains the graph of~$F_t$ in a neighborhood of~$(x,y)$. Let us denote the vertices of~$\mathcal A_t$ by~$V(\mathcal A_t)$. For each vertex~$\mathrm v\in V(\mathcal A_t)$, there is a unique~$(u,v)\in \overline{D_\text{Az}}$ such that the plane~$\Pi(\mathrm v;u,v)$ is parallel to the~$xy$-plane. This defines a map~$\mathrm v\mapsto (u,v)$ from~$V(\mathcal A_t)$ to~$\overline{D_\text{Az}}$ (the map evaluates the gradient of the action function, somewhat similarly to the Legendre transform, yet our function is nonconvex), and its image consists of the vertices of the tropical arctic curve. 

This map is naturally expressed through the function~$f_t$ instead of~$F_t$, cf.~\eqref{eq:intro:action_function_tropical}. 
Indeed, for each vertex~$\mathrm v\in V(\mathcal A_t)$, let the real numbers~$\d_x f_t(\mathrm v)$ and~$\d_y f_t(\mathrm v)$ be such that
\begin{equation}
\d f_t(\eta(e))=(\d_x f_t(\mathrm v),\d_y f_t(\mathrm v))\cdot \eta(e),
\end{equation}
for each edge~$e$ adjacent to~$\mathrm v$, where~$\eta(e)$ is a primitive vector parallel to~$e$, and the dot in the right hand side stands for the dot product. The existence of these (uniquely defined) `partial derivatives' is a consequence of the balance condition~\eqref{eq:balancing_differential} satisfied by~$\d f_t$ and the fact that~$\mathcal A_t$ is a smooth tropical curve. The following theorem completely describes the tropical arctic curve; it is a combination of Proposition~\ref{prop:vertex_map_well_defined} and Theorem~\ref{thm:arctic_curve_tropical} below.
\begin{theorem}\label{thm:intro:tropical_arctic_curve}
Assume that~$\mathcal A_t$ is a smooth tropical curve. The image of the map~$\map_t:V(\mathcal A_t)\to \overline{D_\text{Az}}$ defined by
\begin{equation}
\map_t(\mathrm v)=\frac{1}{k\ell}(\d_y f_t(\mathrm v),-\d_x f_t(\mathrm v))-\frac{1}{k\ell}(k,\ell)
\end{equation}
is the set of vertices of the tropical arctic curve. Furthermore, the tropical arctic curve itself is the union over all pairs of adjacent vertices~$\mathrm v, \mathrm v'\in V(\mathcal A_t)$ of the line segments between~$\map_t(\mathrm v)$ and~$\map_t(\mathrm v')$.
\end{theorem}
A dual expression for~$\map_t$ is discussed in Section~\ref{sec:intro:tropical_arctic} below.

We now turn our attention to local fluctuations. In~\cite{BB23}, it was proved that for fixed~$\beta$, the local fluctuations of~\eqref{eq:intro:model_beta} in a neighborhood of~$(u,v)\in D_\text{Az}$ converge to the (unique) ergodic translation-invariant Gibbs measure with slope~$\nabla \bar h_\beta(u,v)$. As one might expect, the slopes of these local Gibbs measures line up with those of the tropical limit shape for sufficiently large~$\beta$. This is the content of the following theorem, which is a reformulation of Corollary~\ref{cor:local_limit} in the text.
\begin{theorem}\label{cor:intro:local_limit}
Assume that~$\mathcal A_t$ is a smooth tropical curve and let~$(u,v)\in R_\mu\subset D_\text{Az}$ for some slope~$\mu \in \mathcal N$. Then there exists~$\beta_0=\beta_0(u,v)$ such that for any~$\beta>\beta_0$ there is a (macroscopic) neighborhood of~$(u,v)$ in which at the lattice scale~$\PP_{\text{Az},\beta}$ converges weakly, as the size of the Aztec diamond tends to infinity, to the ergodic translation-invariant Gibbs measure with slope~$\mu$.
\end{theorem}

In other words, this theorem says that for any point of the Aztec diamond that ends up inside a facet of a certain slope in the zero-temperature limit, the local fluctuations of the pre-limit dimer models around that point will be described by the Gibbs measures of the same slope for any fixed positive temperature, as long as that temperature is sufficiently small. It is then natural to look into the zero-temperature limit of the Gibbs measures.

Given the exponential behavior of ratios of the edge weights as~$\beta\to \infty$, it seems plausible that the Gibbs measure would concentrate on a single dimer cover, thus eliminating local fluctuations altogether. While this is indeed true generically, there are cases when this does not happen, even though the Aztec limit shape may still tend to a piecewise linear limit, cf. Remark~\ref{rem:multiple_maximizers} and Figure~\ref{fig:multiple_maximizers} below. 

Let us introduce a few new objects. Let~$G_1$ be the \emph{fundamental domain} of our dimer model, which is the smallest non-repeating part of~$G_\text{Az}$ wrapped into a two-dimensional (discrete) torus. Each dimer cover of~$G_1$ has an associated \emph{slope}~$\mu=(\mu_1,\mu_2)\in \mathcal N$ formally defined in~\eqref{eq:energy_slope}; it corresponds to the slope of the height function of the (periodic) lift of this cover to~$\mathbb{Z}^2$. Further, define the \emph{tropical surface tension}~$\mathcal E^*:\mathcal N\to \RR_{>0}$ by
\begin{equation}\label{eq:intro:surface_tension}
\mathcal E^*(\mu)=\max_{\mathcal D}\left\{\sum_{e\in \mathcal D}\log \nu(e)\right\},
\end{equation}
where~$\nu(\,\cdot\,)$ came from~\eqref{eq:intro:model_beta}, and the maximum is taken over all dimer covers of~$G_1$ with the given slope~$\mu$. For~$\mu \in \mathcal N$, let~$G_{1,\mu}$ be the unweighted (equivalently, uniformly weighted) subgraph of~$G_1$ consisting of edges that participate in any of the dimer covers~$\mathcal D$ that attain the maximum in~\eqref{eq:intro:surface_tension}. We denote the lift of~$G_{1,\mu}$ to~$\mathbb{Z}^2$ by~$G_\mu$; this is a periodic, possibly disconnected subgraph of the square lattice. For example, if~\eqref{eq:intro:surface_tension} has a unique maximizer, then~$G_{\mu}$ consists of the (disjoint) edges of~$\mathbb{Z}^2$ that project to the edges of that maximizer when the lattice is wrapped around the torus.

Let~$K_{G_\mu}$ be a \emph{Kasteleyn matrix} of the graph~$G_\mu$ and~$K_{G_{1,\mu}}$ be the corresponding \emph{magnetically altered Kasteleyn} matrix of~$G_{1,\mu}$. We define a probability measure~$\PP_\mu$ on dimer covers~$\mathcal D$ of~$G_\mu$ in terms of its edge occupancy probabilities. For any~$p\ge 1$ and edges~$\tilde e_s=\tilde{\mathrm w}_s\tilde{\mathrm b}_s$,~$s=1,\dots,p$, in~$G_\mu$ that join white and black vertices~$\tilde{\mathrm w}_s$ and~$\tilde{\mathrm b}_s$, those probabilities are given by
\begin{equation}\label{eq:intro:gibbs_zero}
\PP_\mu\left[\tilde e_1,\dots,\tilde e_p\in \mathcal D\right]=\left(\prod_{s=1}^p \left(K_{G_\mu}\right)_{\tilde{\mathrm w}_s\tilde{\mathrm b}_s}\right)\det \left(\left(K_{G_\mu}^{-1}\right)_{\tilde{\mathrm b}_s\tilde{\mathrm w}_{s'}}\right)_{1\leq s,s'\leq p},
\end{equation}
where
\begin{equation}\label{eq:intro:inverse_kasteleyn_zero}
\left(K_{G_\mu}^{-1}\right)_{\tilde{\mathrm b}\tilde{\mathrm w}}=\frac{1}{(2\pi\i)^2}\int_{|z|=1}\int_{|w|=1}\left(K_{G_{1,\mu}}(z,w)^{-1}\right)_{\mathrm b\mathrm w}\frac{z^{n'-n}}{w^{m'-m}}\frac{\d w}{w}\frac{\d z}{z},
\end{equation}
$\mathrm b$ and~$\mathrm w$ are the projections of~$\tilde{\mathrm b}$ and~$\tilde{\mathrm w}$ in~$G_{1,\mu}$, and~$(n'-n,m'-m)$ is the displacement between~$\tilde{\mathrm b}$ and~$\tilde{\mathrm w}$ in~$G_\mu\subset \ZZ^2$. We show that such a probability measure~$\PP_\mu$ exists and is unique. 

If all connected components of~$G_\mu$ are finite\footnote{It seems likely that this should always happen when the zero-temperature limit shape is piecewise linear, or at least when the surface tension~\eqref{eq:intro:surface_tension} is strictly concave, but we have not been able to verify that so far.}, then~$\PP_\mu$ is simply the product of the uniform dimer models on the connected components of~$G_\mu$, and ~\eqref{eq:intro:inverse_kasteleyn_zero} consists of copies of the actual inverses of finite Kasteleyn matrices of those connected components. The following statement does not rely on such finiteness, however (Theorem~\ref{thm:gibbs_limit} in the text).
\begin{theorem}\label{thm:intro:gibbs_limit}
Let~$\mu\in\mathcal N$ and assume that~$\mathcal E^*$ is strictly concave at~$\mu$ (see Definition~\ref{def:concave}). Then the ergodic translation-invariant Gibbs measure with slope~$\mu$ converges weakly, as~$\beta\to\infty$, to the probability measure~$\PP_\mu$. 
\end{theorem}
If there is a unique maximizer~$\mathcal D$ of~$G_1$ with slope~$\mu\in \mathcal N$, then we show that~$\mathcal E^*$ is strictly concave (though the converse is not necessarily true). Theorem~\ref{thm:intro:gibbs_limit} then implies that the Gibbs measures of that slope converge to the delta measure on the dimer cover of~$\mathbb{Z}^2$ obtained by lifting~$\mathcal D$. 

More generally, we prove that the surface tension~$\mathcal E^*$ is strictly concave, and thus Theorem~\ref{thm:intro:gibbs_limit} applies, on the complement to a finitely many hyperplanes in the space of edge weights. 

One may hope that the same measures~$\PP_\mu$ should arise as local limits as both~$\beta$ and the size of the Aztec diamond tend to infinity. Such double limit transitions are beyond the scope of the present paper, and we hope to look into this question in the future.

\subsection{A weighted discrete Laplacian and the tropical arctic curve}\label{sec:intro:tropical_arctic}
Our next goal is to provide a precise description of the tropical arctic curve and the tropical limit shape~$\bar h_t$. Along the way we will also give a characterization of the tropical action function~$F_t$, from~\eqref{eq:intro:action_function_tropical}.
Tropical geometry provides a natural language for this task, and we refer to Mikhalkin~\cite{Mik05} for an introduction to the subject and more details on the notions used below.

 \begin{figure}[t]
 \begin{center}
 \begin{subfigure}[c]{0.23\textwidth}
\includegraphics[scale=.29]{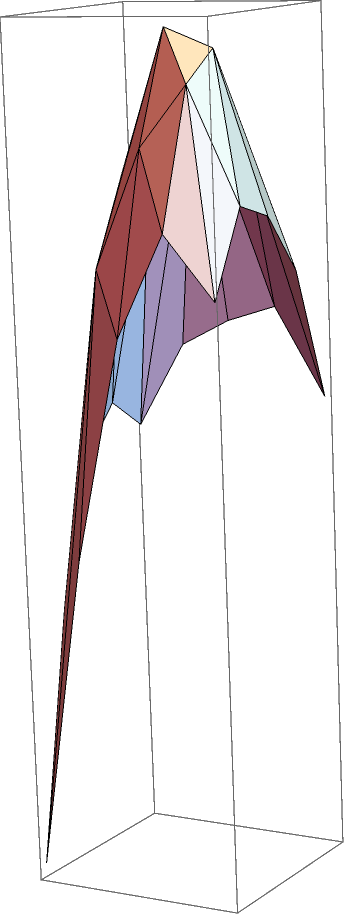}
 \end{subfigure}
\begin{subfigure}[c]{0.23\textwidth}
\includegraphics[scale=.25]{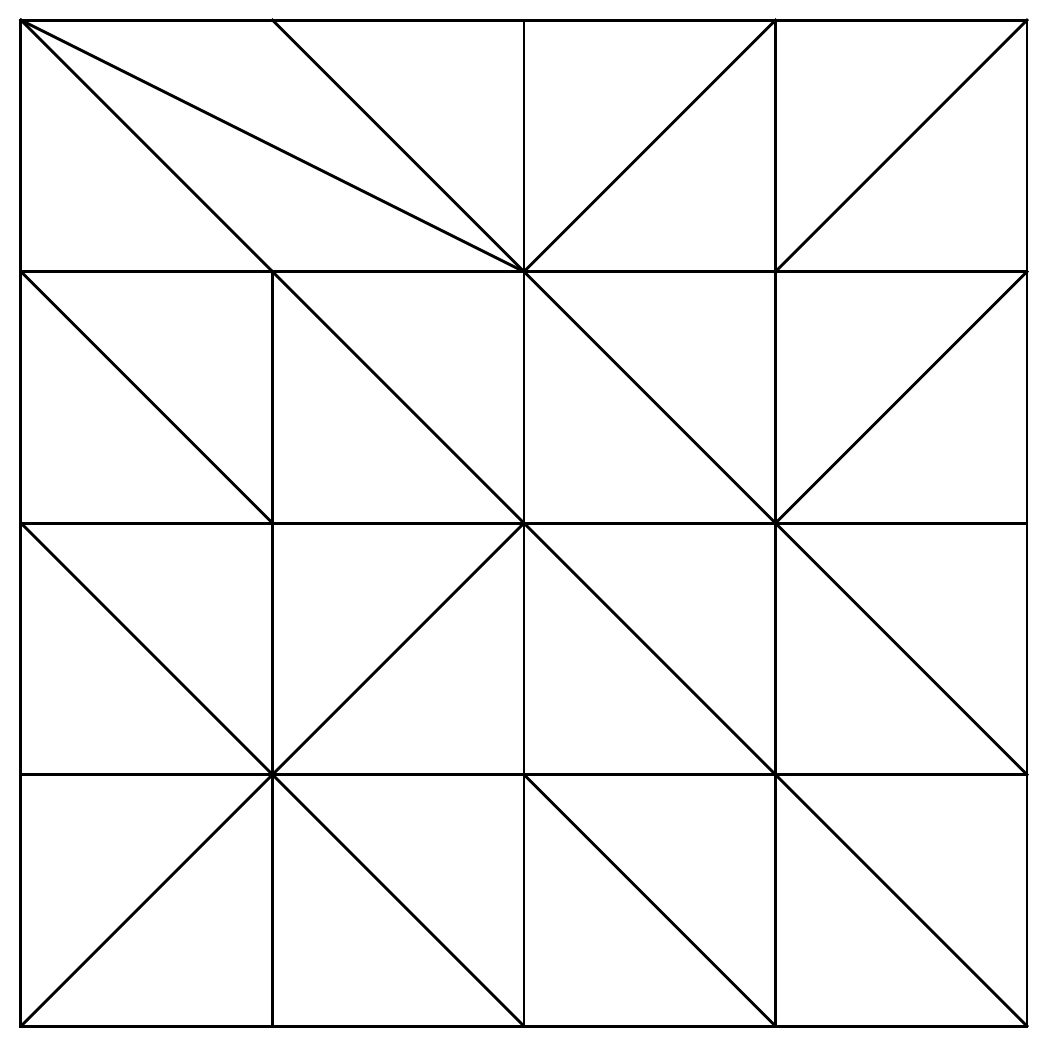}
 \end{subfigure}
 \qquad \qquad
\begin{subfigure}[c]{0.23\textwidth}
\includegraphics[scale=.25, angle=180]{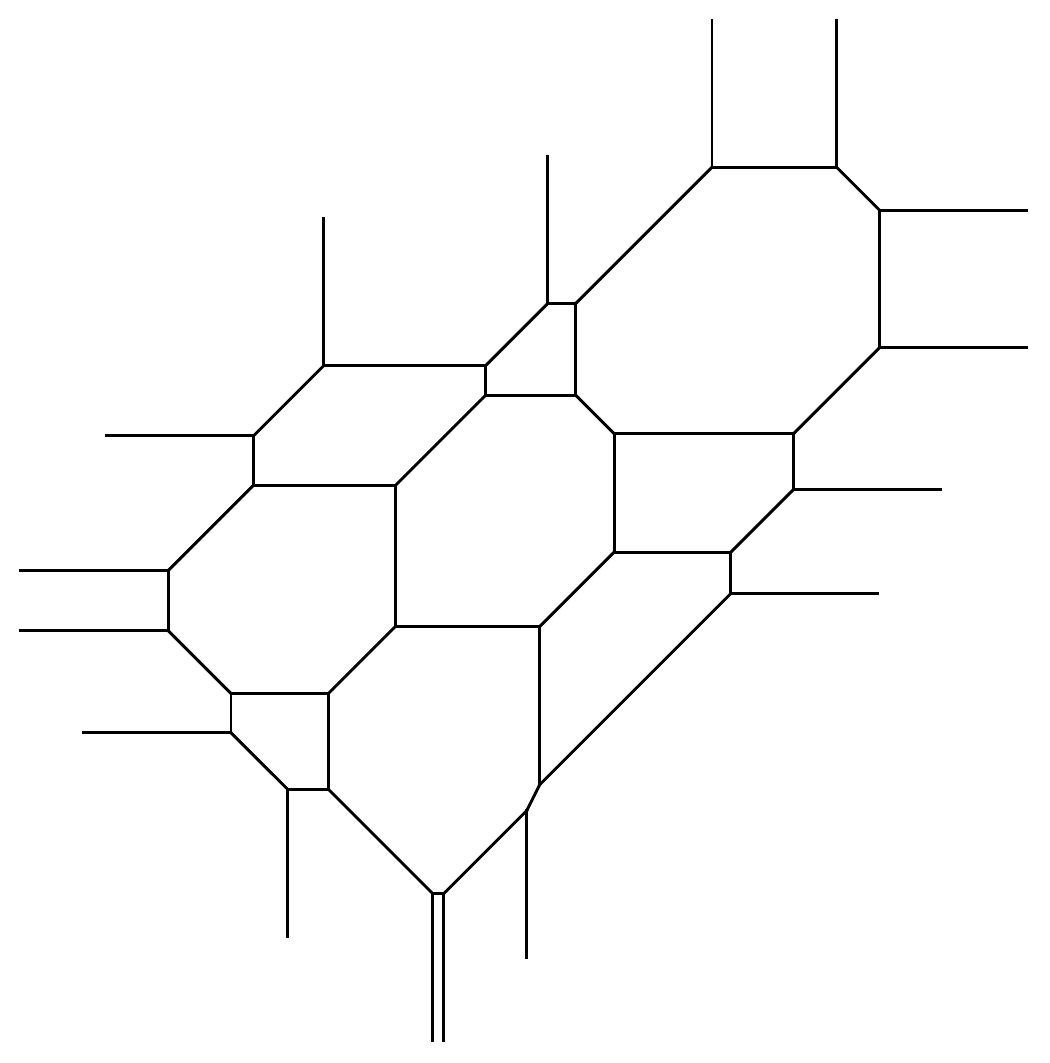}
 \end{subfigure}
 \end{center}
\caption{The extended polyhedral domain~$\tilde N(P_t)$ (left), the subdivision~$N_S(P_t)$ of the Newton polygon (middle) and the corresponding tropical curve~$\mathcal A_t$ (right). The periodicity is~$k=4$ and~$\ell=4$. 
\label{fig:extended_newton_amoeba}}
\end{figure}

Given the tropical surface tension~$\mathcal E^*$ of~\eqref{eq:intro:surface_tension}, we define a \emph{tropical characteristic polynomial} as
\begin{equation}
P_t(x,y)=\max_{\mu=(\mu_1,\mu_2)\in \mathcal N}\left\{\mu_1 x+\mu_2 y+\mathcal E^*(\mu)\right\}.
\end{equation}
The tropical curve~$\mathcal A_t$ is the set of points~$(x,y)\in \RR^2$ where~$P_t$ is not smooth, see the right panel of Figure~\ref{fig:extended_newton_amoeba} for an example.

A somewhat more direct way of accessing~$\mathcal A_t$ from~$\mathcal E^*$ involves the notion of an \emph{extended polyhedral domain} defined as
\begin{equation}\label{eq:intro:extended_polyhedral_domain}
\tilde N(P_t)=\convHull \left\{(\mu,s) \in \RR^3:\mu \in \mathcal N, s\leq \mathcal E^*(\mu)\right\}.
\end{equation}
It naturally projects to the Newton polygon~$N(P)$, and the projection defines a polygonal subdivision~$N_S(P_t)$ of~$N(P)$ by declaring that each bounded face of~$\tilde N(P_t)$ projects to a face in~$N_S(P_t)$, see the left and the middle panels of Figure~\ref{fig:extended_newton_amoeba}. Let us extend the function~$\mathcal E^*$ from~$\mathcal N=N(P)\cap \mathbb{Z}^2$ to all of~$N(P)$ so that the graph of~$\mathcal E^*$ coincides with the top boundary of~$\tilde N(P_t)$. Then the negative gradient~$-\nabla \mathcal E^*$ maps the faces of~$N_S(P_t)$ bijectively to the vertices of~$\mathcal A_t$. Furthermore, two vertices in~$\mathcal A_t$ are connected by a line segment~$e$ if and only if the corresponding faces in~$N_S(P_t)$ are adjacent, and that line segment is orthogonal to the edge~$e^*$ separating the two faces. A vertex in~$\mathcal A_t$ is adjacent to an unbounded edge~$e$ orthogonal to an edge~$e^*$ in~$N_S(P_t)$ if~$e^*$ is a boundary edge. We denote the set of bounded and unbounded edges of~$\mathcal A_t$ by~$E(\mathcal A_t)$ and~$L(\mathcal A_t$), respectively, and the corresponding edges in~$N_S(P_t)$ by~$E(\mathcal A_t)^*$ and~$L(\mathcal A_t)^*$. We further distinguish the sets of unbounded edges (or \emph{leaves}) in~$L(\mathcal A_t)$ that correspond to the left-most, bottom-most, right-most, and top-most edges by~$L_i(\mathcal A_t)$ and~$L_i(\mathcal A_t)^*$,~$i=1,2,3,4$, respectively. That is, leaves from~$L_1(\mathcal A_t)$ point to the left, leaves from~$L_2(\mathcal A_t)$ point downwards, and so on. Section~\ref{sec:tropical_amoeba} below contains a more detailed description. 

The above correspondence or \emph{duality} between~$N_S(P_t)$ and~$\mathcal A_t$ induces a positive function on the edges in~$E(\mathcal A_t)^*$. Indeed, each edge~$e\in E(\mathcal A_t)$ comes with a length~$l(e)$ defined by the relation~$e=l(e)\eta(e)$, where~$\eta(e)$ is the primitive (with coprime coordinates) `tangent' vector parallel to~$e$. We then transfer this function to a function~$l^*$ on~$E(\mathcal A_t)^*$ by setting~$l^*(e^*)=l(e)$ whenever~$e^*\in E(\mathcal A_t)^*$ is dual to~$e\in E(\mathcal A_t)$. 

Given a continuous piecewise linear function~$g_t^*$ on the edges of~$N_S(P_t)$\footnote{That is,~$g_t^*$ is linear on each edge and continuous at each vertex where the edges meet.}, we let~$\d g_t^*(\eta(e^*))\in\mathbb R$ be its differential (or slope), where~$\eta(e^*)$ is the vector parallel to the edge~$e^*$ that is obtained by the~$\pi/2$ counterclockwise rotation from~$\eta(e)$. The action of the \emph{weighted discrete Laplacian}~$\Delta_l$ on ~$g_t^*$ is then defined by
\begin{equation}\label{eq:intro:balancing_dual}
(\Delta_lg_t^*)(\mu)=\sum_{e^*\sim \mu}l^*(e^*)\d g_t^*(\eta(e^*)),
\end{equation}
where~$\mu$ is a vertex in~$N_S(P_t)$, and the sum runs over all edges~$e^*$ that are adjacent to~$\mu$ with~$\eta(e^*)$ oriented away from~$\mu$.  

 \begin{figure}[t]
 \begin{center}
\begin{subfigure}[c]{0.25\textwidth}
\includegraphics[scale=.25, trim={10cm 0 10cm 0}, clip]{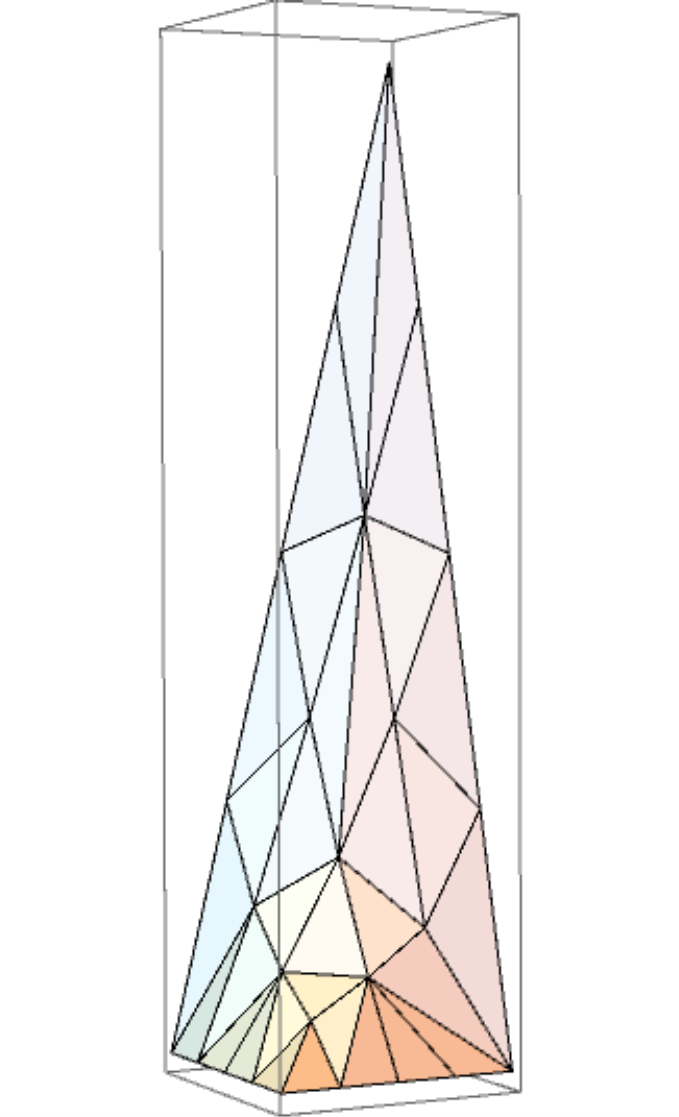}
 \end{subfigure}
\begin{subfigure}[c]{0.25\textwidth}
\includegraphics[scale=.25, angle=180]{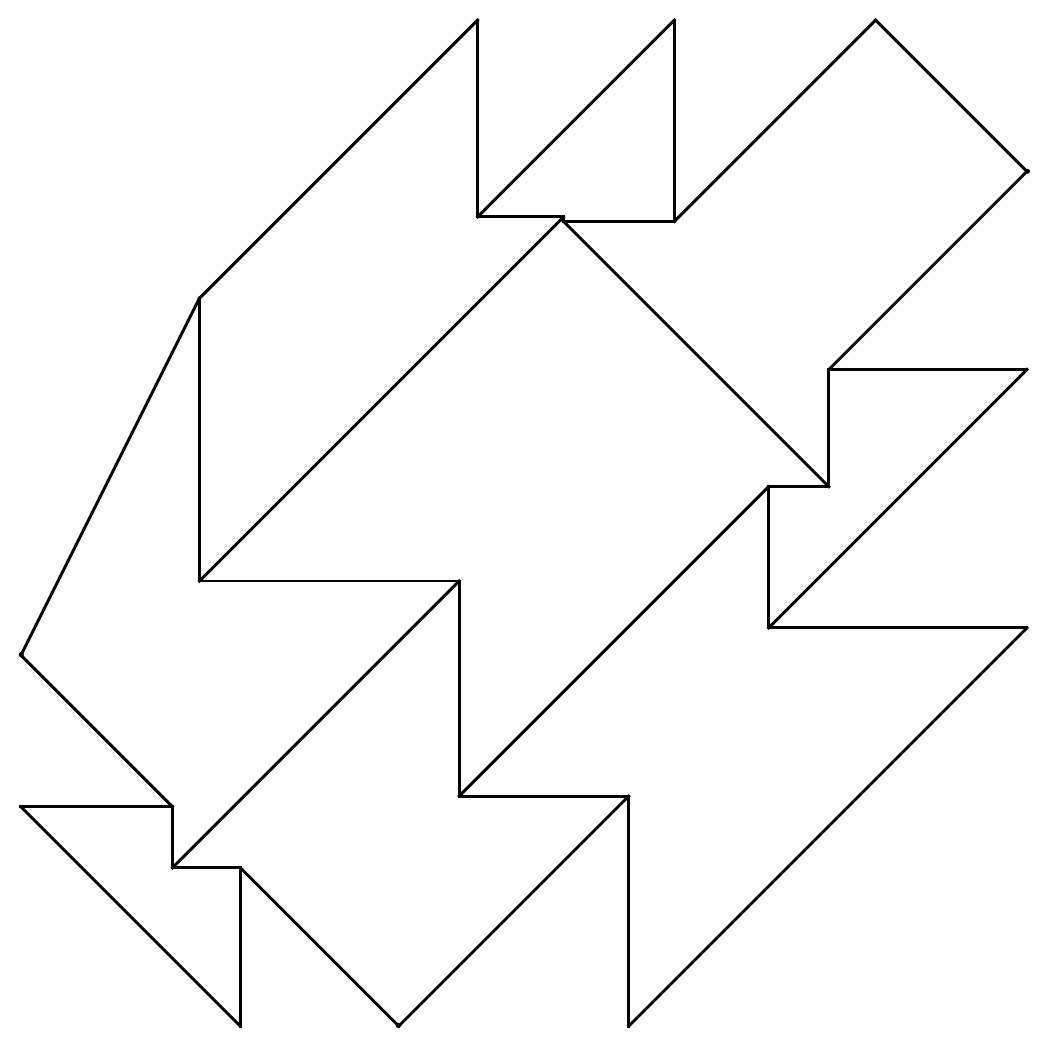}
 \end{subfigure}
 \qquad
 \qquad
\begin{subfigure}[c]{0.25\textwidth}
\includegraphics[scale=.25, angle=180]{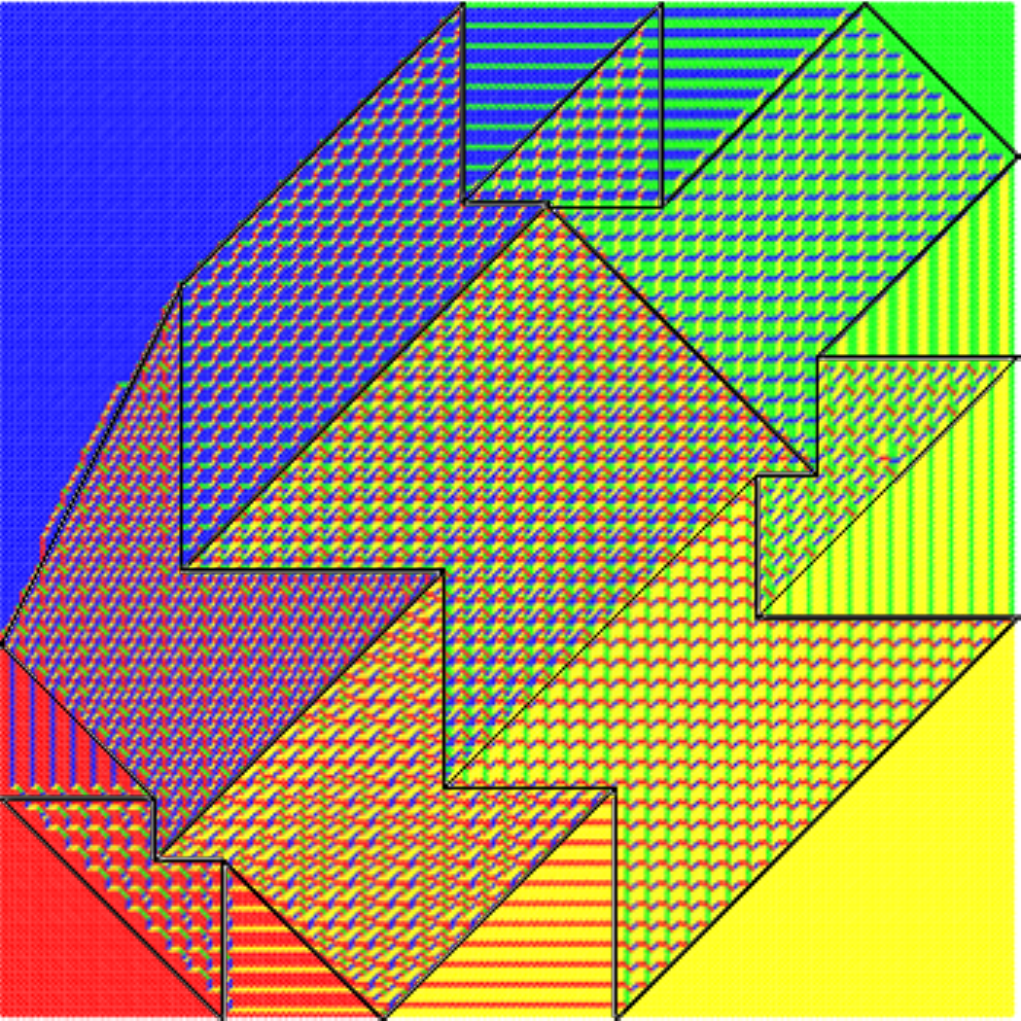}
 \end{subfigure}
 \end{center}
\caption{The graph of the function~$f_t^*$ (left), the arctic curve (middle) and the arctic curve together with a random sampling (right). The weight function is the same as in Figures~\ref{fig:varying_temp} and~\ref{fig:extended_newton_amoeba}. The tropical arctic curve is the union over~$N(P)\backslash \mathcal N$ of the Clarke subdifferential of~$f^*$.
\label{fig:harmonic_arctic}}
\end{figure}

Kirchhoff's theorem implies that there exists a continuous piecewise linear function~$f_t^*$ on the edges of~$N_S(P_t)$, unique up to an additive constant, such that
\begin{equation}\label{eq:intro:kirchhoff}
\begin{cases}
\Delta_l (f_t^*)\equiv 0, & \text{on the inner vertices of $\mathcal N$},\\
\d f_t^*(\eta(e^*))=-\ell, & e^*\in L_1(\mathcal A_t)^*, \\
\d f_t^*(\eta(e^*))=k, & e^*\in L_2(\mathcal A_t)^*, \\
\d f_t(\eta(e^*))=0, & e^*\in L_i(\mathcal A_t)^*, \quad i=3,4.
\end{cases}
\end{equation}
The function~$f_t^*$ is dual to the function~$f_t$ in~\eqref{eq:intro:action_function_tropical}, in the sense that~$\d f_t^*(\eta(e^*))=\d f_t(\eta(e))$, and this equality determines~$f_t$ from~$f_t^*$ and \emph{vise versa} (up to irrelevant additive constants). Thus, the tropical action function~\eqref{eq:intro:action_function_tropical} is also fully determined.

Our running (and generically satisfied) assumption that~$\mathcal A_t$ is a smooth tropical curve implies that the function~$f_t^*$ extends to a continuous piecewise linear function on~$N(P)\subset \RR^2$, see the left panel of Figure~\ref{fig:harmonic_arctic} for an example of its graph. On each face~$\mathrm v^*$ of~$N_S(P_t)$ it is linear; let~$\nabla f_t^*(\mathrm v^*)$ denote its gradient there. For all~$(s,t)\in N(P)$, we further let~$\partial f_t^*(s,t)$ be the Clarke subdifferential of~$f_t^*$, which is the convex hull of all possible limits~$\lim_{(s',t')\to (s,t)}\nabla f_t^*(s',t')$. 

The following statement is a reformulation of Theorem~\ref{thm:intro:tropical_arctic_curve} in terms of~$f_t^*$ instead of~$f_t$, and it is a combination of Proposition~\ref{prop:vertex_map_gradient} and Corollary~\ref{cor:tropical_arctic_gradient} below.
\begin{theorem}\label{thm:intro:tropical_arctic_curve_dual}
Assume that~$\mathcal A_t$ is a smooth tropical curve. The map~$\map_t:V(\mathcal A_t)\to \overline{D_\text{Az}}$ is given by
\begin{equation}
\map_t(\mathrm v)=\frac{1}{k\ell}\nabla f_t^*(\mathrm v^*)-\frac{1}{k\ell}(k,\ell),
\end{equation}  
and the arctic curve is the set
\begin{equation}\label{eq:intro:arctic_curve}
\bigcup_{(s,t)\in N(P)\backslash \mathcal N}\left(\frac{1}{k\ell} \partial f_t^*(s,t)-\frac{1}{k\ell}(k,\ell)\right).
\end{equation}  
\end{theorem}
See Figure~\ref{fig:harmonic_arctic} for an example of~$f_t^*$ as well as the corresponding tropical arctic curve. 

We thus see that the tropical arctic curve and limit shape are determined by a solution to a Kirchhoff problem -- a system of linear equations with~$(k+1)(\ell+1)$ unknowns. The corresponding system of equations for the function~$f_t$ can be found at the end of Section~\ref{sec:tropical_functions} and the beginning of Section~\ref{sec:tropical_action} below. See, \emph{e.g.}, Grushevsky-Krichever-Norton~\cite[Section 1]{GKN19} for a discussion on Kirchhoff's problem. 

There is a certain similarity in the relation between the surface tension~$\mathcal E^*$ and the tropical curve~$\mathcal A_t$, and that between the solution~$f_t^*$ to the Kirchhoff problem and the tropical arctic curve, illustrated in Figures~\ref{fig:extended_newton_amoeba} and~\ref{fig:harmonic_arctic}, respectively. In particular, both can be expressed in terms of the Clarke subdifferential, see Remark~\ref{rem:arctic_tropical_curve} for details.

Also, the map~$\map_t$ of Theorems~\ref{thm:intro:tropical_arctic_curve} and~\ref{thm:intro:tropical_arctic_curve_dual} 
between the tropical curve~$\mathcal A_t$ and the tropical arctic curve, although not, generally speaking, a bijection, 
has the property of slope preservation, similarly to its finite-temperature precursor. That is, for any two adjacent vertices~$\mathrm v$ and~$\mathrm v'$ of~$\mathcal A_t$, the line segment between~$\map_t(\mathrm v)$ and~$\map_t(\mathrm v')$ is parallel to the edge connecting~$\mathrm v$ and~$\mathrm v'$.  Consequently, if~$R_\mu\subset D_{\text{Az}}$ is the polygon corresponding to a facet of slope~$\mu$, then the angles in~$R_\mu$ and the angles in the polygon~$\mathcal A_{t,\mu}$ of the tropical curve~$\mathcal A_t$ are related by~$\theta_{\mathrm v}=n_{\mathrm v}\pi-\theta_{\mathrm v}'$, with~$n_{\mathrm v}\in \{1,2\}$ if~$\map_t$ is bijective on the vertices of~$\mathcal A_{t,\mu}$, or by a limiting version of this relation if it is not. See Section~\ref{sec:arctic_curve_properties} below for the details. 

We conclude this section by stating an explicit formula from Corollary~\ref{cor:limit_shape_newton} below, that expresses the tropical limit shape~$\bar h_t$ in terms of the solution to Kirchhoff's problem~$f_t^*$: For~$(u,v)$ in a macroscopic region~$R_\mu$ of the Aztec diamond of slope~$\mu \in \mathcal N$, using the notation~$\mu_0=(0,k)\in \mathcal N$, we have 
\begin{equation}
\bar h_t(u,v)=\left(u+\frac{1}{\ell},v+\frac{1}{k}\right)\cdot (\mu-\mu_0)+\frac{1}{k\ell}(f_t^*(\mu)-f_t^*(\mu_0))+1.
\end{equation}
For an equivalent formula in terms of~$f_t$, see~\eqref{eq:tropical_height_function} and the discussion after Corollary~\ref{cor:tropical_limit_shape}.

\subsection*{Outline of the paper}
The rest of the paper is organized as follows. Section~\ref{sec:background} provides the necessary background, including definitions of the models, associated objects, and previous results in the finite temperature setting (Sections~\ref{sec:graph}-\ref{sec:finite_beta_results}), followed by a brief introduction to tropical geometry (Sections~\ref{sec:tropical_geometery}-\ref{sec:tropical_limit}). Our main results are presented in Section~\ref{sec:tropical_global_results}. The main object in the proofs is the tropical action function introduced in Section~\ref{sec:tropical_action} and, with the tropical action function at our disposal, we define the tropical arctic curve in the same section. We then prove in Section~\ref{sec:tropical_limit_results} that the tropical arctic curve, as well as the tropical limit shape, indeed are the limits of their finite temperature counterparts. In Sections~\ref{sec:arctic_curve_properties} and~\ref{sec:dual_representation}, we offer a detailed description of the tropical arctic curve and a dual representation of the objects discussed earlier in Section~\ref{sec:tropical_global_results}. Moving beyond the Aztec diamond, Section~\ref{sec:zero_gibbs} examines the zero-temperature limit of the translation-invariant ergodic Gibbs measures. Finally, Section~\ref{sec:subdivision} establishes the convexity of the tropical surface tension and verifies that our assumptions are satisfied for generic values of the edge weights.

\subsection*{Acknowledgments}
We are indebted to Grigory Mikhalkin for multiple very helpful conversation on the subject of tropical geometry. We are also very grateful to Lionel Lang for clarifications on his work~\cite{Lan20}, and to Sunil Chhita and Christophe Charlier for sharing their software implementations of the domino-shuffling sampling algorithm. TB was partially supported by the Knut and Alice Wallenberg Foundation grant KAW 2019.0523 and by A.~Borodin's Simons Investigator grant and the research was conducted during TB's affiliation with MIT. AB was partially supported by the NSF grant DMS-1853981, and the Simons Investigator program.



\section{Background}\label{sec:background}

\subsection{Weighted graphs}\label{sec:graph}
We consider the square lattice rotated by~$\pi/4$ with the vertices colored in black and white, depending on parity, and denote the resulting bipartite graph by~$G=(B,W,E)$. We enumerate the black vertices~$B$ by~$\mathrm b_{\ell m+i,k n+j}$ and the white vertices~$W$ by~$\mathrm w_{\ell m+i,k n+j}$, for~$m,n\in \ZZ$,~$i=0,\dots,\ell-1$, and~$j=0,\dots,k-1$. The parameters~$k$ and~$\ell$ are fixed positive integers that encode the periodicity of the model. The enumeration is such that the edges~$E$ are given by
\begin{equation}
\mathrm b_{x,y}\mathrm w_{x,y}, \quad \mathrm b_{x,y+1}\mathrm w_{x,y}, \quad
\mathrm b_{x+1,y+1}\mathrm w_{x,y}, \quad \text{and} \quad \mathrm b_{x+1,y}\mathrm w_{x,y}, \quad \text{with} \quad (x,y)=(\ell m+i,k n+j),
\end{equation}
and are called West, South, East, and North, respectively. See Figure~\ref{fig:magnetic_weights}. We also associate positive edge weights and Kasteleyn signs to~$G$. The edge weights are given by a periodic function~$\nu:E\to \RR_{>0}$. Their~$(k,\ell)$-periodicity is such that 
\begin{equation}
\nu(\mathrm b_{\ell m+i,k n+j}\mathrm w_{\ell m+i,k n+j})=\nu(\mathrm b_{i,j}\mathrm w_{i,j}),
\end{equation}
and similarly for the South, East and North edges. The Kasteleyn sign~$\sigma=\sigma(e)$ is defined as~$-1$ on the North edges and~$1$ otherwise.  
\begin{remark}
In~\cite{BB23} the notation
\begin{multline}
\nu(\mathrm b_{i,j}\mathrm w_{i,j})=\gamma_{j+1,i+1}, \quad \nu(\mathrm b_{i,j+1}\mathrm w_{i,j})=\alpha_{j+1,i+1}, \quad \nu(\mathrm b_{i+1,j+1}\mathrm w_{i,j})=\beta_{j+1,i+1},  \\
\text{and} \quad \nu(\mathrm b_{i+1,j}\mathrm w_{i,j})=\delta_{j+1,i+1},
\end{multline}
was used.
\end{remark}

\begin{figure}
\begin{center}
\begin{tikzpicture}[scale=1.6]
\foreach \x in {0,...,4}
{\foreach \y in {0,...,2}
{\draw (\x,\y) node[circle,draw=black,fill=black,inner sep=2pt]{};
}
}

\foreach \x in {1,...,4}
{\foreach \y in {1,2}
{\draw (\x-1,\y-1)--(\x,\y);
\draw (\x-1,\y)--(\x,\y-1);
\draw (\x-.5,\y-.5) node[circle,draw=black,fill=white,inner sep=2pt]{};
}
\foreach \y in {0}
{\draw (\x-.5,\y-.5)--(\x,\y);
\draw (\x-1,\y)--(\x-.5,\y-.5);
\draw (\x-.5,\y-.5) node[circle,draw=black,fill=white,inner sep=2pt]{};
}
\foreach \y in {3}
{\draw (\x-1,\y-1)--(\x-.5,\y-.5);
\draw (\x-.5,\y-.5)--(\x,\y-1);
\draw (\x-.5,\y-.5) node[circle,draw=black,fill=white,inner sep=2pt]{};
}
}

\foreach \y in {0,1,2}
{\foreach \x in {0,1,2,3}
{\draw (\x,\y) node[right] {$\mathrm{b}_{\x,\y}$};
\draw (\x+.5,\y+.5) node [below] {$\mathrm{w}_{\x,\y}$};
}
}

\foreach \y in {0,1,2}
{\draw (4,\y) node [right] {$\mathrm{b}_{0,\y}$};
}

\foreach \x in {0,1,2,3}
{
\draw (\x+.5,-.5) node[below] {$\mathrm{w}_{\x,2}$};
}

\foreach \x in {1,2,3,4}
{
 \draw (\x+.4-1.1,-.35) node {$z^{-1}$};
 \draw (\x+.85-1.1,-.05) node {$z^{-1}$};
}

{
 \draw (0.9+3.,3-.65) node {$w$};
 \draw (0.9+3.,2-.4) node {$w$};
 \draw (0.9+3.,2-.65) node {$w$};
 \draw (0.9+3.,1-.4) node {$w$};
 \draw (0.9+3.,1-.65) node {$w$};
  \draw (.9+3.,-.4) node {$w$};
}

 \draw [<-](-.25,-.25)--(4.25,-.25);
 \draw (-.27,-.25) node[below] {$\gamma_u$};
  \draw [<-](3.75,-.75)--(3.75,2.75);
 \draw (3.75,-.75) node[right] {$\gamma_v$};
 
 \pic[scale=.2, rotate=-45] at (5.2,2.) {compass};
 \end{tikzpicture}
\end{center}
\caption{The fundamental domain with~$k=3$,~$\ell=4$. \label{fig:magnetic_weights}}
\end{figure}
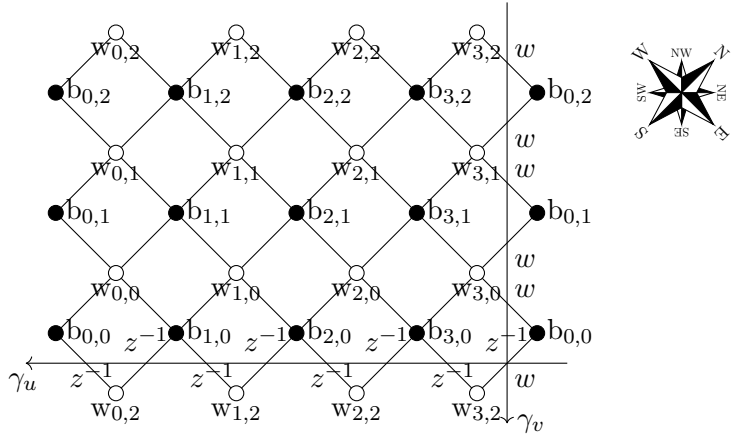

We embed the fundamental domain of the weighted graph~$G$ in the torus by identifying the vertices~$\mathrm b_{\ell m+i,kn+j}$ and~$\mathrm w_{\ell m+i,kn+j}$ with~$\mathrm b_{i,j}$ and~$\mathrm w_{i,j}$, for all~$m,n\in \ZZ$, and we denote the resulting graph on the torus by~$G_1=(B_1,W_1,E_1)$. We further define two oriented simple loops on the torus~$\gamma_u$ and~$\gamma_v$ such that they intersect the edges~$\mathrm b_{i,0}\mathrm w_{i',k-1}$,~$i,i'=0,\dots,\ell-1$, and~$\mathrm b_{0,j'}\mathrm w_{\ell-1,j}$,~$j,j'=0,\dots,k-1$, respectively. The orientation is taken so that~$\gamma_u$ has black vertices to the right and~$\gamma_v$ has black vertices to the left. Following~\cite{KOS06}, we introduce \emph{magnetic fields} by multiplying the edge weights of the edges that intersect~$\gamma_u$ by~$w\in \CC$, and multiplying by~$z^{-1}\in \CC$ the edge weights for edges that intersect~$\gamma_v$. See Figure~\ref{fig:magnetic_weights}.

Finally, the Aztec diamond is the subgraph of~$G$, denoted by~$G_\text{Az}=(B_\text{Az},W_\text{Az}, E_\text{Az})$, consisting of the vertices,~$\mathrm b_{\ell m+i,k n+j}$ with~$\ell m+i=0,\dots,k\ell N$ and~$k n+j=0,\dots,k\ell N-1$,~$\mathrm w_{\ell m+i,k n+j}$ with~$\ell m+i=0,\dots,k\ell N-1$ and~$k n+j=-1,0,\dots,k\ell N-1$. Here~$k\ell N$ is the size of the Aztec diamond which is taken to be divisible by~$k$ and~$\ell$ for simplicity. 

We embed~$G_\text{Az}$ into~$\RR^2$ so that the midpoints of the faces of~$G_\text{Az}$ (including the boundary faces) are 
\begin{multline}\label{eq:aztec_faces}
\{\left(2(\ell m+i),2(k n+j)\right):\ell m+i,k n+j=0,\dots,k\ell N\} \\
\cup\{\left(2(\ell m+i)+1,2(k n+j)+1\right):\ell m+i,k n+j=0,\dots,k\ell N-1\}.
\end{multline}
Moreover, we scale the graph by letting
\begin{equation}
u=u_N=-\frac{m}{k\ell N}, \quad v=v_N=-\frac{n}{k\ell N},
\end{equation}
so that~$(u,v)\in \overline{D_\text{Az}}$ where
\begin{equation}\label{eq:aztec_scaled}
D_\text{Az}=\left(-\frac{1}{\ell},0\right)\times \left(-\frac{1}{k},0\right) \subset \RR^2.
\end{equation}
\begin{remark}
The above definitions of~$u$ and~$v$ are chosen to match the definition of~$(u,v)$ in~\cite[Remark 4.10]{BB23}.
\end{remark}

\subsection{Finite temperature models}\label{sec:measures}
In this paper, we introduce a temperature parameter to two probabilistic systems, the Aztec diamond dimer model and its local limits, the translation-invariant Gibbs measures on~$G$. Let~$\beta>0$ be a large parameter, which we think of as the inverse temperature of the system. 

A dimer cover of the Aztec diamond~$G_\text{Az}$ is a subset of~$E_\text{Az}$, with elements called dimers, so that each vertex in~$G_\text{Az}$ is covered by exactly one dimer. The probability measure is defined on the space of all dimer covers~$\mathcal D_\text{Az}$ of the Aztec diamond of size~$k\ell N$ and is given by
\begin{equation}\label{eq:model_beta}
\PP_{\text{Az},\beta}(\mathcal D_\text{Az})=\frac{1}{Z_\beta}\prod_{e\in \mathcal D_\text{Az}}\nu(e)^\beta =\frac{1}{Z_\beta}\,\e^{\beta \mathcal E(\mathcal D_\text{Az})},
\end{equation}
where~$\mathcal E(\mathcal D_\text{Az})=\sum_{e\in \mathcal D_\text{Az}}\log\nu(e)$ is the \emph{energy} of the dimer cover and~$Z_\beta=\sum_{\mathcal D_\text{Az}}\e^{\beta \mathcal E(\mathcal D_\text{Az})}$ is the \emph{partition function} with the sum running over all dimer covers of the Aztec diamond. The probability measure is a \emph{determinantal point process} and can be expressed in terms of the \emph{Kasteleyn matrix} and its inverse, see~\cite{Ken97, Ken04}. The Kasteleyn matrix~$K_{\text{Az},\beta}:\CC^{B_\text{Az}}\to \CC^{W_\text{Az}}$ is defined by
\begin{equation}\label{eq:kasteleyn_aztec}
\left(K_{\text{Az},\beta}\right)_{\mathrm w\mathrm b}=\one_{\mathrm w\mathrm b\in E_\text{Az}} \left(\nu(\mathrm w\mathrm b)\right)^\beta\sigma(\mathrm w\mathrm b),
\end{equation}
where~$\sigma$ is a \emph{Kasteleyn sign}. We follow~\cite{BB23} and take~$\sigma(e)=-1$ for all North edges~$e=\mathrm w\mathrm b$ and~$\sigma(e)=1$ otherwise.
Given edges~$e_m=\mathrm w_m\mathrm b_m\in E_\text{Az}$,~$m=1,\dots,p$, 
\begin{multline}\label{eq:prob_beta}
\PP_{\text{Az},\beta}\left[e_1,\dots,e_p\in \mathcal D_\text{Az}\right]=\left(\prod_{m=1}^p\left(K_{\text{Az},\beta}\right)_{\mathrm w_{m}\mathrm b_{m}}\right)\det \left(\left(K_{\text{Az},\beta}^{-1}\right)_{\mathrm b_m\mathrm w_{m'}}\right)_{1\leq m,m'\leq p} \\
=\det \left(\left(K_{\text{Az},\beta}\right)_{\mathrm w_{m'}\mathrm b_{m'}}\left(K_{\text{Az},\beta}^{-1}\right)_{\mathrm b_m\mathrm w_{m'}}\right)_{1\leq m,m'\leq p}.
\end{multline}

To describe the translation-invariant Gibbs measures we define~$K_{G,\beta}:\CC^{B_G}\to \CC^{W_G}$ and~$K_{G_1,\beta}:\CC^{B_{G_1}}\to \CC^{W_{G_1}}$ by
\begin{equation}\label{eq:kasteleyn_gibbs}
\left(K_{G,\beta}\right)_{\mathrm w\mathrm b}=\one_{\mathrm w\mathrm b\in E_G} \nu(\mathrm w\mathrm b)^\beta\sigma(\mathrm w\mathrm b) 
\quad \text{and} \quad
\left(K_{G_1,\beta}(z,w)\right)_{\mathrm w\mathrm b}=\one_{\mathrm w\mathrm b\in E_{G_1}} \nu(\mathrm w\mathrm b)^\beta\sigma(\mathrm w\mathrm b)\frac{w^{\mathrm w\mathrm b\wedge \gamma_u}}{z^{\mathrm w\mathrm b\wedge \gamma_v}}, 
\end{equation}
where~$\mathrm w\mathrm b\wedge \gamma$ is~$1$ if~$\mathrm w\mathrm b$ intersect the curve~$\gamma$ and~$0$ otherwise, for~$\gamma=\gamma_u,\gamma_v$. The infinite matrix~$K_{G_\beta}$ does not have a unique inverse. For each point~$(x,y)\in \RR^2$ there is an inverse, which we denote by~$K_{G,\beta,(x,y)}^{-1}$, whose matrix elements are given by
\begin{equation}
\left(K_{G,\beta,(x,y)}^{-1}\right)_{\mathrm b_{\ell m+i,k n+j}\mathrm w_{\ell m'+i',k n'+j'}}=\frac{1}{(2\pi\i)^2}\int_{|z|=\e^{\beta x}}\int_{|w|=\e^{\beta y}}\left(K_{G_1,\beta}(z,w)^{-1}\right)_{\mathrm b_{i,j}\mathrm w_{i',j'}}\frac{z^{n'-n}}{w^{m'-m}}\frac{\d w}{w}\frac{\d z}{z}.
\end{equation}
Note that we have scaled~$(x,y)$ with~$\beta$ in the paths of integration. As~$\beta\to \infty$ this will be the right scaling to consider and it will be consistent with the scaling of the amoeba in Section~\ref{sec:finite_spectral_curve}. The translation-invariant ergodic Gibbs measure~$\PP_{\beta,(x,y)}$, indexed by~$(x,y)$, is a measure on dimer covers~$\mathcal D$ of~$G$, such that, given edges~$e_m=\mathrm w_m\mathrm b_m\in E$,~$m=1,\dots,p$, the joint edge probabilities are given by 
\begin{equation}\label{eq:gibbs_beta}
\PP_{\beta,(x,y)}\left[e_1,\dots,e_p\in \mathcal D\right]=\det \left(\left(K_{G,\beta,(x,y)}\right)_{\mathrm w_{m'}\mathrm b_{m'}}\left(K_{G,\beta,(x,y)}^{-1}\right)_{\mathrm b_m\mathrm w_{m'}}\right)_{1\leq m,m'\leq p}.
\end{equation}
It is more natural to index the Gibbs measures by~$(x,y)\in \mathcal A_\beta$, where~$\mathcal A_\beta$ is the associated amoeba, as defined in Section~\ref{sec:finite_spectral_curve}. Indeed, for~$(x,y)$ and~$(x',y')$ in the same component of the complement of the amoeba, the associated Gibbs measures are equal,~$\PP_{\beta,(x,y)}=\PP_{\beta,(x',y')}$.

\subsection{The height function}\label{sec:height_function}
The \emph{height function} is defined on the faces of~$G_\text{Az}$, or at their midpoints~\eqref{eq:aztec_faces}, as follows: For two faces~$\mathrm f$ and~$\mathrm f'$ in~$G_\text{Az}$, we define the height function~$h$ for a dimer cover~$\mathcal D_\text{Az}$ so that
\begin{equation}\label{eq:height_difference_aztec}
h(\mathrm f')-h(\mathrm f)=\sum_{e=\mathrm w\mathrm b}(\pm)\left(1_{e\in \mathcal D_\text{Az}}-1_{e\in \text{North}}\right),
\end{equation}
where the sum runs over the edges intersecting a path of the dual graph to~$G_\text{Az}$ going from~$f$ to~$f'$, and the sign is~$+$ if the path intersects the edge~$e$ with the white vertex on the right, and~$-1$ if the white vertex is on the left. Here~$\text{North}$ is the set of all North edges in~$E_\text{Az}$. This determines~$h$ up to an additive constant. We fix the constant by setting~$h(0,0)=0$. See Figure~\ref{fig:height_function}. It is not difficult to see that going around a vertex does not change the value of~$h$, which shows that~$h$ is a well-defined function.

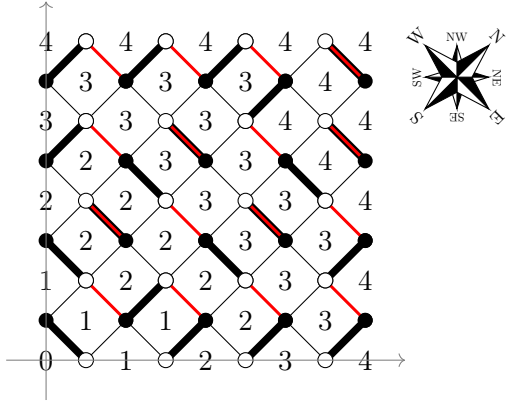
\begin{figure}
\begin{center}
\begin{tikzpicture}[scale=.75, rotate=-45]

\foreach \x in {0,1,2,3}
\foreach \y in {0,...,3}
{\draw (-3.5+\x+\y,.5-\x+\y) rectangle (-2.5+\x+\y,1.5-\x+\y);
}
\foreach \x/\y in {-4/0,-3/1,-2/2,-3/-1,0/-2}
{ 
\draw[line width = 1mm] (\x+.5,\y+.5)--(\x+.5,\y+1.5);
\draw (\x+.5,\y+.5) node[circle,fill,inner sep=2pt]{};
\draw (\x+.5,\y+1.5) node[circle,draw=black,fill=white,inner sep=2pt]{};
}

\foreach \x/\y in {1/-2,2/-1,3/0,2/1,-1/2}
{
\draw[line width = 1mm] (\x+.5,\y+.5)--(\x+.5,\y+1.5);
\draw (\x+.5,\y+.5) node[circle,draw=black,fill=white,inner sep=2pt]{};
\draw (\x+.5,\y+1.5) node[circle,fill,inner sep=2pt]{};
}

\foreach \x/\y in {-1/-3,-2/-2,0/2,0/0,-2/0}
{
\draw[line width = 1mm] (\x+.5,\y+.5)--(\x+1.5,\y+0.5);
\draw (\x+.5,\y+.5) node[circle,fill,inner sep=2pt]{};
\draw (\x+1.5,\y+.5) node[circle,draw=black,fill=white,inner sep=2pt]{};
}

\foreach \x/\y in {-2/1,0/1,0/3,-1/4,-2/-1}
{
\draw[line width = 1mm] (\x+.5,\y+.5)--(\x+1.5,\y+0.5);
}

\foreach \x/\y in {-1/-2,0/-1,1/0,2/1, -2/-1,-1/0,0/1,1/2, -3/0,-2/1,-1/2,0/3, -4/1,-3/2,-2/3,-1/4}
{
\draw[line width = .4mm, red] (\x+.5,\y+.5)--(\x+1.5,\y+0.5);
\draw (\x+.5,\y+.5) node[circle,draw=black,fill=white,inner sep=2pt]{};
\draw (\x+1.5,\y+.5) node[circle,fill,inner sep=2pt]{};
}

\foreach \x in {0,1,2,3,4}
{\draw (\x,\x-3) node {\x};
\draw (-\x,5-\x) node {4};
}
\foreach \y in {1,2,3}
{\draw (-\y,\y-3) node {\y};
\draw (\y,5-\y) node {4};
}
\draw (0,-2) node {1};
\draw (0,-1) node {2};
\draw (0,0) node {2};
\draw (0,1) node {3};
\draw (0,2) node {3};
\draw (0,3) node {4};
\draw (0,4) node {4};

\draw (-1,-1) node {2};
\draw (-1,0) node {2};
\draw (-1,1) node {3};
\draw (-1,2) node {3};
\draw (-1,3) node {3};

\draw (1,-1) node {1};
\draw (1,0) node {2};
\draw (1,1) node {3};
\draw (1,2) node {3};
\draw (1,3) node {4};

\draw (-2,0) node {2};
\draw (-2,1) node {3};
\draw (-2,2) node {3};

\draw (2,0) node {2};
\draw (2,1) node {3};
\draw (2,2) node {3};

\draw (-3,1) node {3};

\draw (3,1) node {3};

 \draw [gray,->](-.5,-3.5)--(4.5,1.5);
 \draw [gray,->](.5,-3.5)--(-4.5,1.5);
 
 \pic[scale=.2, rotate=-45] at (1.6,5.7) {compass};
\end{tikzpicture}
\end{center}
\caption{A dimer cover of the Aztec diamond of size~$4$ together with the values of the height function. The black edges are the dimers in the dimer cover, and the red edges are the reference set~$\text{North}$. The union of the dimers and the reference set form the level curves of the height function. The coordinate axes are the coordinate axes of~$\RR^2$ and show how the Aztec diamond is embedded in~$\RR^2$. \label{fig:height_function}}
\end{figure}

\subsection{The characteristic polynomial and the tropical surface tension}\label{sec:characteristic_polynomial}
Given a dimer cover~$\mathcal D$ of~$G_1$, we define its \emph{energy} and \emph{slope} by 
\begin{equation}\label{eq:energy_slope}
\mathcal E(\mathcal D)=\sum_{e\in \mathcal D} \log \nu(e)\in \RR_{> 0} \quad \text{and} \quad \mu(\mathcal D)=\sum_{e\in \mathcal D}(-e\wedge \gamma_u,e\wedge \gamma_v)\in \ZZ^2,
\end{equation}
where the notation~$e\wedge \gamma$ is as in Section~\ref{sec:graph}. We write~$\mu(\mathcal D)=(\mu_1(\mathcal D),\mu_2(\mathcal D))$.

Following~\cite{KOS06}, we define the \emph{characteristic polynomial} as
\begin{equation}\label{eq:characteristic_polynomial_def}
P_\beta(z,w)=\det K_{G_1}(z,w).
\end{equation}
This is a polynomial in~$w$ and~$z^{-1}$. Using Leibniz's formula for determinants, we can expand~$P_\beta$ as a sum over dimer covers of~$G_1$. Indeed, if we let~$s=s(\mathcal D)\in S_{k\ell}$ be the permutation in the~$k\ell$th symmetric group corresponding to the dimer cover~$\mathcal D$, we get that
\begin{equation}\label{eq:characteristic_polynomial_finite}
P_\beta(z,w)=\sum_{\mathcal D}\sgn(s(\mathcal D))\left(\prod_{e\in \mathcal D}\sigma(e)\right)\e^{\beta \mathcal E(\mathcal D)}z^{\mu_1(\mathcal D)}w^{\mu_2(\mathcal D)}.
\end{equation}
The sign~$\sgn(s(\mathcal D))\prod_{e\in \mathcal D}\sigma(e)$ can be expressed in terms of~$\mu(\mathcal D)$, see~\cite[Proposition 3.1]{KOS06}, but it will not be necessary for our purposes.

Given the characteristic polynomial, we define the Newton polygon~$N(P)$ as the convex hull of all possible slopes of dimer covers of~$G_1$. That is, we set
\begin{equation}
\mathcal N=\left\{\mu=(\mu_1,\mu_2)\in \ZZ^2:a_\mu\neq 0, \,P(z,w)=\sum_{\mu}a_\mu z^{\mu_1}w^{\mu_2}\right\},
\end{equation}
and define~$N(P)$ as the convex hull of~$\mathcal N$. Explicitly,~$N(P)$ is the rectangle with vertices~$(0,0)$, $(-\ell,0)$,~$(-\ell,k)$ and~$(0,k)$. We are particularly interested in a specific subdivision of the Newton polygon which will be discussed later, see Section~\ref{sec:tropical_amoeba}. Moreover, it will be convenient to define the disjoint sets~$\mathcal F$,~$ \mathcal Q$,~$\mathcal S$, where~$\mathcal F$ consists of the corners of~$\mathcal N$,~$\mathcal Q=\partial \mathcal N\backslash \mathcal F$, and~$\mathcal S=\mathcal N^\circ$, see the right image in Figure~\ref{fig:amoeba_tropical_amoeba}. In particular,
\begin{equation}
\mathcal N=\mathcal F\sqcup\mathcal Q\sqcup \mathcal S.
\end{equation}

The subdivision of~$N(P)$ is defined through the maximal energy for each slope~$\mu \in \mathcal N$. We set
\begin{equation}\label{eq:tropical_surface_tension}
\mathcal E^*(\mu)=\max_{\mu(\mathcal D)=\mu}\{\mathcal E(\mathcal D)\},
\end{equation}
where the maximum is taken over all dimer covers of~$G_1$ with slope~$\mu$. We will refer to~$\mathcal E^*$ as the \emph{tropical surface tension}. We justify the name by showing in Appendix~\ref{appendix:surface_tension} that the surface tension for the finite~$\beta$ model tends to~$\mathcal E^*$ as~$\beta \to \infty$.

\subsection{The spectral curve}\label{sec:finite_spectral_curve}
In this section, we recall the definition of the spectral curve and its amoeba. This is a shortened version of~\cite[Section 3.2]{BB23}, and we refer to that section and references therein for more details.

Given the characteristic polynomial~$P_\beta$, the \emph{spectral curve}~$\mathcal R_\beta^\circ$ is the zero set in~$\left(\CC^*\right)^2$, where~$\CC^*=\CC\backslash \{0\}$,
\begin{equation}
\mathcal R_\beta^\circ=\{(z,w)\in \left(\CC^*\right)^2:P_\beta(z,w)=0\}.
\end{equation}
The spectral curve is naturally embedded in the toric variety associated with the set~$\mathcal N$, see, \emph{e.g.},~\cite{MR00}. Indeed, let 
\begin{equation}
\CC T_{\mathcal N}^\circ=\{(z^{0}w^{0}:z^{-1}w^0:\dots:z^{-\ell}w^k):(z,w)\in (\CC^*)^2\} \subset \CC P^{|\mathcal N|-1}, 
\end{equation}
where~$\CC P^{|\mathcal N|-1}$ is the~$(|\mathcal N|-1)$-dimensional projective space over~$\CC$, and let~$\CC T_{\mathcal N}$ be its closure. We embed the group~$(\CC^*)^2$ into~$\CC T_{\mathcal N}$, so that~$\mathcal R_\beta^\circ \subset \CC T_{\mathcal N}$. The closure of~$\mathcal R_\beta^\circ$ in~$\CC T_{\mathcal N}$ is denoted by~$\mathcal R_\beta$. We denote the points at infinity, known as \emph{angles}, by~$L(\mathcal R_\beta)=\mathcal R_\beta\backslash \left(\CC^*\right)^2$. Before we describe this set in more detail, we introduce the amoeba of the curve.

 \begin{figure}[t]
 \begin{center}
 \begin{subfigure}[c]{0.3\textwidth}
     \begin{tikzpicture}[scale=1]
    \draw (0,0) node {\includegraphics[angle=180, scale=.4]{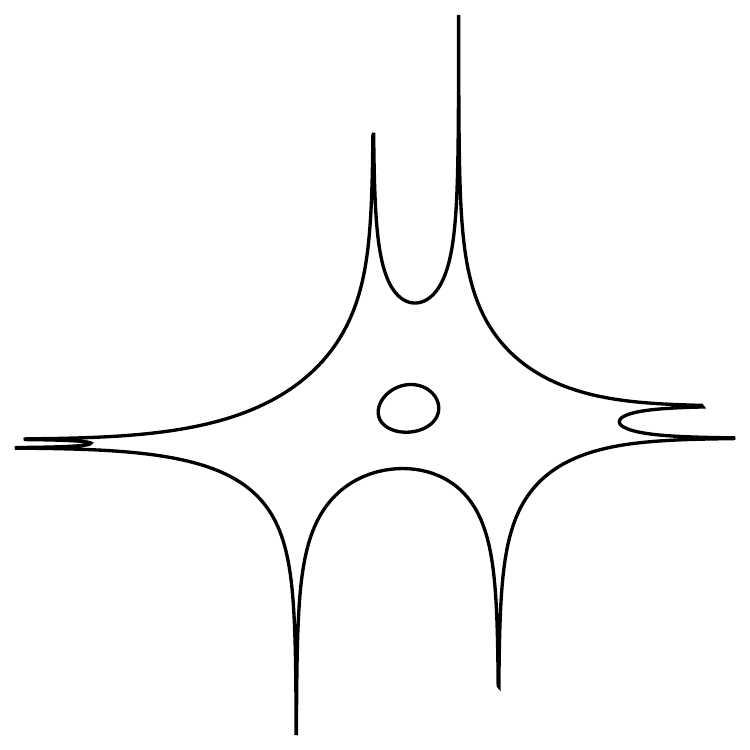}};
    \draw (-.14,.26) node {\includegraphics[scale=.281]{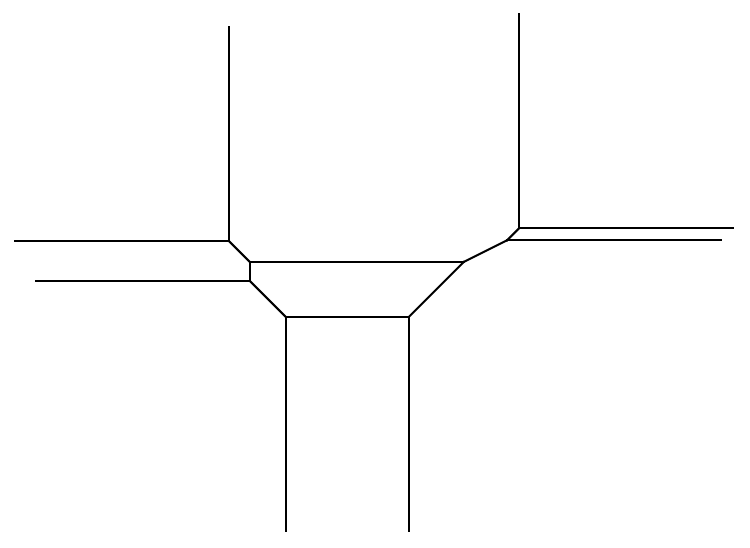}};
	\draw (1.2,1.5) node {\footnotesize{$\mathcal A_{a,\mu_6}$}};
	\draw (-.2,1.5) node {\footnotesize{$\mathcal A_{a,\mu_7}$}};
	\draw (-1.5,1.5) node {\footnotesize{$\mathcal A_{a,\mu_8}$}};
	\draw (2.1,.1) node {\footnotesize{$\mathcal A_{a,\mu_3}$}};
	\draw (-.44,.58) node {\footnotesize{$\mathcal A_{a,\mu_4}$}};
	\draw (-2.2,.0) node {\footnotesize{$\mathcal A_{a,\mu_5}$}};
	\draw (-1.2,-.5) node {\footnotesize{$\mathcal A_{a,\mu_2}$}};
	\draw (.0,-2.0) node {\footnotesize{$\mathcal A_{a,\mu_1}$}};
	\draw (.8,-.5) node {\footnotesize{$\mathcal A_{a,\mu_0}$}};
	\draw (-2.2,.8) node {\footnotesize{$L_1(\mathcal X)$}};
	\draw (-1.3,-2.2) node {\footnotesize{$L_2(\mathcal X)$}};	
	\draw (2.1,.8) node {\footnotesize{$L_3(\mathcal X)$}};
	\draw (1.2,2.2) node {\footnotesize{$L_4(\mathcal X)$}};	
	\draw [-latex, blue, thick] (-.2,.4)--(-.21,.405);
	\draw [-latex, blue, thick] (.75,.85)--(.74,.86);
  \end{tikzpicture}
  \end{subfigure}
  \qquad
  \qquad
  \qquad
\begin{subfigure}[c]{0.2\textwidth}
\begin{tikzpicture}[scale=1]
 \colorlet{Frozen}{orange}
 \colorlet{Qfrozen}{red}
 \colorlet{Smooth}{green}

  \draw (0,0) rectangle (2,2);

    \filldraw[Frozen] (2,0) circle (3pt);
    \draw (2,0) node[above right] {$\mu_0$}; 
    \filldraw[Qfrozen] (1,0) circle (3pt);
    \draw (1,0) node[above right] {$\mu_1$}; 
    \filldraw[Frozen] (0,0) circle (3pt);
    \draw (0,0) node[above right] {$\mu_2$}; 
    \filldraw[Qfrozen] (2,1) circle (3pt);
    \draw (2,1) node[above right] {$\mu_3$}; 
    \filldraw[Smooth] (1,1) circle (3pt);
    \draw (1,1) node[above right] {$\mu_4$}; 
    \filldraw[Qfrozen] (0,1) circle (3pt);
    \draw (0,1) node[above right] {$\mu_5$};     
    \filldraw[Frozen] (2,2) circle (3pt);
    \draw (2,2) node[above right] {$\mu_6$}; 
    \filldraw[Qfrozen] (1,2) circle (3pt);
    \draw (1,2) node[above right] {$\mu_7$}; 
    \filldraw[Frozen] (0,2) circle (3pt);
    \draw (0,2) node[above right] {$\mu_8$}; 

	\draw[Frozen] (3.5,2) node {$\mathcal F$}; 
	\draw[Qfrozen] (3.5,1) node {$\mathcal Q$};    
	\draw[Smooth] (3.5,0) node {$\mathcal S$};
\end{tikzpicture}     
\end{subfigure}
\end{center}
\caption{Left: An example of an amoeba~$\mathcal A_\beta$ and the corresponding tropical curve~$\mathcal A_t$, with~$k=2$ and~$\ell=2$. The subsets~$\mathcal A_{\beta,\mu}$,~$\mathcal A_{t,\mu}$,~$L_i(\mathcal R_\beta)$ and~$L_i(\mathcal A_t)$),~$i=1,\dots,4$, are indicated. In the image,~$a=\beta$ or~$a=t$ and~$\mathcal X=\mathcal R_\beta$ or~$\mathcal X=\mathcal A_t$, depending on whether it corresponds to the amoeba or the tropical curve. 
Right: The Newton polygon~$N(P)$ with the partition~$\mathcal N=\mathcal F\sqcup \mathcal Q\sqcup \mathcal S$ highlighted. In this example,~$\mathcal F=\{\mu_0,\mu_2,\mu_6,\mu_8\}$,~$\mathcal Q=\{\mu_1,\mu_3,\mu_5,\mu_7\}$ and~$\mathcal S=\{\mu_4\}$. 
In both images, the slopes~$\mu\in \mathcal N$ are given by~$\mu_{3j+i}=(-i,j)$ for~$i,j=0,1,2$. 
\label{fig:amoeba_tropical_amoeba}}
\end{figure}
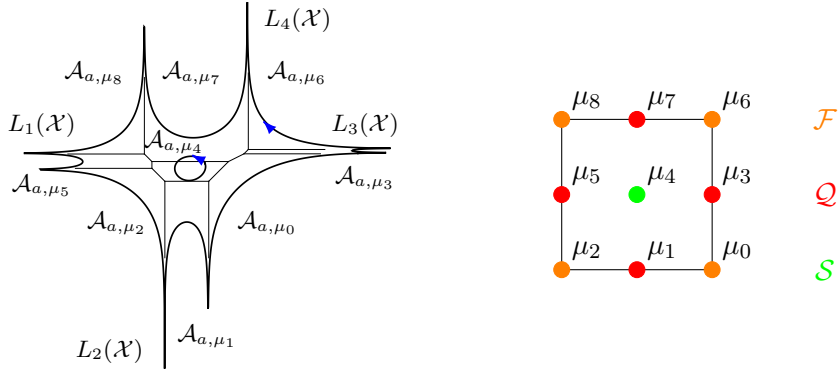

Given the spectral curve, its amoeba is defined as follows. The function~$\Log_\beta:(\CC^*)^2\mapsto \RR^2$ is defined by 
\begin{equation}\label{eq:log_map}
\Log_\beta(z,w)=\beta^{-1}(\log|z|,\log|w|)=(\log_{e^\beta}|z|,\log_{e^\beta}|w|).
\end{equation}
We write~$(x,y)=(\log_{\e^\beta}|z|,\log_{\e^\beta}|w|)$. The \emph{amoeba}~$\mathcal A_\beta$ is the image of~$\mathcal R_\beta^\circ$ under the map~$\Log_\beta$,~$\mathcal A_\beta=\Log(\mathcal R_\beta^\circ)$. The amoeba defined here is scaled by the factor~$\beta^{-1}$ compared with the definition in~\cite{BB23}, which is consistent with the tropical literature. See the left image of Figure~\ref{fig:amoeba_tropical_amoeba} for an example of an amoeba. 

It was proven in~\cite[Theorem 5.1]{KOS06} that~$\mathcal R_\beta^\circ$ is a (possibly singular) \emph{Harnack curve}. The Harnack curves can be characterized in terms of their amoebas. Namely, the curve~$\mathcal R_\beta^\circ$ is a (possibly singular) Harnack curve if and only if the map~$\Log_\beta|_{\mathcal R_\beta^\circ}:\mathcal R_\beta^\circ \to \mathcal A_\beta$ is at most~$2$-to-$1$,~\cite[Theorem 1]{MR00}. Since the polynomial~$P_\beta$ is real, we have~$(z,w)\in \mathcal R_\beta^\circ$ if and only if~$(\bar z,\bar w)\in \mathcal R_\beta^\circ$, implying that the map is~$2$-to-$1$ away from the real part of~$\mathcal R_\beta^\circ$. Moreover, the map is~$1$-to-$1$ on the non-singular part of the real part of~$\mathcal R_\beta^\circ$, which equals~$\mathcal R_\beta^\circ\cap \Log_\beta^{-1}(\partial \mathcal A_\beta)$, see~\cite[Theorem 1 and Corollary 3]{MR00}. In particular, we may think of~$\mathcal R_\beta^\circ$ as the union of two copies of~$\mathcal A_\beta$ joined together along their boundaries.

The fact that~$\mathcal R_\beta^\circ$ is a Harnack curve implies that the real part of~$\mathcal R_\beta$ consists of~$(k-1)(\ell-1)+1$ disjoint topological circles, denoted by~$A_i$,~$i=0,\dots,(k-1)(\ell-1)$ (some of them might be points) where the angles~$L(\mathcal R_\beta)$ is a finite set contained in one of the real components, say~$A_0$. These points are grouped in four disjoint families, which we denote by~$L_i(\mathcal R_\beta)$,~$i=1,\dots,4$, such that when traversing~$A_0$, they appear one group at a time in a specific order. More precisely, the set~$L_i(\mathcal R_\beta)$ corresponds to the points for which~$z=0$ if~$i=1$,~$w=0$ if~$i=2$,~$z=\infty$ if~$i=3$, and~$w=\infty$ if~$i=4$.
\begin{remark}\label{rem:angles}
In~\cite{BB23} the angles where denoted by~$p_{0,j}$,~$q_{0,j}$,~$p_{\infty,j}$ and~$q_{\infty,j}$, and the sets~$L_i(\mathcal R_\beta)$ were such that
\begin{equation}
p_{0,j}\in L_1(\mathcal R_\beta), \quad q_{0,j}\in L_2(\mathcal R_\beta), \quad p_{\infty,j}\in L_3(\mathcal R_\beta), \quad q_{\infty,j}\in L_4(\mathcal R_\beta).
\end{equation}
Moreover
\begin{equation}\label{eq:angles_1}
q_{0,i+1}=\left((-1)^k\prod_{j=0}^{k-1}\frac{\nu(\mathrm b_{i,j+1}\mathrm w_{i,j})^\beta}{\nu(\mathrm b_{i,j}\mathrm w_{i,j})^\beta},0\right), \quad q_{\infty,i+1}=\left(\prod_{j=0}^{k-1}\nu(\mathrm b_{i+1,j+1}\mathrm w_{i,j})^\beta,\infty\right), \quad i=1,\dots,\ell,
\end{equation}
and
\begin{equation}\label{eq:angles_2}
p_{0,j+1}=\left(0,(-1)^\ell\prod_{i=0}^{\ell-1}\frac{\nu(\mathrm b_{i,j+1}\mathrm w_{i,j})^\beta}{\nu(\mathrm b_{i+1,j+1}\mathrm w_{i,j})^\beta}\right), \quad p_{\infty,j+1}=\left(\infty,\prod_{i=0}^{\ell-1}\nu(\mathrm b_{i,j}\mathrm w_{i,j})^\beta\right), \quad j=1,\dots,k.
\end{equation}
See~\cite[(3.3) and (3.4)]{BB23}.
\end{remark}

There is an injective map from the set of components of the complement of~$\mathcal A_\beta$ to~$\mathcal N$ see~\cite{FPT00, KOS06}. This map associates each component to the slope of the \emph{Ronkin function} (see Appendix~\ref{appendix:surface_tension} for a definition), which lies in~$\mathcal N$. Generically this map is a bijection. We denote the boundary of the component of the complement of~$\mathcal A_\beta$ corresponding to~$\mu\in\mathcal N$, via the above-mentioned map, by~$\mathcal A_{\beta,\mu}$. We orient~$\mathcal A_{\beta,\mu}$ in positive direction, seen as a closed loop in~$\RR^2$, if~$\mu$ is in the interior of~$\mathcal N$. For~$\mu\in \partial A$, we orient~$\mathcal A_{\beta,\mu}$ so that~$\cup_{\mu\in \partial A}\mathcal A_{\beta,\mu}$ is positively oriented. See the left image of Figure~\ref{fig:amoeba_tropical_amoeba}.

The real components~$A_i$ of~$\mathcal R_\beta$, that are not points, are mapped under the~$\Log_\beta$ map to~$\mathcal A_{\beta,\mu}$ with~$\mu$ in the interior of~$\mathcal N$ if~$i\neq 0$, and~$A_0\backslash L(\mathcal R_\beta)$ is mapped to the union~$\cup_{\mu\in \partial A} \mathcal A_{\beta,\mu}$. The union~$\cup_{\mu\in \partial A} \mathcal A_{\beta,\mu}$ consists of simple unbounded curves, with so-called tentacles going to infinity, forming the outer boundary of~$\mathcal A_\beta$. Each tentacle corresponds in this way to an angle. Note that the formulas in Remark~\ref{rem:angles} show that~$\Log_\beta(p)$ with~$p\in L(\mathcal R_\beta)$ is independent of~$\beta$. Going forward, we will instead denote the real components of~$\mathcal R_\beta$ by~$A_{\beta,\mu}=\Log_\beta^{-1}(\mathcal A_{\beta,\mu})$. The map from~$\mathcal A_{\beta,\mu}$ to~$A_{\beta,\mu}$ fixes an orientation on~$A_{\beta,\mu}$.

We end this section by recalling the notion of \emph{imaginary normalized differentials}. An imaginary normalized differential~$\omega$ on~$\mathcal R_\beta$ is a meromorphic differential~$1$-form, with at most simple poles at the angles and no other poles, such that
\begin{equation}
\re \left(\int_\gamma \omega\right)=0
\end{equation} 
for all simple closed curves~$\gamma$ in~$\mathcal R_\beta$. For any set of real numbers~$r_p$,~$p\in L(\mathcal R_\beta)$, with~$\sum_{p\in L(\mathcal R_\beta)} r_p=0$, there exists a unique imaginary normalized differential with residues~$r_p$ at the angle~$p\in L(\mathcal R_\beta)$, see, \emph{e.g.},~\cite[Theorem 2.3]{Lan20} or~\cite[Lemma 2.1]{Kri14}. This type of differential was probably known already to Maxwell, as discussed in~\cite{Kri14}.

\subsection{Previous results in the finite temperature regime}\label{sec:finite_beta_results}
In this section, we revisit the relevant objects and results from~\cite{BB23}. The results were obtained under the generic assumption that the component~$\mathcal A_{\beta,\mu}$ is nontrivial for all~$\mu\in\mathcal N$. Furthermore, it was assumed that the asymptotes of the tentacles of the amoeba corresponding to~$L_2(\mathcal R_\beta)$ lie to the right of the line~$x=0$, while those corresponding to~$L_4(\mathcal R_\beta)$ lie to the left. These assumptions were imposed for technical reasons and were recently removed in~\cite{BdT24} and~\cite{BB24}. To maintain consistency with the notation of~\cite{BB23}, we assume throughout this section that a component~$\mathcal A_{\beta,\mu}$ exists for all~$\mu\in\mathcal N$. The assumption that the line~$x=0$ separates the angles, however, can be removed without making any edits in the discussion and will not be assumed below.

A central object used to study the finite temperature regime in~\cite{BB23} is the action function~$F_\beta$. For~$(u,v)\in D_\text{Az}$, the differential of the action function~$\d F_\beta$ is a meromorphic~$1$-form on~$\mathcal R_\beta$ given by
\begin{equation}\label{eq:action_function_beta}
\d F_\beta(z,w;u,v)=k(1+\ell u)\d \log w-\ell(1+kv)\d \log z-\d \log f_\beta(z,w),
\end{equation}
where~$\d \log f_\beta$ is a (unique) imaginary normalized differential with residue~$-\ell$ at all points in~$L_1(\mathcal R_\beta)$, residue~$k$ at all points in~$L_2(\mathcal R_\beta)$, and no poles in~$L_3(\mathcal R_\beta)$ and~$L_4(\mathcal R_\beta)$.
\begin{remark}
In~\cite{BB23} the action function and its differential were defined in terms of prime forms, or theta functions, and it was proven that~$\d F_\beta$ is an imaginary normalized differential. Taking~\eqref{eq:action_function_beta} as its definition here aligns better with the tropical limit, as we will see in Section~\ref{sec:tropical_limit}. 
\end{remark}

The macroscopic regions in the finite temperature regime are defined via the zeros of~$\d F_\beta$. We denote the number of simple zeros of~$\d F_\beta$ in~$A_{\beta,\mu}$ by~$Z_{\beta,\mu}$. For~$(u,v)\in D_\text{Az}$,~$\d F_\beta$ has~$2k\ell$ zeros and 
\begin{equation}\label{eq:zeros_dfbeta}
Z_{\beta,\mu}\in \{0,2\}, \text{ if}\,\, \mu \in \mathcal F,  \quad Z_{\beta,\mu} \in \{1,3\}, \text{ if}\,\, \mu \in \mathcal Q,  \quad Z_{\beta,\mu} \in \{2,4\}, \text{ if}\,\, \mu \in \mathcal S.
\end{equation}
See~\cite[Lemma 4.6]{BB23}. Note that if~$Z_{\beta,\mu}$ attains its minimal possible value for all~$\mu$, then two of the zeros of~$\d F_\beta$ are double zeros or two zeros are not contained in any of the real components. 
\begin{definition}\label{def:finite_region}
Let~$(u,v)\in D_\text{Az}$. We say that~$(u,v)$ is in the frozen phase corresponding to~$\mu\in \mathcal F\cup \mathcal Q$, if~$Z_{\beta,\mu}=2$ and~$\mu\in \mathcal F$ or if~$Z_{\beta,\mu}=3$ and~$\mu \in \mathcal Q$. The point~$(u,v)$ is in the smooth phase corresponding to~$\mu \in \mathcal S$ if~$Z_{\beta,\mu}=4$. If~$(u,v)$ is such that~$\d F_\beta$ has non-real zeros (not contained in the real part of~$\mathcal R_\beta$), we say that~$(u,v)$ is in the rough region. If~$\d F_\beta$ has zeros of higher order, we say that~$(u,v)$ is in the arctic curve. We denote the region in~$D_\text{Az}$ consisting of the points in the phase corresponding to~$\mu$ by~$R_{\beta,\mu}$.
\end{definition}

It was proven in~\cite{BB23} that the closure of the rough region in~$D_\text{Az}$ is homeomorphic to the amoeba~$\mathcal A_\beta$, and the homeomorphism was denoted by~$\Log_\beta \circ \,\Omega_\beta$, see~\cite[Theorem 4.11]{BB23}.
   
The limit shape, the large~$N$ limit of the normalized height function, of the periodically weighted Aztec diamond, that is, under the probability measure~\eqref{eq:model_beta}, was explicitly obtained in~\cite{BB23}. The limit shape~$\bar h_\beta$ of the dimer model can be described in terms of the action function~\eqref{eq:action_function_beta}. Indeed, we set
\begin{equation}\label{eq:limit_shape}
\bar h_\beta(u,v)=\frac{1}{k\ell}\frac{1}{2\pi \i}\int_{\gamma_{u,v}}\d F_\beta+1,
\end{equation}  
where~$\gamma_{u,v}$ is a symmetric (with respect to conjugation) simple curve in~$\mathcal R_\beta$ defined as follows. By the symmetry, it is sufficient to define the curve in the amoeba. The curve in the amoeba is the oriented curve going from~$\Log_\beta\circ \, \Omega_\beta(u,v)$ to any point in the component~$\mathcal A_{\beta,\mu_0}$ with~$\mu_0=(0,k)\in \mathcal F$. In particular, if~$(u,v)\in R_{\beta,\mu}$ for some~$\mu$, then the curve in the amoeba is a curve from~$\mathcal A_{\beta,\mu}$ to~$\mathcal A_{\beta,\mu_0}$, and~$\gamma_{u,v}$ is a simple loop in~$\mathcal R_\beta$ crossing the real part of~$\mathcal R_\beta$ at~$A_{\beta,\mu_0}$ and~$A_{\beta,\mu}$. See Figure~\ref{fig:curves_height_function} for examples of the curve in the amoeba.

 \begin{figure}[t]
 \begin{center}
    \begin{tikzpicture}[scale=1]
  \tikzset{->-/.style={decoration={
  markings, mark=at position .5 with {\arrow{stealth}}},postaction={decorate}}}
    \tikzset{-->-/.style={decoration={
  markings, mark=at position .7 with {\arrow{stealth}}},postaction={decorate}}}
    \draw (0,0) node {\includegraphics[angle=180, scale=.4]{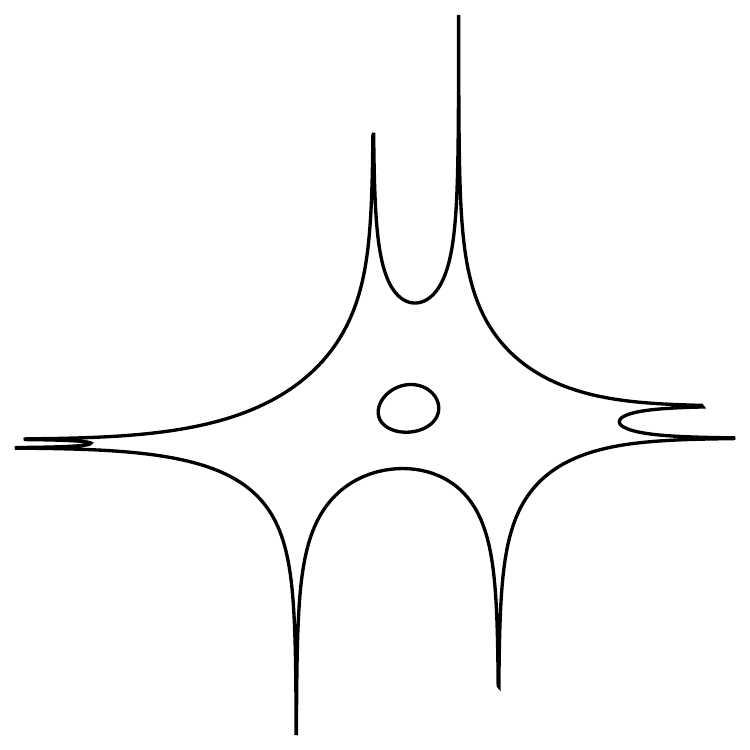}};
	\draw [->-, very thick] (-.2,.4)--(-.21,.405);
	\draw [->-, very thick] (.75,.85)--(.74,.86);
	
	\draw [->-, green, very thick] (-.03,.28)--(.65,1.05);
	\draw [->-, blue, very thick] (.4,.15)--(.68,1.);
	\draw (.4,.15) node[circle,draw=blue,fill=blue,inner sep=1pt]{};
	\draw [-->-, orange, very thick] (1.1,.3)--(1.1,.65);	
  \end{tikzpicture}
 \end{center}
\caption{Examples of the curve in the amoeba defining the contour~$\gamma_{u,v}$ in the definition of the limit shape. The curves corresponds to~$(u,v)\in R_{\beta,\mu}$ for~$\mu=(-1,1)$ (green),~$\mu=(0,0)$ (orange) and~$(u,v)$ in the rough region (blue).
\label{fig:curves_height_function}}
\end{figure}

\begin{theorem}[Proposition 4.22 of~\cite{BB23}]\label{thm:finite_limit_shape}
Let~$m$,~$n$,~$u_N$ and~$v_N$ be as in Section~\ref{sec:graph} and assume that~$(u_N,v_N)\to(u,v)\in D_\text{Az}$, as~$N\to \infty$, with~$(u,v)$ in the frozen, rough or smooth phase. Then
\begin{equation}
\lim_{N\to \infty}\frac{1}{k\ell N}\,\EE\left[h(2(\ell m+i),2(k n+j))\right]=\bar h_\beta(u,v),
\end{equation}
for~$i=0,\dots,\ell-1$,~$j=0,\dots,k-1$, where~$\bar h_\beta$ is given in~\eqref{eq:limit_shape} and~$h$ in Section~\ref{sec:height_function}. \end{theorem}
\begin{remark}
It is known that the scaled height function converges in probability to a deterministic limit, cf.~\cite{CKP00, Gor21, Kuc17}. The function~$\bar h_\beta$ in the previous theorem is therefore the limit shape. This statement was reproved with relaxed assumptions in~\cite[Theorem 2]{BB24}.
\end{remark}

In addition to the limit shape result, the convergence of local fluctuations was obtained as well (see~\cite[Proposition 28]{BdT24} for the same statement with relaxed assumptions). 
\begin{theorem}[Corollary 4.26 of~\cite{BB23}]\label{thm:finite_local_limit}
Let~$(u,v)\in D_\text{Az}$, not in the arctic curve, and let~$(x,y)=\Log_\beta\circ\,\Omega_\beta$ if~$(u,v)$ is in the rough region, and let~$(x,y)$ be in the interior of~$\mathcal A_{\beta,\mu}$ if~$(u,v)\in R_{\beta,\mu}$,~$\mu\in\mathcal N$. Let 
\begin{equation}
e_r^{(N)}=\mathrm{b}_{\ell m_r+i_r,kn_r+j_r}\mathrm{w}_{\ell m_r'+i_r',kn_r'+j_r'}, \quad r=1,\dots,p,
\end{equation}
be edges in~$G_\text{Az}$, with
\begin{equation}
 m_r=-k\ell N u_N+\kappa_r, \quad m_r'=-k\ell Nu_N+\kappa_r',
\quad n_r=-k\ell N v_N+\zeta_r, \quad n_r'=-k\ell Nv_N+\zeta_r',
\end{equation}
with~$(\kappa_r,\zeta_r),(\kappa_r',\zeta_r')\in \ZZ^2$, and set~$e_r=\mathrm{b}_{\ell \kappa_r+i_r,k\zeta_r+j_r}\mathrm{w}_{\ell \kappa_r'+i_r',k\zeta_r'+j_r'}$. Assume that~$D_\text{Az}\ni(u_N,u_N)\to(u,v)$, as~$N\to \infty$. Then
\begin{equation}
\lim_{N\to\infty}\PP_{\text{Az},\beta}\left[e_1^{(N)},\dots,e_p^{(N)}\in \mathcal D_\text{Az}\right]=\PP_{\beta,(x,y)}\Big[e_1,\dots,e_p\in \mathcal D\Big],
\end{equation}
where the probability measures on the left and right hand side are defined by~\eqref{eq:prob_beta} and~\eqref{eq:gibbs_beta}, respectively. 
\end{theorem}

\subsection{Tropical geometry}\label{sec:tropical_geometery}
The goal of this section is to define tropical characteristic polynomials, tropical curves, and tropical harmonic\footnote{We say \emph{regular} below.} functions that we will need to describe the~$\beta \to \infty$ limit of the dimer model.

\subsubsection{The Newton polygon and the tropical curve}\label{sec:tropical_amoeba}
In this section, we define the tropical curve and discuss its duality to a subdivision of the Newton polygon.

The \emph{tropical characteristic polynomial} is defined as (with~$``x+y"=\max\{x,y\}$, and~$``xy"=x+y$)
\begin{equation}\label{eq:characteristic_polynomial_tropical}
P_t(x,y)=``\sum_{\mu=(\mu_1,\mu_2)\in\mathcal N} \mathcal E^*(\mu)x^{\mu_1}y^{\mu_2}"=\max_{\mu=(\mu_1,\mu_2)\in \mathcal N}\left\{\mu_1 x+\mu_2 y+\mathcal E^*(\mu)\right\},
\end{equation}
and its variety~$\mathcal A_t\subset \RR^2$ is defined as the set of points~$(x,y)\in \RR^2$ where~$P_t$ is not smooth\footnote{It is not uncommon in the tropical literature to use~$t=\e^\beta$ as the large parameter. The subscript in the notation is then~$t$ representing the finite temperature objects, while we use subscript~$t$ for the tropical (\emph{i.e.}, infinite temperature) objects.}. The function~$P_t$ is piecewise-linear and~$\mathcal A_t$ is the set of points where the maximum in the definition of~$P_t$ is not uniquely attained. See~\cite[Section 3]{Mik05}. We refer to~$\mathcal A_t$ as the \emph{tropical curve}. The \emph{extended polyhedral domain} is defined by
\begin{equation}\label{eq:extended_polyhedral_domain}
\tilde N(P_t)=\convHull \left\{(\mu,s) \in \RR^3:\mu \in \mathcal N, s\leq \mathcal E^*(\mu)\right\}.
\end{equation}
The extended polyhedral domain naturally projects to the Newton polygon~$N(P)$, and the projection defines a subdivision~$N_S(P_t)$ of~$N(P)$ such that each bounded face of~$\tilde N(P_t)$ projects to a face in the subdivision~$N_S(P_t)$.

 \begin{figure}[t]
 \begin{center}
\begin{subfigure}[c]{0.35\textwidth}
\includegraphics[scale=.3, angle=180]{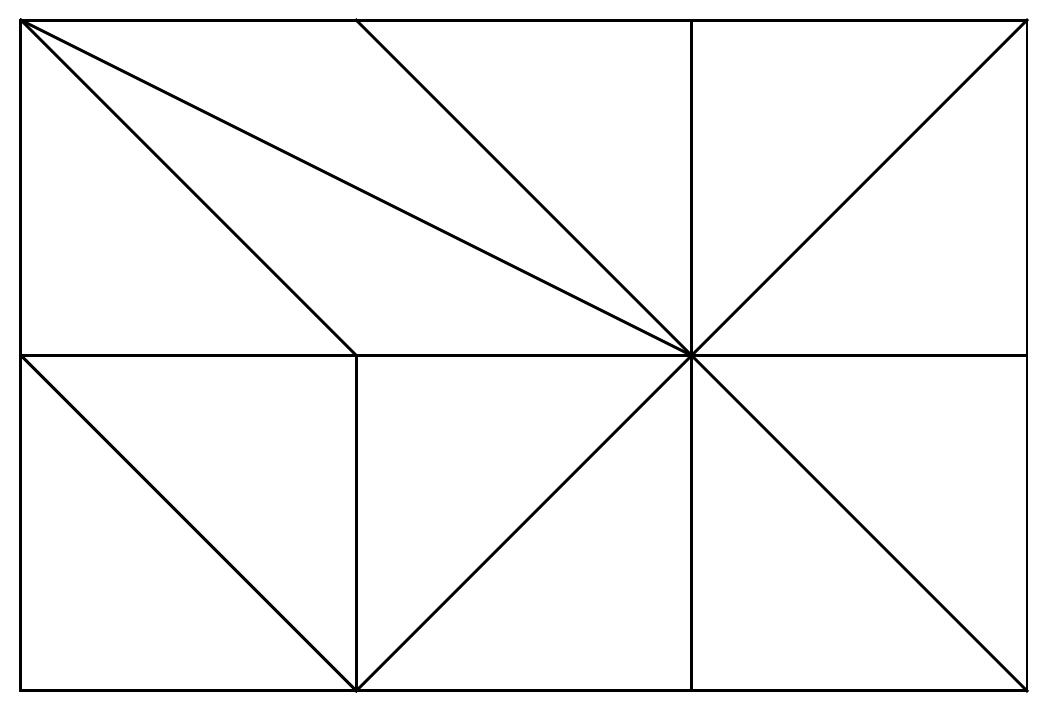}
 \end{subfigure}
\quad
\begin{subfigure}[c]{0.35\textwidth}
\includegraphics[scale=.28]{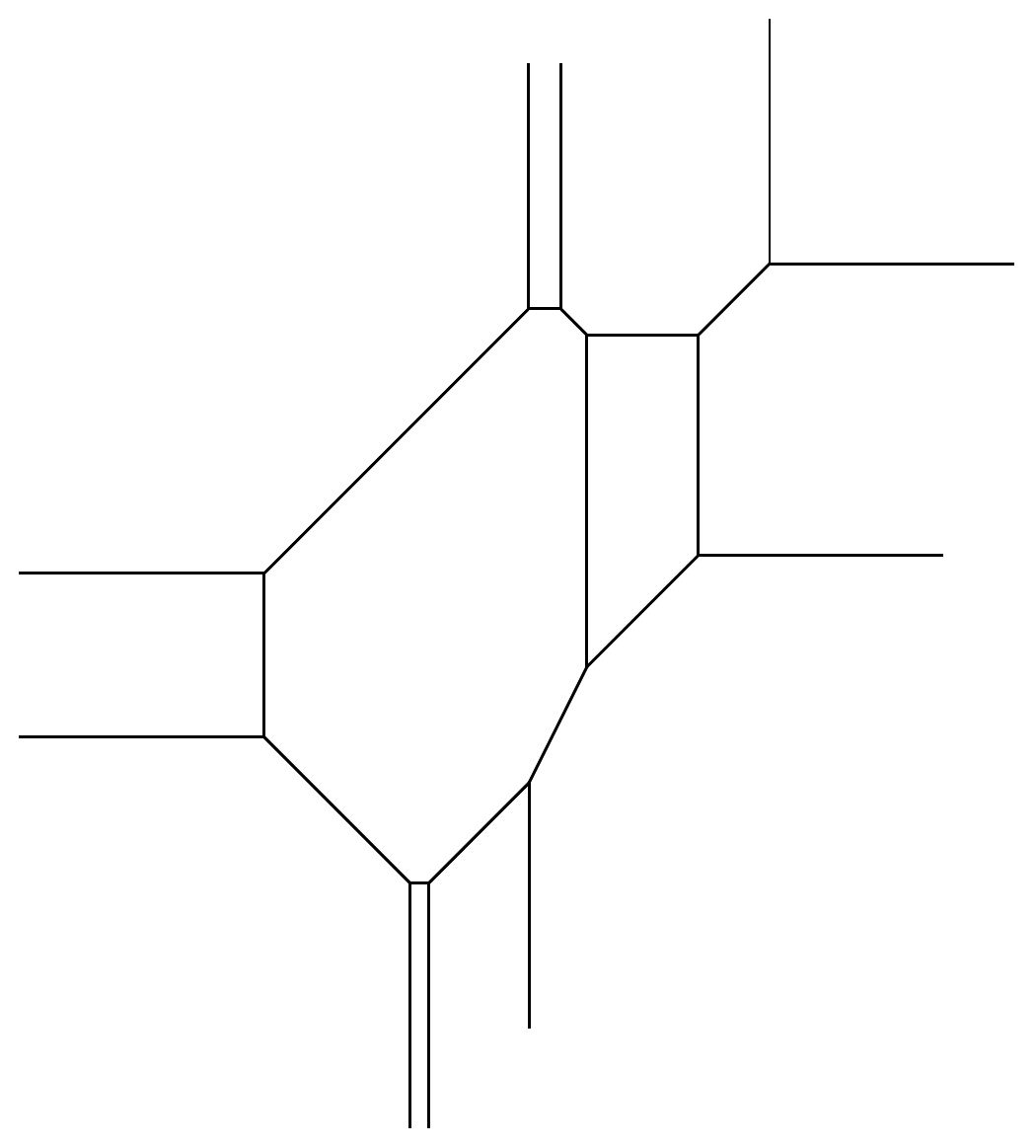}
 \end{subfigure}
 \end{center}
\caption{A subdivision~$N_S(P_t)$ of the Newton polygon (left) and the corresponding tropical curve~$\mathcal A_t$ (right). The periodicity is~$k=2$ and~$\ell=3$.
\label{fig:newton_amoeba}}
\end{figure}

The subdivision~$N_S(P_t)$ is dual to the tropical curve~$\mathcal A_t$ in the following sense (see~\cite[Proposition 2.1]{Mik04b} or~\cite[Proposition 3.11]{Mik05}): For each edge~$e^*\in N_S(P_t)$, there is a closed (possibly unbounded) line segment~$e\subset \mathcal A_t$, with non-empty interior, which is orthogonal to~$e^*$, and each face~$\mathrm v^*\in N_S(P)$ corresponds to a point~$\mathrm v\in \mathcal A_t$, with the degree of~$\mathrm v$ equal to the number of sides of the face~$\mathrm v^*$. The line segment is unbounded precisely if~$e^*$ is an edge in the boundary of~$N(P)$. Moreover, the interiors of all edges~$e$ and all points~$\mathrm v$ form a partition of~$\mathcal A_t$. See Figure~\ref{fig:newton_amoeba}. We denote the vertices of~$\mathcal A_t$ by~$V(\mathcal A_t)$, the infinite edges, called \emph{leaves}, by~$L(\mathcal A_t)$, the bounded edges by~$E(\mathcal A_t)$, and~$LE(\mathcal A_t)=L(\mathcal A_t)\cup E(\mathcal A_t)$. We divide the set~$L(\mathcal A_t)$ into four subsets~$L_i(\mathcal A_t)$ so that the edges in~$L_i(\mathcal A_t)$,~$i=1,2,3,4$, correspond to the left-most, bottom-most, right-most, and top-most edges of~$N_S(P_t)$, respectively. The above duality defines a natural bijection between the vertices of~$N_S(P_t)$ and the components of the complement of~$\mathcal A_t$. We denote the subset of~$\mathcal A_t$ bounding a component corresponding to a vertex~$\mu$ of~$N_S(P_t)$ by~$\mathcal A_{t,\mu}$, and the set of all~$\mathcal A_{t,\mu}$ by~$B(\mathcal A_t)$. See the left image of Figure~\ref{fig:amoeba_tropical_amoeba}.

By definition of the subdivision~$N_S(P_t)$, the vertices of~$N_S(P_t)$ is a subset of~$\mathcal N$ (hence the notation~$\mu$ in the previous paragraph). For our purposes, we will assume that the set of vertices of~$N_S(P_t)$ and~$\mathcal N$ coincide. This is generically true, as we prove in Section~\ref{sec:subdivision} below.

Each edge~$e\in LE(\mathcal A_t)$ has a primitive integer tangent (\emph{i.e.}, co-linear) vector~$\eta=(\eta_1,\eta_2)\in \ZZ^2$ with~$\eta_1$ and~$\eta_2$ relatively prime. For an oriented edge~$e\in LE(\mathcal A_t)$ we write~$\eta=\eta(e)$ where~$\eta$ points in the direction of~$e$. For a vertex~$\mathrm v\in V(\mathcal A_t)$ we say that~$\eta$ is an \emph{outward pointing primitive vector} of~$\mathrm v$ if~$\eta=\eta(e)$ is the primitive integer tangent vector of an edge~$e$ adjacent to~$\mathrm v$ directed away from~$\mathrm v$. For~$e\in E(\mathcal A_t)$, we define the \emph{length}~$l(e)$ to be the positive real number such that~$l(e)\eta(e)=e$.

\begin{figure}
\begin{center}
\begin{tikzpicture}
    \coordinate (v1) at (0,0);

    \draw[-latex] (v1) -- ($(v1)+(-1,1)$);
    \draw ($(v1)+(-1,1.1)$) node[right] {$\eta=(-1,1)$};
    \draw ($(v1)+(-.9,.9)$) -- ($(v1)+(-1.5,1.5)$);
    \draw[-latex] (v1) -- ($(v1)+(-1,0)$) node[below] {$\eta=(-1,0)$};
    \draw ($(v1)+(-.9,0)$) -- ($(v1)+(-1.5,0)$);
    \draw[-latex] (v1) -- ($(v1)+(2,-1)$);
    \draw ($(v1)+(2.7,-1)$) node[above] {$\eta=(2,-1)$};
    \draw ($(v1)+(1.9,-.95)$) -- ($(v1)+(2.5,-1.25)$);

    \node[above right] at (v1) {$\mathrm v$};
    
    \coordinate (v2) at (7.5,1.);

    \draw[-latex] (v2) -- ($(v2)+(-1,-1)$);
    \draw ($(v2)+(-1,-1)$) node[above left] {$(-1,-1)$};
    \draw[-latex] ($(v2)+(-1,-1)$) -- ($(v2)+(-1,-2)$) node[left] {$(0,-1)$};
    \draw[-latex] ($(v2)+(-1,-2)$) -- (v2) node[right] {$(1,2)$};

    \node[right] at ($(v2)+(-1,-1)$) {$\mathrm v^*$};
\end{tikzpicture}
\end{center}
\caption{An example of outward pointing primitive vectors of the vertex~$\mathrm v\in V(\mathcal A_t)$ corresponding to the face~$\mathrm v^*\in N_S(P_t)$ showing that they necessarily sum to zero.
\label{fig:balancing_property}}
\end{figure}
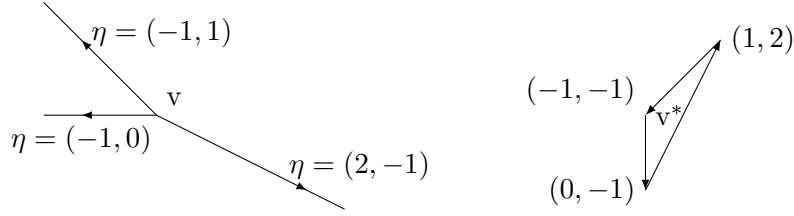

The tropical curve~$\mathcal A_t$ satisfies the following balancing property, see~\cite[Equation (3)]{Mik05}. For each vertex~$\mathrm v\in V(\mathcal A_t)$,
\begin{equation}\label{eq:balancing_primitive}
\sum_{e\sim \mathrm v} \eta(e)=0,
\end{equation} 
where the sum runs over all edges adjacent to~$\mathrm v$ oriented away from~$\mathrm v$. Without the assumption that the set of vertices of~$N_S(P_t)$ coincides with~$\mathcal N$, the linear combination~\eqref{eq:balancing_primitive} would need to be weighted by the integer length of the corresponding edge in~$N_S(P_t)$. Note that if we rotate the vectors~$\eta(e)$ with~$e\sim \mathrm v$ by~$\pi/2$, they would coincide with the (oriented) edges of the face in~$N_S(P_t)$ corresponding to the vertex~$\mathrm v$, and the balancing property simply says that the sum of the oriented edges of any face is zero, see, \emph{e.g.},~\cite[Section 2.3]{BIMS15} and Figure~\ref{fig:balancing_property}.

\subsubsection{Regular functions and differentials}\label{sec:tropical_functions}
In this section, we recall the notion of regular functions and~$1$-forms on~$\mathcal A_t$. For a more detailed discussion, we refer to~\cite[Section 2.5]{Lan20}, see also~\cite{MZ08}.

Let~$f$ be a continuous piecewise affine function on~$\mathcal A_t$, and let~$\mathrm v\in V(\mathcal A_t)$ be a vertex with an outward pointing primitive vector~$\eta$. The \emph{slope} of~$f$ at~$\mathrm v$ in the direction of~$\eta$, denoted by~$\d_\eta f(\mathrm v)$, is given by~$f(\mathrm v+\eta)-f(\mathrm v)$. A \emph{regular} function~$f:\mathcal A_t\to \RR$ is a continuous piecewise affine function such that the balancing property
\begin{equation}\label{eq:balancing_function}
\sum_{\eta \sim \mathrm v} \d_\eta f(\mathrm v)=0
\end{equation}
holds, where the sum runs over all outward pointing primitive vectors of~$\mathrm v$. If the degree of the vertex is three (which is the setting we will be interested in), the balancing property is equivalent to the fact that there is a neighborhood of~$\mathrm v$ such that the graph of~$f$ in that neighborhood lies in some plane.

A (tropical) \emph{$1$-form}~$\mathrm w$ on~$\mathcal A_t$ is a locally constant, real-valued 1-from on~$\mathcal A_t \backslash V(\mathcal A_t)$ satisfying the balancing condition 
\begin{equation}\label{eq:balancing_differential}
\sum_{\eta \sim \mathrm v}\mathrm w(\eta)=0
\end{equation}
at each vertex~$\mathrm v\in V(\mathcal A_t)$, where the sum runs over the outward pointing primitive vectors of~$\mathrm v$. We denote the vector space of~$1$-forms by~$\Omega(\mathcal A_t)$. A 1-form~$\mathrm w$ is determined by assigning a value to each oriented primitive integer tangent vector~$\eta=\eta(e)$, and we write~$\mathrm w(\eta)$ for that value. The \emph{residue} of a 1-form~$\mathrm w$ at a leaf in~$L(\mathcal A_t)$ is~$\mathrm w(\eta)$ where~$\eta$ is a primitive vector pointing toward the vertex adjacent to the leaf. 

A \emph{path} in~$\mathcal A_t$ is a continuous function~$\rho:[0,1]\to \mathcal A_t$ such that~$\rho(0),\rho(1)\in V(\mathcal A_t)$ and its restriction to~$(0,1)$ is injective. A \emph{loop} is a path~$\rho$ with~$\rho(0)=\rho(1)$. The integral of a~$1$-form~$\mathrm w$ along a path~$\rho$ joining two vertices of~$\mathcal A_t$ is defined as
\begin{equation}\label{eq:line_integral}
\int_\rho \mathrm w=\sum_{i=1}^nl(e_i)\mathrm w(\eta(e_i)),
\end{equation}
where~$e_i\in E(\mathcal A_t)$ are the edges in the image of~$\rho$. The space of \emph{exact~$1$-forms} is the subset~$\Omega_0(\mathcal A_t)\subset \Omega(\mathcal A_t)$ consisting of the~$1$-forms with~$\int_\rho \mathrm w=0$ for all loops~$\rho$. The space~$\Omega_0(\mathcal A_t)$ is characterized by the residues, more precisely: Let~$r_e\in \RR$,~$e\in L(\mathcal A_\beta)$, be the residues of~$\mathrm w\in \Omega_0(\mathcal A_t)$ at~$e$; then~$\sum_{e\in L(\mathcal A_\beta)} r_e=0$. Conversely, if~$r_e\in \RR$,~$e\in L(\mathcal A_\beta)$, with~$\sum_{e\in L(\mathcal A_\beta)} r_e=0$, then there exists a unique exact 1-form with residues~$r_e$ at the leaves~$e\in L(\mathcal A_\beta)$, see, \emph{e.g.}, {\cite[Proposition 2.27]{Lan20}}.

Given a regular function~$f$ we define~$\d f$ as the~$1$-form satisfying~$\d f(\eta)=\d_\eta f$ for all~$\eta=\eta(e)$,~$e\in LE(\mathcal A_t)$. By the balancing condition~\eqref{eq:balancing_function}, we have~$\d f\in \Omega_0(\mathcal A_t)$. Conversely, given a 1-form~$\mathrm w\in \Omega_0(\mathcal A_t)$, we can define~$f$ as a piecewise affine function on~$\mathcal A_t$ by the relation~$\mathrm w(\eta)=\d_\eta f$ for all~$\eta=\eta(e)$ with~$e\in LE(\mathcal A_t)$. This defines a continuous function~$f$, up to an additive constant, since~$\mathrm w$ is exact. Moreover, the balancing condition~\eqref{eq:balancing_differential} ensures that~$f$ is regular. This is a simple corollary of Kirchhoff's theorem, see also~\cite[Proposition 3.9]{Lan20}.

\subsection{Tropical limit}\label{sec:tropical_limit}

In this section, we connect Sections~\ref{sec:finite_spectral_curve} and~\ref{sec:tropical_geometery} by taking the large~$\beta$ limit. The results that are relevant for us were obtained in~\cite{Mik04b} and~\cite{Lan20}. See also~\cite{GKN19, Iwa10}. 

The tropical curve~$\mathcal A_t$ is constructed to be the limit of a sequence of amoebas. In our setting, this limit is especially easy. Indeed, the amoebas~$\mathcal A_\beta$, defined in Section~\ref{sec:finite_spectral_curve}, are constructed from the characteristic polynomial~$P_\beta$ given in~\eqref{eq:characteristic_polynomial_finite}. This polynomial has the form of a so-called patchworking polynomial, and therefore we can rely directly on~\cite[Corollary 6.4]{Mik04b} (see also~\cite[Theorem 2.12]{BIMS15}) which tells us that
\begin{equation}\label{eq:limit_amoeba}
\mathcal A_\beta \to \mathcal A_t \quad\text{as } \beta \to \infty,
\end{equation}
where the convergence is in the Hausdorff sense. The convergence can be strengthened according to~\cite[Theorem 5]{Mik04b}, but such a stronger result is not necessary for our purposes.

In addition to the limit of the spectral curve to its tropical version, imaginary normalized differentials, see Section~\ref{sec:finite_spectral_curve}, converge in the large~$\beta$ limit to the exact tropical 1-forms, as we now describe.

There is a natural bijection between the angles~$p\in L(\mathcal R_\beta)$ and the leaves~$e\in L(\mathcal A_t)$. Indeed, for each leaf, there is an angle so that the leaf is an asymptote of the tentacle corresponding to the angle, see Figure~\ref{fig:amoeba_tropical_tentacles}. Let us denote the leaf corresponding to the angle~$p\in L(\mathcal R_\beta)$ by~$e_p\in L(\mathcal A_t)$.

Let~$R$ be a collection of real numbers~$r_{p,j}\in \RR$,~$p\in L(\mathcal R_\beta)$,~$j=1,\dots,m$, for some positive integer~$m$, with~$\sum_{p\in L(\mathcal R_\beta)} r_{p,j}=0$ for all~$j$. Let~$\omega_{\beta,j}$ be the imaginary normalized differentials with residues~$r_{p,j}$ at~$p\in L(\mathcal R_\beta)$, and let~$\mathrm w_j\in \Omega_0(\mathcal A_t)$ be the exact~$1$-form with residues~$r_{p,j}$ at the leaf~$e_p\in L(\mathcal A_t)$, see Section~\ref{sec:tropical_functions}. Set
\begin{equation}\label{eq:log_map_generalized}
\Log_{\beta,R}(q_\beta)=\frac{1}{\beta}\left(\re\left(\int_{p_\beta}^{q_\beta}\omega_{\beta,1}\right),\dots,\re \left(\int_{p_\beta}^{q_\beta}\omega_{\beta,m}\right)\right), \quad q_\beta \in \mathcal R_\beta, 
\end{equation}
and
\begin{equation}\label{eq:log_map_generalized_tropical}
\Log_{t,R}(q_t)=\left(\int_{\mathrm v_t}^{q_t} \mathrm w_1,\dots,\int_{\mathrm v_t}^{q_t}\mathrm w_m\right), \quad q_t\in \mathcal A_t,
\end{equation}
for some~$\mathrm v_t\in V(\mathcal A_t)$ and some points~$p_\beta \in \mathcal R_\beta$. With this setup,~\cite[Theorem 1]{Lan20} tells us that the image of the map~$\Log_{\beta,R}$ converges to the image of~$\Log_{t,R}$, 
\begin{equation}\label{eq:limit_harmonic_amoeba}
\Log_{\beta,R}(\mathcal R_\beta)\to \Log_{t,R}(\mathcal A_t) \quad \text{as } \beta \to \infty,
\end{equation}
in Hausdorff distance on compact sets under the condition that the points~$p_\beta$ are picked so that $\Log_\beta (p_\beta)\to \mathrm v_t$ as~$\beta \to \infty$.\footnote{The convergence of the (normalized by~$\beta$ logarithms of the absolute values of the) coefficients of the characteristic polynomial~$P_\beta(z,w)$ given by~\eqref{eq:characteristic_polynomial_finite} to the ``coefficients''~$\mathcal E^*(\mu)$ of the tropical polynomial~$P_t(x,y)$ from \eqref{eq:characteristic_polynomial_tropical} implies the abstract tropical convergence of~\cite[Definition 1.1]{Lan20} of the Riemann surfaces~$\mathcal R_\beta$ to the tropical curve~$\mathcal A_t$ under the assumption that~$\mathcal A_t$ is smooth in the sense of Definition~\ref{def:t-curve} below~\cite{Lan24}.} The freedom in the choice of the points~$p_\beta$ is explained in the proof of~\cite[Theorem 1]{Lan20}. The convergence~\eqref{eq:limit_harmonic_amoeba} is used in the proof of Theorem~\ref{thm:arctic_curve_tropical} below to prove the convergence of functions defined by residues of their differentials to their tropical counterparts.

Next, we will discuss integrals of~$\omega_{\beta,j}$,~$j=1,\dots,m$, over shrinking closed loops. As~$\beta\to \infty$, we can think of the part of~$\mathcal R_\beta$ converging to an edge~$e\in LE(\mathcal A_t)$ as a shrinking tube. We define a set of disjoint simple closed loops~$\gamma_{\beta,e}$,~$e\in E(\mathcal A_t)$, such that~$\gamma_{\beta,e}$ circles around the shrinking tube around~$e$. These loops are constructed so that 
\begin{equation}
\mathcal R_\beta \backslash\left(\bigcup_{e\in E(\mathcal A_t)} \gamma_{\beta,e}\bigcup_{p\in L(\mathcal R_\beta)}\{p\}\right)
\end{equation}
consists of~$|V(\mathcal A_t)|$ connected components, called \emph{pairs-of-pants}. Each pair of pants corresponds naturally to a vertex~$v\in V(\mathcal A_t)$\footnote{In our (generic) situation, all vertices of~$\mathcal A_t$ have valency~$3$, in correspondence with triangulations of Section~\ref{sec:tropical_amoeba}, see also a discussion around Definition~\ref{def:t-curve} below.}. See, \emph{e.g.},~\cite[Section 2.3]{Lan20} for a more detailed definition. 

For an edge~$e\in E(\mathcal A_t)$ oriented towards~$\mathrm v\in V(\mathcal A_t)$, we orient~$\gamma_{\beta,e}$ so that the pair of pants corresponding to~$\mathrm v$ lies to the right of~$\gamma_{\beta,e}$. 
Additionally, we define simple closed curves~$\gamma_{\beta,e}$,~$e=e_p\in L(\mathcal A_t)$, enclosing the corresponding angle~$p\in L(\mathcal A_\beta)$ and oriented in positive direction around~$p$. These curves are taken so that~$\gamma_{\beta,e}$,~$e\in LE(\mathcal A_t)$, are disjoint. It now follows from~\cite[Theorem 3]{Lan20} that
\begin{equation}\label{eq:limit_b-cycles}
\lim_{\beta\to\infty}\frac{1}{2\pi\i}\int_{\gamma_{\beta,e}}\omega_{\beta,j}=\mathrm w_j(\eta(e)),
\end{equation}
for~$j=1,\dots,m$, and all~$e\in LE(\mathcal A_t)$.

\section{The tropical arctic curves and limit shape}\label{sec:tropical_global_results}
The aim of this section is to obtain the zero-temperature limit of the Aztec diamond dimer model~\eqref{eq:model_beta}. Tropical geometry provides a natural language to describe this limit. We first identify the tropical counterparts of essential objects from the finite~$\beta$ setting. This allows us to describe the zero-temperature limit in a manner analogous to the finite~$\beta$ case. To this end, we define the tropical action function, which plays the role of the action function~\eqref{eq:action_function_beta}. Subsequently, we introduce the appropriate (for our purposes) notion of zeros of the derivative of the tropical action function and use these zeros to define the limiting phases within the Aztec diamond. With those notions in place, we derive the large~$\beta$ limit of the limit shape~\eqref{eq:limit_shape} and obtain a limiting result for local statistics.

Recall the subdivision of the Newton polygon~$N_S(P_t)$. As discussed in the previous section, the tropical curve~$\mathcal A_t$ can take many different forms, depending on the subdivision. Generically, however, the tropical curve is smooth, as defined next.  
\begin{definition}\label{def:t-curve}
A tropical curve is said to be a \emph{smooth} if the subdivision~$N_S(P_t)$ is a triangulation consisting of triangles with area~$1/2$.
\end{definition}
Smooth tropical curves were discussed already in~\cite{IV96} (see also~\cite{Mik00}), where they were used to construct Harnack curves using a patchworking method. We will prove in Section~\ref{sec:subdivision} that~$\mathcal A_t$ indeed is, generically, a smooth tropical curve. Throughout this section, we will assume that~$\mathcal A_t$ is a smooth tropical curve.

\subsection{The tropical action function and the tropical arctic curve}\label{sec:tropical_action}
As described in Section~\ref{sec:finite_beta_results}, the asymptotic analysis of the doubly periodically weighted Aztec diamond was performed in~\cite{BB23} via the action function~$F_\beta$. In the large~$\beta$ limit, the central object for our analysis here will be the limit of~$\re F_\beta$. 

Let~$\d f_t\in \Omega_0(\mathcal A_t)$ be the unique exact~$1$-form with residues~$-\ell$ at each~$e\in L_1(\mathcal A_t)$, residues~$k$ at each~$e\in L_2(\mathcal A_t)$, and residues~$0$ in~$L_i(\mathcal A_t)$,~$i=3,4$, see the end of Section~\ref{sec:tropical_functions}. For~$(u,v)\in \RR^2$ we define the \emph{tropical action function}~$F_t$ for~$(x,y)\in\mathcal A_t$ by
\begin{equation}\label{eq:action_function_tropical}
F_t(x,y;u,v)=k(1+\ell u)y-\ell(1+kv)x-f_t(x,y).
\end{equation}
The differential~$\d f_t$ only defines~$f_t$ up to an additive constant, see Section~\ref{sec:tropical_functions}, however, the constant will not be relevant for our purposes. The residues of the~$1$-form~$\d F_t$ at the leaves are
\begin{multline}\label{eq:action_function_residues}
\d F_t(\eta(e);u,v)=-k\ell v, \,\, e\in L_1(\mathcal A_t), \quad \d F_t(\eta(e);u,v)=k\ell u, \,\, e\in L_2(\mathcal A_t), \\
\d F_t(\eta(e);u,v)=\ell(1+k v), \,\, e\in L_3(\mathcal A_t), \quad \text{and} \quad \d F_t(\eta(e);u,v)=-k(1+\ell u), \,\, e\in L_4(\mathcal A_t). 
\end{multline}
Recall that~$\eta(e)$ is the primitive tangent vector of~$e$ pointing towards the adjacent vertex of the leaf.

In~\cite{BB23}, limiting regions of the dimer model were determined through the zeros of the differential of the action function. As its replacement, we consider here the~$1$-form~$\d F_t\in \Omega_0(\mathcal A_t)$. After defining the right notion of zeros, the zeros of~$\d F_t$ will determine the macroscopic regions in the zero-temperature limit. 

We define four types of zeros of~$\d F_t$ which we divide into two groups, see Figure~\ref{fig:local_slope}.
\begin{definition}\label{def:v-zero}
For a vertex~$\mathrm v\in V(\mathcal A_t)$ in a component~$\mathcal A_{t,\mu}\in B(\mathcal A_t)$, we say that~$\d F_t$ has a \emph{v-zero at~$\mathrm v$ with respect to~$\mathcal A_{t,\mu}$} if~$\d F_t(\eta;u,v)$ and~$\d F_t(\eta';u,v)$ are either both~$0$ or have the same sign, where~$\eta$ and~$\eta'$ are outward pointing primitive vectors at~$\mathrm v$ of the adjacent edges in~$\mathcal A_{t,\mu}$. 

We say that the v-zero is \emph{simple} if~$\d F_t(\eta;u,v)\neq 0$ and we say it is a \emph{triple} v-zero if~$\d F_t(\eta;u,v)=\d F_t(\eta';u,v)=0$.\footnote{``Triple'' refers to the fact that~$\d F_t$ is zero on all three edges adjacent to~$\mathrm v$. The simple v-zeros and e-zeros play a role analogous to that of the simple critical points in~\cite{BB23}, while the double e-zeros and triple v-zeros serve a purpose similar to their double and triple counterparts in the finite~$\beta$ case, which is another reason for their names.}
\end{definition}
The zeros defined above naturally live at the vertices, therefore the name v-zeros. Let us now define zeros which naturally live on the edges.

\begin{definition}\label{def:e-zeros}
For~$e\in E(\mathcal A_t)$, we say that~$\d F_t$ has an \emph{e-zero} at~$e$ if~$\d F_t(\eta(e),u,v)=0$ and~$\d F_t(\eta(e');u,v)\neq 0$ for all adjacent edges~$e'$ of~$e$.

Consider a component~$\mathcal A_{t,\mu}\in B(\mathcal A_t)$ containing the edge~$e$ and let~$\mathrm v', \mathrm v'' \in V(\mathcal A_t)$ be the endpoints of~$e$ that are not tripple v-zeros. In addition, let~$\eta'\neq \pm\eta(e)$ and~$\eta''\neq \pm\eta(e)$ be outward pointing primitive vectors at~$\mathrm v'$ and~$\mathrm v''$, respectively, in~$\mathcal A_{t,\mu}$. We say that the e-zero~$e$ of~$\d F_t$ is \emph{simple} if the signs of~$\d F_t(\eta';u,v)$ and~$\d F_t(\eta'';u,v)$ are the same. If the signs of~$\d F_t(\eta;u,v)$ and~$\d F_t(\eta';u,v)$ are different, we say that~$e$ is a \emph{double} e-zero of~$\d F_t$.
\end{definition}
\begin{remark}\label{rem:adjacent_components}
In the above definition, we singled out one of the two adjacent components. However, the definition is independent of this choice, by the balancing condition~\eqref{eq:balancing_primitive}. We did not specify any orientation on the edge~$e$, which means~$\eta(e)$ is only specified up to a sign, however, such choice is irrelevant for the definition. 
\end{remark}

 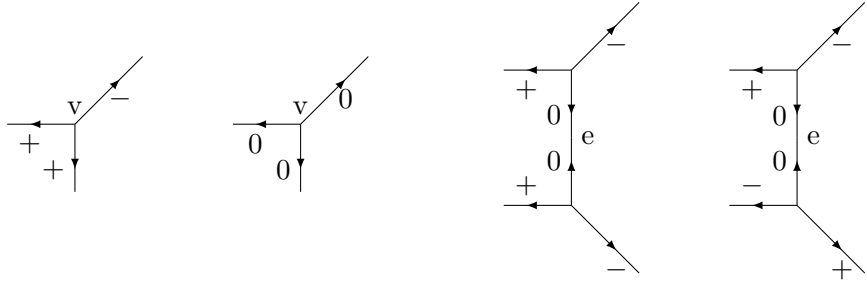
\begin{figure}[t]
 \begin{center}
\begin{tikzpicture}[scale=.6]
\begin{scope}
\draw[-latex] (0,0) -- (1,1) node[below] {$-$};
\draw (.9,.9) -- (1.5,1.5);
\draw[-latex] (0,0) -- (-1,0) node[below] {$+$};
\draw (-.9,0) -- (-1.5,0);
\draw[-latex] (0,0) -- (0,-1.) node[left] {$+$};
\draw (0,-.9) -- (0,-1.5);
\draw (0,0) node[above] {$\mathrm v$};
\end{scope}


\begin{scope}[xshift=5cm]
\draw[-latex] (0,0) -- (1,1) node[below] {$0$};
\draw (.9,.9) -- (1.5,1.5);
\draw[-latex] (0,0) -- (-1,0) node[below] {$0$};
\draw (-.9,0) -- (-1.5,0);
\draw[-latex] (0,0) -- (0,-1.) node[left] {$0$};
\draw (0,-.9) -- (0,-1.5);
\draw (0,0) node[above] {$\mathrm v$};
\end{scope}

\begin{scope}[xshift=11cm, yshift=1.2cm]
\draw[-latex] (0,0) -- (1,1) node[below] {$-$};
\draw (.9,.9) -- (1.5,1.5);
\draw[-latex] (0,0) -- (-1,0) node[below] {$+$};
\draw (-.9,0) -- (-1.5,0);
\draw[-latex] (0,0) -- (0,-1.) node[left] {$0$};
\draw (0,-.9) -- (0,-1.5);
\draw[-latex] (0,-3) -- (1,-4) node[below] {$-$};
\draw (.9,-3.9) -- (1.5,-4.5);
\draw[-latex] (0,-3) -- (-1,-3) node[above] {$+$};
\draw (-.9,-3) -- (-1.5,-3);
\draw[-latex] (0,-3) -- (0,-2) node[left] {$0$};
\draw (0,-2.1) -- (0,-1.5);
\draw (0,-1.5) node[right] {$\mathrm e$};
\end{scope}

\begin{scope}[xshift=16cm, yshift=1.2cm]
\draw[-latex] (0,0) -- (1,1) node[below] {$-$};
\draw (.9,.9) -- (1.5,1.5);
\draw[-latex] (0,0) -- (-1,0) node[below] {$+$};
\draw (-.9,0) -- (-1.5,0);
\draw[-latex] (0,0) -- (0,-1.) node[left] {$0$};
\draw (0,-.9) -- (0,-1.5);
\draw[-latex] (0,-3) -- (1,-4) node[below] {$+$};
\draw (.9,-3.9) -- (1.5,-4.5);
\draw[-latex] (0,-3) -- (-1,-3) node[above] {$-$};
\draw (-.9,-3) -- (-1.5,-3);
\draw[-latex] (0,-3) -- (0,-2) node[left] {$0$};
\draw (0,-2.1) -- (0,-1.5);
\draw (0,-1.5) node[right] {$\mathrm e$};
\end{scope}
\end{tikzpicture}
 \end{center}
\caption{The possible local configurations (up to an affine transformation at each vertex) of signs and zeros of~$\d F_t$. Two left-most figures show a simple v-zero and a triple v-zero, respectively. Two right-most figures show a simple e-zero and a double e-zero, respectively.
\label{fig:local_slope}}
\end{figure}

Let us describe how v-zeros and e-zeros of~$\d F_t$ depend on~$(u,v)$. Note, first of all, that for~$(u,v)\in D_\text{Az}$,~$\d F_t$ is never zero on a leaf, cf.~\eqref{eq:action_function_residues}. Next, we have the following lemma.
\begin{lemma}\label{lem:possible_zeros}
Assume~$(u,v)\in \RR^2$ is such that~$\d F_t$ has no triple v-zeros. If~$\mathrm v\in V(\mathcal A_t)$ then either~$\mathrm v$ is a v-zero with respect to one of the adjacent components~$\mathcal A_{t,\mu}$ or exactly one of the edges adjacent to~$\mathrm v$ is an e-zero of~$\d F_t$. 
\end{lemma}
\begin{proof}
If~$\d F_t$ is non-zero on all edges adjacent to~$\mathrm v$, the balancing condition~\eqref{eq:balancing_differential} implies that~$\mathrm v$ is a simple v-zero with respect to one of the adjacent components (two of the three numbers that add up to~$0$ must be of the same sign). If~$\d F_t$ is zero on one of the adjacent edges, then, since there are no triple v-zeros, the edge is an e-zero of~$\d F_t$. 
\end{proof}
The previous lemma tells us that if~$(u,v)$ is such that~$\d F_t(\eta(e);u,v)\neq 0$ for all edges~$e\in LE(\mathcal A_t)$, then each vertex~$\mathrm v\in V(\mathcal A_t)$ is a v-zero with respect to some~$\mathcal A_{t,\mu}$. This is indeed the case for all~$(u,v)$ away from a finite number of lines in the~$(u,v)$-plane.
\begin{lemma}\label{lem:line_e-zeros}
For a fixed edge~$e\in LE(\mathcal A_t)$ (with some fixed orientation), the equation~$\d F_t(\eta(e),u,v)=0$ defines a line in the~$(u,v)$-plane that is parallel to~$\eta(e)$.
\end{lemma}
\begin{proof}
Let~$\eta(e)=(\eta_1,\eta_2)$. The definition of the action function~\eqref{eq:action_function_tropical} implies that
\begin{equation}
\d F_t(\eta(e);u,v)=k(1+\ell u)\eta_2-\ell(1+kv)\eta_1-\d f_t(\eta(e)).
\end{equation}
Equating this expression to~$0$ defines a line which is parallel to~$\eta(e)$. 
\end{proof}

 \begin{figure}[t]
 \begin{center}
\begin{tikzpicture}[scale=.6]
\begin{scope}
\draw[-latex] (0,0) -- (1,1) node[below] {$-$};
\draw (.9,.9) -- (1.5,1.5);
\draw[-latex] (0,0) -- (-1,0) node[below] {$+$};
\draw (-.9,0) -- (-1.5,0);
\draw[-latex] (0,0) -- (0,-1.) node[left] {$0$};
\draw (0,-.9) -- (0,-1.5);
\draw[-latex] (0,-3) -- (1,-4) node[below] {$-$};
\draw (.9,-3.9) -- (1.5,-4.5);
\draw[-latex] (0,-3) -- (-1,-3) node[above] {$+$};
\draw (-.9,-3) -- (-1.5,-3);
\draw[-latex] (0,-3) -- (0,-2) node[left] {$0$};
\draw (0,-2.1) -- (0,-1.5);
\draw (0,-1.5) node[right] {$\mathrm e$};

\draw[->] (1.5,-1.5) -- (2.5,-1.5);
\end{scope}

\begin{scope}[xshift=5cm]
\draw[-latex] (0,0) -- (1,1) node[below] {$-$};
\draw (.9,.9) -- (1.5,1.5);
\draw[-latex] (0,0) -- (-1,0) node[below] {$+$};
\draw (-.9,0) -- (-1.5,0);
\draw[-latex] (0,0) -- (0,-1.) node[left] {$+$};
\draw (0,-.9) -- (0,-1.5);
\draw (0,0) node[below left] {$\mathrm v$};
\draw[-latex] (0,-3) -- (1,-4) node[below] {$-$};
\draw (.9,-3.9) -- (1.5,-4.5);
\draw[-latex] (0,-3) -- (-1,-3) node[above] {$+$};
\draw (-.9,-3) -- (-1.5,-3);
\draw[-latex] (0,-3) -- (0,-2) node[left] {$-$};
\draw (0,-2.1) -- (0,-1.5);
\draw (0,-3) node[right] {$\mathrm v'$};
\end{scope}

\begin{scope}[xshift=13cm]
\draw[-latex] (0,0) -- (1,1) node[below] {$+$};
\draw (.9,.9) -- (1.5,1.5);
\draw[-latex] (0,0) -- (-1,0) node[below] {$-$};
\draw (-.9,0) -- (-1.5,0);
\draw[-latex] (0,0) -- (0,-1.) node[left] {$0$};
\draw (0,-.9) -- (0,-1.5);
\draw[-latex] (0,-3) -- (1,-4) node[below] {$-$};
\draw (.9,-3.9) -- (1.5,-4.5);
\draw[-latex] (0,-3) -- (-1,-3) node[above] {$+$};
\draw (-.9,-3) -- (-1.5,-3);
\draw[-latex] (0,-3) -- (0,-2) node[left] {$0$};
\draw (0,-2.1) -- (0,-1.5);
\draw (0,-1.5) node[right] {$\mathrm e$};

\draw[->] (1.5,-1.5) -- (2.5,-1.5);
\end{scope}

\begin{scope}[xshift=18cm]
\draw[-latex] (0,0) -- (1,1) node[below] {$+$};
\draw (.9,.9) -- (1.5,1.5);
\draw[-latex] (0,0) -- (-1,0) node[below] {$-$};
\draw (-.9,0) -- (-1.5,0);
\draw[-latex] (0,0) -- (0,-1.) node[left] {$+$};
\draw (0,-.9) -- (0,-1.5);
\draw (0,0) node[right] {$\mathrm v$};
\draw[-latex] (0,-3) -- (1,-4) node[below] {$-$};
\draw (.9,-3.9) -- (1.5,-4.5);
\draw[-latex] (0,-3) -- (-1,-3) node[above] {$+$};
\draw (-.9,-3) -- (-1.5,-3);
\draw[-latex] (0,-3) -- (0,-2) node[left] {$-$};
\draw (0,-2.1) -- (0,-1.5);
\draw (0,-3) node[right] {$\mathrm v'$};
\end{scope}
\end{tikzpicture}
 \end{center}
\caption{The figures illustrate the signs and zeros of~$\d F_t$ for a simple e-zero (left) and a double e-zero (right).  Slightly changing the position of~$(u,v)$ transforms these e-zeros into v-zeros. The positions of~$\mathrm v$ and~$\mathrm v'$ indicate the specific components to which these v-zeros belong. The v-zeros are with respect to different components if the e-zero is simple, and with respect to the same component if the e-zero is double.
\label{fig:e-zeros}}
\end{figure}
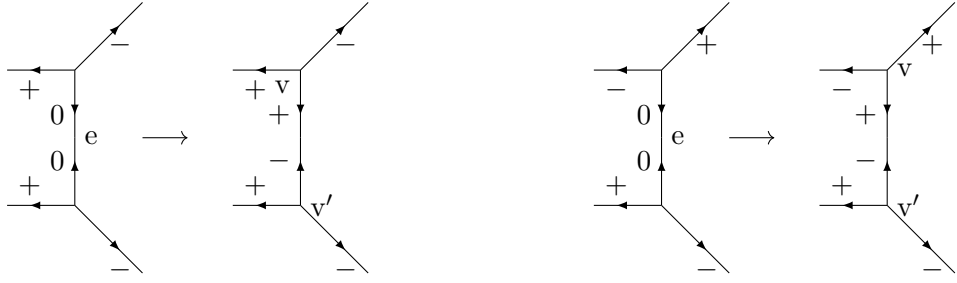

The balancing condition~\eqref{eq:balancing_differential} implies that an e-zero is either a simple or a double e-zero. These e-zeros are different and they will play different roles for us. Indeed, let~$e\in E(\mathcal A_t)$ have adjacent vertices~$\mathrm v$ and~$\mathrm v'$. It follows from Lemmas~\ref{lem:possible_zeros} and~\ref{lem:line_e-zeros} that if~$e$ is an e-zero of~$\d F_t$ for some~$(u,v)$, then if we vary~$(u,v)$ slightly away from the line of Lemma~\ref{lem:line_e-zeros},~$\mathrm v$ and~$\mathrm v'$ become v-zeros with respect to some of their adjacent components. The choice of the adjacent components depends on the type of e-zero. Assuming~$e\in \mathcal A_{t,\mu}\cap \mathcal A_{t,\mu'}$, for some~$\mu,\mu'\in\mathcal N$, it follows from~\eqref{eq:balancing_differential} that the v-zeros have to be with respect to one of these components, see Figure~\ref{fig:e-zeros}. If~$e$ is a simple e-zero then the v-zeros will be with respect to different components, say~$\mathrm v$ is a v-zero with respect to~$\mathcal A_{t,\mu}$ and~$\mathrm v'$ with respect to~$\mathcal A_{t,\mu'}$. If~$e$ instead is a double e-zero then, in contrast to if~$e$ is simple, both~$\mathrm v$ and~$\mathrm v'$ are simple v-zeros of~$\d F_t$ with respect to the same component. Moreover, moving~$(u,v)$ to one side of the line implies that the vertices are v-zeros with respect to, say~$\mathcal A_{t,\mu}$, and moving~$(u,v)$ to the other side implies that they are v-zeros with respect to~$\mathcal A_{t,\mu'}$.

Let~$\mu \in\mathcal N$ and let~$\mathrm v\in \mathcal A_{t,\mu}$ be a simple v-zero of~$\d F_t$ with respect to~$\mathcal A_{t,\mu}$. Since~$\d F_t$ is the slope of~$F_t$, it follows from the definition that the function~$F_t$ attains a local maximum or minimum at~$\mathrm v$ as a function on~$\mathcal A_{t,\mu}$. Similarly, if~$e\in \mathcal A_{t,\mu}$ is a simple e-zero, then~$F_t$ attains a local maximum or minimum along the edge~$e$, as a function on~$\mathcal A_{t,\mu}$. It is therefore natural to define~$Z_\mu=Z_\mu(u,v)$ as the number of simple v-zeros of~$\d F_t$ with respect to~$\mathcal A_{t,\mu}$ plus the number of simple e-zeros in~$\mathcal A_{t,\mu}$. In this way,~$Z_\mu$ captures the number of local maxima and minima of~$F_t$ as a function on~$\mathcal A_{t,\mu}$.

Recall that~$\mathcal F$,~$\mathcal Q$ and~$\mathcal S$ form a partition of~$\mathcal N$ with~$\mathcal S$ being the interior of~$\mathcal N$,~$\mathcal F$ being the corners and~$\mathcal Q$ consisting of the points on the boundary that are not corners.
\begin{lemma}\label{lem:critical_points}
For~$(u,v)\in \RR^2\backslash \partial D_\text{Az}$ and~$\mu\in\mathcal N$, let~$Z_\mu=Z_\mu(u,v)$ be the number of simple v-zeros of~$\d F_t$ with respect to~$\mathcal A_{t,\mu}$ plus the number of simple e-zeros in~$\mathcal A_{t,\mu}$. If~$(u,v)$ is such that~$\d F_t$ has no triple v-zeros, then
\begin{equation}\label{eq:number_zero}
2\cdot \#\{\text{double e-zeros}\}+\sum_{\mu\in \mathcal N}Z_\mu=2k\ell.
\end{equation}
Moreover, 
\begin{equation}
Z_\mu \in \{1,3\}, \text{ if}\,\, \mu \in \mathcal Q,  \quad Z_\mu \in \{2,4\}, \text{ if}\,\, \mu \in \mathcal S, 
\end{equation}
and 
\begin{equation}
Z_\mu \in \{0,2\}, \text{ if}\,\, \mu\in \mathcal F, (u,v)\in D_\text{Az} \quad \text{and} \quad Z_\mu \in \{0,1\}, \text{ if}\,\, \mu\in \mathcal F, (u,v)\in \RR^2\backslash \overline{D_\text{Az}}.
\end{equation}
In particular,
\begin{equation}
2k\ell-2\leq \sum_{\mu\in \mathcal N}Z_\mu\leq 2k\ell, \text{ if } (u,v)\in D_\text{Az} \quad \text{and} \quad \sum_{\mu\in \mathcal N}Z_\mu= 2k\ell \text{ if } (u,v)\in \RR^2\backslash \overline{D_\text{Az}}.
\end{equation}
\end{lemma}
\begin{proof}
By Lemma~\ref{lem:possible_zeros}, each vertex is a simple v-zero with respect to one component or it is adjacent to an e-zero. A simple e-zero is counted twice in the sum in~\eqref{eq:number_zero} since any edge is contained in~$\mathcal A_\mu\cap \mathcal A_{\mu'}$ for some~$\mu\neq \mu'$ (cf. Remark~\ref{rem:adjacent_components}). This implies that the left hand side is equal to the number of vertices of~$\mathcal A_t$. Furthermore, the number of vertices of~$\mathcal A_t$ is equal to the number of triangles in~$N_S(P_t)$, and by Definition~\ref{def:t-curve} each triangle has area~$1/2$ while the area of the Newton polygon is~$k\ell$, so the number of triangles is~$2k\ell$. This yields~\eqref{eq:number_zero}.

If~$\mu\in \mathcal S$, then~$\mathcal A_{t,\mu}$ is a simple loop. As discussed after Lemma~\ref{lem:possible_zeros}, the number~$Z_\mu$ counts the number of local maxima and minima of~$F_t$ as a function of~$\mathcal A_{t,\mu}$. By continuity of~$F_t$,~$Z_\mu$ is even and~$Z_\mu \geq 2$. If~$\mu\in \mathcal Q$, the residues of~$\d F_t$, that is, the slopes of~$F_t$ at the two adjacent leaves, are equal, see~\eqref{eq:action_function_residues}. It follows, again by continuity of~$F_t$, that~$Z_\mu$ is odd and~$Z_\mu\geq 1$. Similarly, if~$(u,v)\in D_\text{Az}$ and~$\mu \in \mathcal F$, then the residues of~$\d F_t$ at the two adjacent leaves have different signs so~$Z_\mu$ is even (possibly zero). If~$(u,v)\in \RR^2\backslash \overline{D_\text{Az}}$ instead, then there are two~$\mu\in \mathcal F$ such that the residues of~$\d F_t$ at the two adjacent leaves have the same sign, so~$S_\mu$ is odd. This implies that
\begin{equation}
\sum_{\mu\in \mathcal N}Z_\mu\geq 2|\mathcal S|+|\mathcal Q|=2k\ell-2, \quad \text{and} \quad \sum_{\mu\in \mathcal N}Z_\mu\geq 2|\mathcal S|+|\mathcal Q|+2=2k\ell,
\end{equation} 
if~$(u,v)\in D_\text{Az}$ and~$(u,v)\in \RR^2\backslash \overline{D_\text{Az}}$, respectively. The upper bound on the possible value of~$Z_\mu$ follows from~\eqref{eq:number_zero} and the above inequalities.
\end{proof}
The previous statement shows that if there are no triple v-zeros, then there is at most one~$\mu \in\mathcal N$ such that~$Z_\mu$ is not equal to its minimal allowed value. If~$Z_\mu$ is equal to its minimal allowed value for all~$\mu\in\mathcal N$, then there is a single double e-zero. With this, we are ready to define the macroscopic regions in the tropical limit of the Aztec diamond. These are defined in terms of the values of~$Z_\mu$.  
\begin{definition}\label{def:regions}
For~$(u,v)\in D_\text{Az}$ such that~$\d F_t$ has no triple v-zeros, we say that~$(u,v)$ is in the \emph{frozen phase} corresponding to~$\mu\in \mathcal F\cup \mathcal Q$ if~$\mu \in \mathcal F$ and~$Z_\mu=2$ or~$\mu \in \mathcal Q$ and~$Z_\mu=3$. We say that~$(u,v)$ is in the \emph{smooth phase} corresponding to~$\mu \in \mathcal S$ if~$Z_\mu=4$ and~$\d F_t$ has no triple v-zeros. If~$(u,v)$ is neither in the frozen nor the smooth region, we say that~$(u,v)$ is in the \emph{tropical arctic curve}. We denote the frozen or smooth phase corresponding to~$\mu$ by~$R_\mu$.
\end{definition}
These phases are correctly defined, in the sense that the sets~$R_\mu$,~$\mu \in\mathcal N$, are disjoint. This follows from Lemma~\ref{lem:critical_points}. 
 
\begin{remark}\label{rem:phases_appearing}
The definition indicates that in contrast to the generic situation in the finite temperature regime~\cite[Corollary 4.12]{BB23}, there may not exist a (macroscopic) phase corresponding to any~$\mu \in\mathcal N$, that is, some~$R_\mu$ might be empty even if the interior of~$\mathcal A_{t,\mu}$ is non-empty. Indeed, if~$\mathcal A_{t,\mu}$ with~$\mu \in \mathcal F$ contains only two edges (which must then be leaves), then it does not correspond to a macroscopic region in~$D_\text{Az}$, since the number of edges gives an upper bound for~$Z_\mu$. Similarly, if~$\mu \in \mathcal Q$ and~$\mathcal A_{t,\mu}$ contains only three edges, or if~$\mu\in \mathcal S$ and~$\mathcal A_{t,\mu}$ contains only three edges,~$R_\mu$ must be empty.
\end{remark}

In the finite temperature regime, it was natural to define the map~$\Log_\beta \circ \, \Omega_\beta$ from the closure of the rough region to the amoeba, see Section~\ref{sec:finite_beta_results} and~\cite[Definition 4.8]{BB23}. In the zero-temperature limit, the rough region, as well as the interior of the amoeba, vanishes and such a map does not exist. As its replacement, we define a map from the vertices of the tropical curve to points in the Aztec diamond. We will see that the image of this map consists of the vertices of the tropical arctic curves.
\begin{definition}\label{def:vertex_map}
We define the map~$\map_t:V(\mathcal A_t)\to \overline{D_\text{Az}}$, from the vertices of~$\mathcal A_t$ to the closure of the Aztec diamond, as follows. If~$\mathrm v\in V(\mathcal A_t)$, we define~$\map_t(\mathrm v)=(u,v)\in \overline D_\text{Az}$, where~$(u,v)$ is the unique point such that~$\mathrm v$ is a triple v-zero.
\end{definition}
The map~$\map_t$ is well-defined, as shown in Proposition~\ref{prop:vertex_map_well_defined} below.
\begin{remark}
The map~$\map_t$ should be thought of as the tropical limit of~$\left(\Log_\beta\circ\,\Omega_\beta\right)^{-1}$ rather than the tropical limit of~$\Log_\beta\circ\,\Omega_\beta$ itself.
\end{remark}

The map~$\map_t$ can be expressed in terms of the regular function~$f_t$. Indeed, since any vertex~$\mathrm v\in V(\mathcal A_t)$ is of degree~$3$, the balancing condition~\eqref{eq:balancing_function} means that in some neighborhood of~$\mathrm v$ the graph of~$f_t$ is contained in a plane~$\Pi=\Pi(\mathrm v)$ in~$\RR^3$. We define~$\d_x f_t(\mathrm v)$ and~$\d_y f_t(\mathrm v)$ as the slopes of~$\Pi$ in the~$x$ and~$y$ directions, that is,
\begin{equation}
\d_x f_t(\mathrm v)=f_t(\mathrm v+(1,0))-f_t(\mathrm v), \quad \text{and} \quad \d_y f_t(\mathrm v)=f_t(\mathrm v+(0,1))-f_t(\mathrm v),
\end{equation} 
where the value of~$f_t$ away from~$\mathcal A_t$ is the value from the plane~$\Pi$. Note that for an outward pointing primitive vector~$\eta$ at~$\mathrm v$, 
\begin{equation}\label{eq:direction_derivative}
\d f_t(\eta)=(\d_x f_t(\mathrm v),\d_y f_t(\mathrm v))\cdot \eta,
\end{equation}
using the dot product notation. We define~$\d_x F_t$ and~$\d_y F_t$ similarly.

\begin{proposition}\label{prop:vertex_map_well_defined}
The function~$\map_t$, from Definition~\ref{def:vertex_map}, is well-defined. Moreover,
\begin{equation}
\map_t(\mathrm v)=\frac{1}{k\ell}(\d_y f_t(\mathrm v),-\d_x f_t(\mathrm v))-\frac{1}{k\ell}(k,\ell).
\end{equation} 
\end{proposition}
\begin{proof}
The condition that~$\mathrm v$ is a triple v-zero is equivalent to~$\d_x F_t(\mathrm v;u,v)=\d_y F_t(\mathrm v;u,v)=0$. Hence, by~\eqref{eq:action_function_tropical},~$(u,v)=\map_t(\mathrm v)$ if and only if
\begin{equation}
-\ell(1+kv)-\d_x f_t(\mathrm v)=0, \quad \text{and} \quad k(1+\ell u)-\d_y f_t(\mathrm v)=0.
\end{equation}
Solving for~$u$ and~$v$ provides the formula for~$\map_t$.

 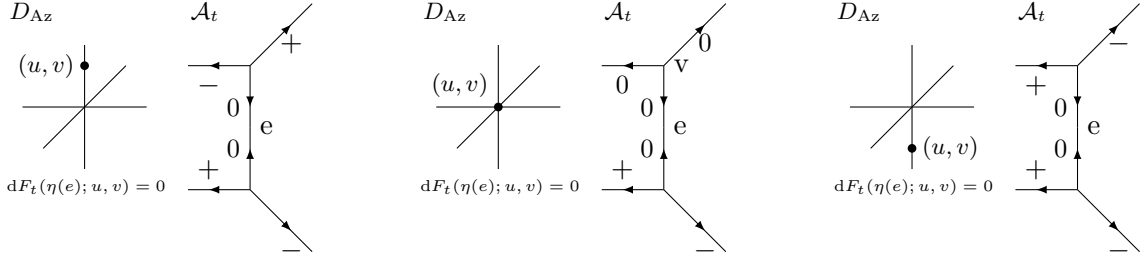
\begin{figure}[t]
 \begin{center}
\begin{tikzpicture}[scale=.55]

\begin{scope}
\draw (-1,-1) -- (1,1);
\draw (1.5,0) -- (-1.5,0);
\draw (0,1.5) -- (0,-1.5) node[below] {\tiny{$\d F_t(\eta(e);u,v)=0$}};
\fill (0,1) circle (3pt);
\draw (0,1) node[left] {\footnotesize{$(u,v)$}};
\draw (-.5,2.3) node[left] {\footnotesize{$D_\text{Az}$}};
\end{scope}

\begin{scope}[xshift=4cm, yshift=1cm]
\draw[-latex] (0,0) -- (1,1) node[below] {$+$};
\draw (.9,.9) -- (1.5,1.5);
\draw[-latex] (0,0) -- (-1,0) node[below] {$-$};
\draw (-.9,0) -- (-1.5,0);
\draw[-latex] (0,0) -- (0,-1.) node[left] {\small{$0$}};
\draw (0,-.9) -- (0,-1.5);
\draw[-latex] (0,-3) -- (1,-4) node[below] {$-$};
\draw (.9,-3.9) -- (1.5,-4.5);
\draw[-latex] (0,-3) -- (-1,-3) node[above] {$+$};
\draw (-.9,-3) -- (-1.5,-3);
\draw[-latex] (0,-3) -- (0,-2) node[left] {\small{$0$}};
\draw (0,-2.1) -- (0,-1.5);
\draw (0,-1.5) node[right] {$\mathrm e$};
\draw (-.5,1.3) node[left] {\footnotesize{$\mathcal A_t$}};
\end{scope}

\begin{scope}[xshift=10cm]
\draw (-1,-1) -- (1,1);
\draw (1.5,0) -- (-1.5,0);
\draw (0,1.5) -- (0,-1.5) node[below] {\tiny{$\d F_t(\eta(e);u,v)=0$}};
\fill (0,0) circle (3pt);
\draw (0,0) node[above left] {\footnotesize{$(u,v)$}};
\draw (-.5,2.3) node[left] {\footnotesize{$D_\text{Az}$}};
\end{scope}

\begin{scope}[xshift=14cm, yshift=1cm]
\draw[-latex] (0,0) -- (1,1) node[below] {\small{$0$}};
\draw (.9,.9) -- (1.5,1.5);
\draw[-latex] (0,0) -- (-1,0) node[below] {\small{$0$}};
\draw (-.9,0) -- (-1.5,0);
\draw[-latex] (0,0) -- (0,-1.) node[left] {\small{$0$}};
\draw (0,-.9) -- (0,-1.5);
\draw (0,0) node[right] {$\mathrm v$};
\draw[-latex] (0,-3) -- (1,-4) node[below] {$-$};
\draw (.9,-3.9) -- (1.5,-4.5);
\draw[-latex] (0,-3) -- (-1,-3) node[above] {$+$};
\draw (-.9,-3) -- (-1.5,-3);
\draw[-latex] (0,-3) -- (0,-2) node[left] {\small{$0$}};
\draw (0,-2.1) -- (0,-1.5);
\draw (0,-1.5) node[right] {$\mathrm e$};
\draw (-.5,1.3) node[left] {\footnotesize{$\mathcal A_t$}};
\end{scope}

\begin{scope}[xshift=20cm]
\draw (-1,-1) -- (1,1);
\draw (1.5,0) -- (-1.5,0);
\draw (0,1.5) -- (0,-1.5) node[below] {\tiny{$\d F_t(\eta(e);u,v)=0$}};
\fill (0,-1) circle (3pt);
\draw (0,-1) node[right] {\footnotesize{$(u,v)$}};
\draw (-.5,2.3) node[left] {\footnotesize{$D_\text{Az}$}};
\end{scope}

\begin{scope}[xshift=24cm, yshift=1cm]
\draw[-latex] (0,0) -- (1,1) node[below] {$-$};
\draw (.9,.9) -- (1.5,1.5);
\draw[-latex] (0,0) -- (-1,0) node[below] {$+$};
\draw (-.9,0) -- (-1.5,0);
\draw[-latex] (0,0) -- (0,-1.) node[left] {\small{$0$}};
\draw (0,-.9) -- (0,-1.5);
\draw[-latex] (0,-3) -- (1,-4) node[below] {$-$};
\draw (.9,-3.9) -- (1.5,-4.5);
\draw[-latex] (0,-3) -- (-1,-3) node[above] {$+$};
\draw (-.9,-3) -- (-1.5,-3);
\draw[-latex] (0,-3) -- (0,-2) node[left] {\small{$0$}};
\draw (0,-2.1) -- (0,-1.5);
\draw (0,-1.5) node[right] {$\mathrm e$};
\draw (-.5,1.3) node[left] {\footnotesize{$\mathcal A_t$}};
\end{scope}
\end{tikzpicture}
 \end{center}
\caption{The figure illustrates how the zeros of~$\d F_t$ change as~$(u,v)$ varies along the line~$\d F_t(\eta;u,v)=0$. The zeros transition from a double e-zero (left), to a triple v-zero (middle), and finally to a simple e-zero (right).
\label{fig:slope_zero_vary}}
\end{figure}

What remains is to prove that~$(u,v)\in \overline{D_\text{Az}}$. Assume for a contradiction that~$(u,v)\notin \overline{D_\text{Az}}$ and~$\d F_t$ has a triple v-zero at~$\mathrm v$. Let~$e$ and~$e'$ be adjacent edges to~$\mathrm v$, oriented away from~$\mathrm v$. We may assume (by moving to neighboring vertices/edges if necessary) that~$e$ is adjacent to one vertex at which~$\d F_t$ does not have a triple v-zero, meaning that~$\d F_t$ is non-zero on an edge adjacent to~$e$. The two lines~$\d F_t(\eta(e);u,v)=0$ and~$\d F_t(\eta(e');u,v)=0$ in the~$(u,v)$-plane have tangent vectors~$\eta(e)$ and~$\eta(e')$, respectively, see Lemma~\ref{lem:line_e-zeros}, so they intersect only at the point~$\map_t(\mathrm v)$. If we vary~$(u,v)$ slightly along the first line, then~$\d F_t(\eta(e');u,v)\neq 0$ and~$e$ turns into an e-zero. Moreover, if we vary~$(u,v)$ in one direction, then~$e$ is a simple e-zero, while it is a double e-zero if we vary~$(u,v)$ in the other direction. See Figure~\ref{fig:slope_zero_vary}. However, there are no double e-zeros for~$(u,v)\in \RR^2\backslash \overline{D_\text{Az}}$, according to Lemma~\ref{lem:critical_points}. Hence,~$(u,v)\in \overline{D_\text{Az}}$.
\end{proof}
The map~$\map_t$ is not, in general, a bijection. If, for instance,~$\mathcal A_{t,\mu}$, with~$\mu \in \mathcal Q$, only contains three edges, then if one of the vertices is a triple v-zero, so is the other vertex, and both are mapped to the same point. Similarly, if~$\mu\in \mathcal S$ and~$\mathcal A_{t,\mu}$ contain three edges, then the three vertices of the triangle are mapped to the same point. Indeed, if the exact~$1$-form~$\d F_t$ vanishes on two sides of the triangle, it must also vanish on the third one. Recall also, that in these situations~$\mathcal A_{t,\mu}$ does not correspond to a macroscopic region in the Aztec diamond, as discussed in Remark~\ref{rem:phases_appearing}.

The map~$\map_t$ can also be used to describe the tropical arctic curves.
\begin{theorem}\label{thm:arctic_curve_tropical}
Given two adjacent vertices~$\mathrm v,\mathrm v'\in V(\mathcal A_t)$, consider the (possibly empty) line segment in~$\overline{D_\text{Az}}$ with endpoints~$\map_t(\mathrm v)$ and~$\map_t(\mathrm v')$. The tropical arctic curve is the union of all such line segments with the union taken over all pairs of adjacent vertices of~$\mathcal A_t$. 

The line segment with endpoints at~$\map_t(\mathrm v)$ and~$\map_t(\mathrm v')$ is a subset of the line~$\d F_t(\eta(e);u,v)=0$ where~$e$ is the edge connecting~$\mathrm v$ and~$\mathrm v'$.
\end{theorem}
\begin{proof}
If~$(u,v)\in D_\text{Az}$ is such that~$\d F_t$ has only simple v-zeros and simple e-zeros, then Lemma~\ref{lem:critical_points} implies that there exists a single~$\mu\in\mathcal N$ such that~$(u,v)\in R_\mu$. In particular,~$(u,v)$ is not in the tropical arctic curve. 

If~$(u,v)\in D_\text{Az}$ is such that~$\d F_t$ has a double e-zero, then the same lemma implies that~$(u,v)$ lies in the tropical arctic curve. Let~$\mathrm v, \mathrm v' \in V(\mathcal A_t)$ be adjacent to an edge~$e\in E(\mathcal A_t)$, oriented from, say,~$\mathrm v$ to~$\mathrm v'$. The solutions of the equation~$\d F_t(\eta(e);u,v)=0$ form a line in the~$(u,v)$-plane. For~$(u,v)$ in this line,~$e$ is either a simple or double e-zero or at least one of~$\mathrm v, \mathrm v'$ is a triple v-zero. In particular,~$\map_t(\mathrm v)$ and~$\map_t(\mathrm v')$ are contained in this line, and they divide the line into three parts. If~$\map_t(\mathrm v)=\map_t(\mathrm v')$, then the middle part has zero length. On neighboring parts,~$e$ changes from a simple e-zero to a double e-zero, or the other way around (if the length is zero this change happens twice at the same point), see Figure~\ref{fig:slope_zero_vary}. In other words, along this line, the e-zero at~$e$ is (simple, double, simple), or (double, simple, double). However, since~$\d F_t$ cannot have a double e-zero if~$(u,v)\in \RR^2\backslash \overline{D_\text{Az}}$, by Lemma~\ref{lem:critical_points}, the order has to be (simple, double, simple). That is, the part of the line~$\d F_t(\eta(e);u,v)=0$ that is contained in the tropical curve is the line segment with endpoints at~$\Omega(\mathrm v)$ and~$\Omega(\mathrm v')$.
\end{proof}

\subsection{The zero-temperature limit}\label{sec:tropical_limit_results}
In this section, we leverage the results for finite~$\beta$ established in~\cite{BB23}, specifically those discussed in Section~\ref{sec:finite_beta_results}, to obtain results in the large~$\beta$ limit. To ensure the assumption of Section~\ref{sec:finite_beta_results}, that~$\mathcal A_{\beta,\mu}$ exists for all~$\mu\in \mathcal N$ for large enough~$\beta$, it is sufficient to assume that~$\mathcal A_t$ is a smooth tropical curve (cf. Definition~\ref{def:t-curve}), which, indeed, is our running assumption in this section. the fact that this is sufficient follows from the convergence~\eqref{eq:limit_amoeba}.

To connect the finite~$\beta$ results with the large~$\beta$ limit, we first show that the macroscopic regions of Definition~\ref{def:finite_region} converge to those of Definition~\ref{def:regions}. Recall that the macroscopic region in the finite temperature corresponding to~$\mu\in\mathcal N$ is denoted by~$R_{\beta,\mu}$ and the region in the zero-temperature limit corresponding to~$\mu\in\mathcal N$ is denoted by~$R_\mu$. 
\begin{theorem}\label{thm:tropical_limit_arctic_curve}
Let~$(u,v)\in D_\text{Az}$ be in the region~$R_\mu$ with~$\mu \in \mathcal F\cup \mathcal Q\cup \mathcal S$. There is a~$\beta_0=\beta_0(u,v)$ such that~$(u,v)\in R_{\beta,\mu}$ if~$\beta>\beta_0$.\footnote{While convergence of the arctic curves themselves might be an interesting avenue for future investigation, this is not the focus of this paper.}
\end{theorem}
\begin{proof}
If~$(u,v)\in R_\mu$, then Lemma~\ref{lem:critical_points} and Definition~\ref{def:regions} implies that~$Z_\mu=2$ if~$\mu\in \mathcal F$,~$Z_\mu=3$ if~$\mu \in \mathcal Q$, and~$Z_\mu=4$ if~$\mu \in \mathcal S$. The statement will follow from by showing that if~$\beta$ is large enough, then~$Z_{\beta,\mu}=Z_\mu$. The inequality~$Z_{\beta,\mu}\leq Z_\mu$ follows from~\eqref{eq:zeros_dfbeta}. The opposite inequality follows from~\eqref{eq:limit_harmonic_amoeba} as we will prove momentarily.

Let~$R=\{r_{p,j}\}$,~$p\in L(\mathcal R_\beta)$ and~$j=1,2,3$, be the following collection of real numbers: 
\begin{align}
 r_{p,1}&=1, & r_{p,2}&=0, & r_{p,3}&=-\ell, \quad &\text{for } p\in L_1(\mathcal R_\beta), \\
 r_{p,1}&=0, & r_{p,2}&=1, & r_{p,3}&=k, \quad &\text{for } p\in L_2(\mathcal R_\beta), \\
 r_{p,1}&=-1, & r_{p,2}&=0, & r_{p,3}&=0, \quad &\text{for } p\in L_3(\mathcal R_\beta), \\
 r_{p,1}&=0, & r_{p,2}&=-1, & r_{p,3}&=0, \quad &\text{for } p\in L_4(\mathcal R_\beta).  
\end{align}
As in Section~\ref{sec:tropical_limit}, we let~$\omega_{\beta,j}$ be the imaginary normalized differential with residues~$r_{p,j}$,~$p\in L(\mathcal R_\beta)$, and let~$\mathrm w_j$ be the exact 1-form with residues~$r_{p,j}$ at~$e_p\in L(\mathcal A_t)$, for~$j=1,2,3$, see Sections~\ref{sec:finite_spectral_curve} and~\ref{sec:tropical_functions}. By construction, 
\begin{equation}
\beta^{-1}\omega_{\beta,1}=\d \log_{\e^\beta}z, \quad \beta^{-1}\omega_{\beta,2}=\d \log_{\e^\beta}w, \quad \text{and} \quad \beta^{-1}\omega_{\beta,3}=\d \log_{\e^\beta}(f_\beta),
\end{equation}
and similarly,
\begin{equation}
\mathrm w_1=\d x, \quad \mathrm w_2=\d y, \quad \text{and} \quad \mathrm w_3=\d f_t.
\end{equation}
In particular,
\begin{equation}
\Log_{\beta,R}(q_\beta)=\left(\re\left(\int_{p_\beta}^{q_\beta}\d \log_{\e^\beta}z\right),\re\left(\int_{p_\beta}^{q_\beta}\d \log_{\e^\beta}w\right),\re \left(\int_{p_\beta}^{q_\beta}\d \log_{\e^\beta}(f_\beta)\right)\right),
\end{equation}
and
\begin{equation}
\Log_{t,R}(q_t)=\left(\int_{\mathrm v_t}^{q_t} \d x,\int_{\mathrm v_t}^{q_t} \d y,\int_{\mathrm v_t}^{q_t}\d f_t\right),
\end{equation}
where the left hand sides of the previous two equalities are defined in~\eqref{eq:log_map_generalized} and~\eqref{eq:log_map_generalized_tropical}, respectively. 

The result of~\cite[Theorem 1]{Lan20}, that is,~\eqref{eq:limit_harmonic_amoeba} tells us that we can take~$\Log_\beta p_\beta\to \mathrm v_t$ so that~$\Log_{\beta,R}(\mathcal R_\beta)\to \Log_{t,R}(\mathcal A_t)$ in Hausdorff distance on compact subsets, as~$\beta \to \infty$. Hence, the graph
\begin{equation}
\mathcal R_\beta^\circ\ni (z,w)\mapsto \left(\log_{\e^\beta}|z|,\log_{\e^\beta}|w|,\log_{\e^\beta}|f_\beta(z,w)|\right)
\end{equation}
converges in Hausdorff distance on compact subsets to the graph
\begin{equation}
\mathcal A_t\ni (x,y)\mapsto \left(x,y,f_t(x,y)\right),
\end{equation}
as~$\beta \to \infty$. Recall that in the definitions of~$f_t$ and~$\log_{\e^\beta}|f_\beta|$ we did not specify the arbitrary additive constant, and for the convergence above to be true we make sure to match those constants (for example, by matching the values at~$\mathrm v_t$ and~$p_\beta$). 

Since~$(u,v)\in R_\mu$, the function~$F_t$ has~$Z_\mu$ maxima and minima as a function on~$\mathcal A_{t,\mu}$, recall the discussion just before Lemma~\ref{lem:critical_points}. Moreover,
\begin{equation}
\beta^{-1}\re F_\beta(z,w;u,v)=k(1+\ell u)\log_{\e^\beta}|w|-\ell(1+kv)\log_{\e^\beta}|z|-\log_{\e^\beta}|f_\beta(z,w)|,
\end{equation}
and
\begin{equation}
F_t(x,y;u,v)=k(1+\ell u)y-\ell(1+kv)x-f_t(x,y),
\end{equation}
so the above convergence implies that there is a~$\beta_0=\beta_0(u,v)$ such that if~$\beta>\beta_0$, then~$\re F_\beta$ has at least~$Z_\mu$ critical points as a function on~$A_{\beta,\mu}$, that is,~$\d F_\beta$ has at least~$Z_\mu$ zeros in~$A_{\beta,\mu}$. Hence,~$Z_{\beta,\mu}\geq Z_\mu$. 
\end{proof}

With the previous theorem in our hands, we can use the results described in Section~\ref{sec:finite_beta_results} to obtain statements about the zero-temperature limit of the dimer model. We discuss below the zero-temperature limit of the limit shape and local fluctuations.

\begin{definition}\label{def:discrete_curve}
Let~$\mu_0=(0,k)\in \mathcal F$. For~$\mu \in \mathcal F\cup \mathcal Q\cup \mathcal S$ consider an oriented simple path~$\gamma_\mu$ in the dual graph of~$\mathcal A_t$ going from~$\mu$ to~$\mu_0$. The set~$\Gamma_\mu\subset LE(\mathcal A_t)$ consists of the edges crossed by~$\gamma_\mu$, oriented so that they cross~$\gamma_\mu$ from left to right.
\end{definition}
A slightly different perspective on the previous definition is the following. Pick a path along the edges of~$N_S(P_t)$ starting at~$\mu$ and ending at~$\mu_0$. Then~$\Gamma_\mu$ consists of the edges in~$\mathcal A_t$ corresponding, through the duality discussed in Section~\ref{sec:tropical_amoeba}, to the edges in that path in~$N_S(P_t)$. See the left image in Figure~\ref{fig:discrete_curve} for an example of the path~$\gamma_\mu$ and the set~$\Gamma_\mu$.

 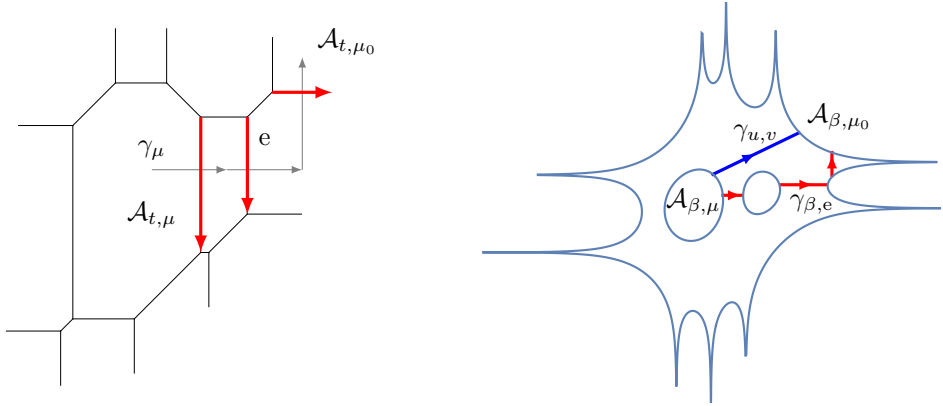
\begin{figure}[t]
 \begin{center}
\begin{subfigure}[c]{0.35\textwidth}
    \begin{tikzpicture}[scale=1]
    \draw (0,0) node {\includegraphics[scale=.25]{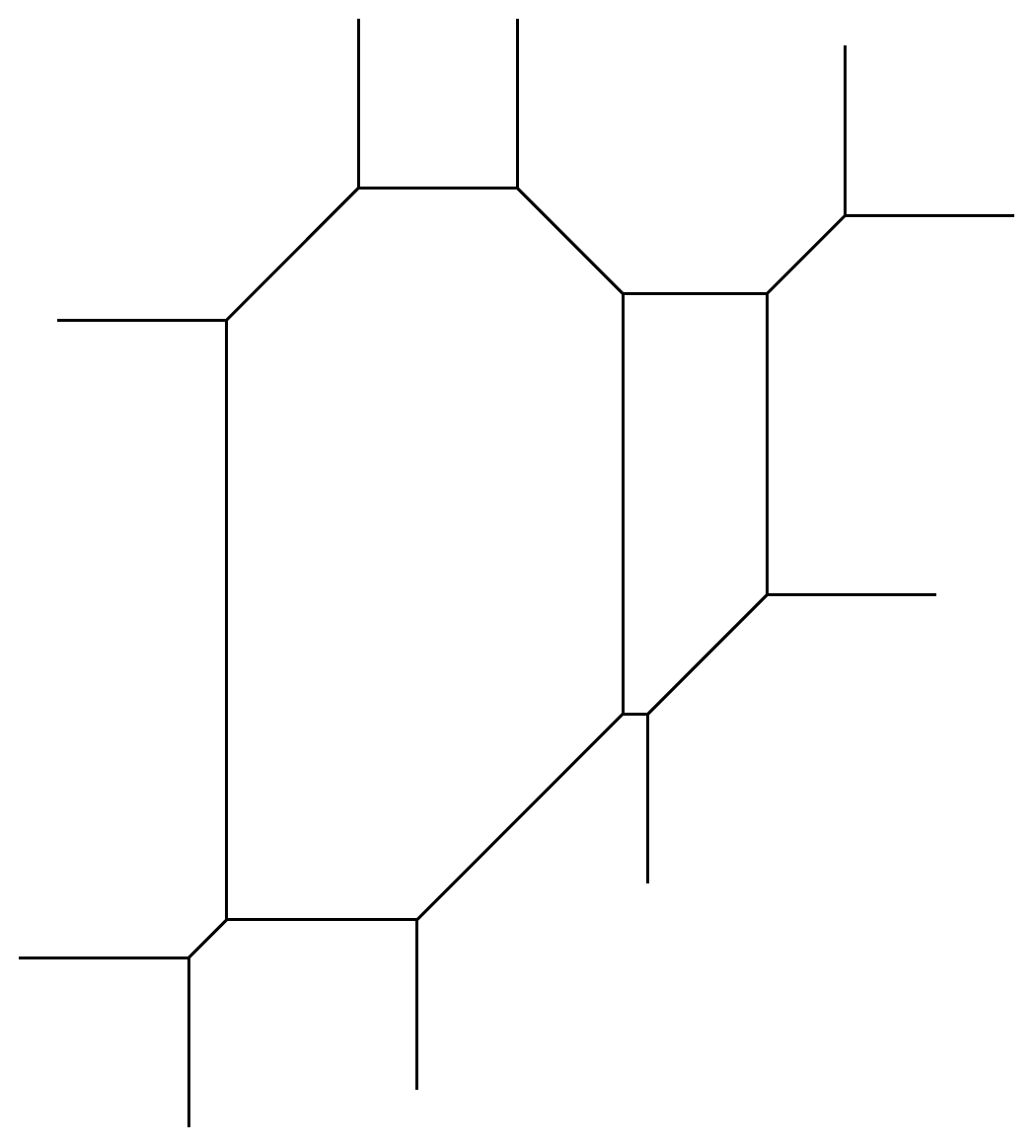}};
    \draw[-latex, gray] (1.8,.5) -- (1.8,2);
    \draw[-latex, gray] (.8,.5) -- (1.8,.5);
    \draw[-latex, gray] (-.2,.5) -- (.8,.5);

    \draw[-latex, very thick, red] (.45,1.2) -- (.45,-.6);
    \draw[-latex, very thick, red] (1.07,1.2) -- (1.07,-.1);
    \draw[-latex, very thick, red] (1.4,1.53) -- (2.2,1.53);

    \draw (2.4,2.2) node{\small{$\mathcal A_{t,\mu_0}$}};
    \draw (-.2,-.1) node{\small{$\mathcal A_{t,\mu}$}};
    \draw (1.3,.9) node{\small{$\mathrm e$}};
    \draw (-.2,.5) node[above] {\small{$\gamma_\mu$}};

   \end{tikzpicture}
 \end{subfigure}
 \quad
\begin{subfigure}[c]{0.35\textwidth}
    \begin{tikzpicture}[scale=1]
    \draw[very thick, red] (1.6,.75) -- (1.6,.4);
    \draw[-latex, thick, red] (1.6,.68) -- (1.6,.7);
    \draw[very thick, red] (1.55,.3) -- (.9,.3);
    \draw[-latex, thick, red] (1.33,.3) -- (1.35,.3);
    \draw[very thick, red] (.13,.16) -- (.43,.16);
    \draw[-latex, thick, red] (.38,.16) -- (.4,.16);

    \draw[very thick, blue] (.0,.43) -- (1.18,1);
    \draw[-latex, thick, blue] (.0,.43) -- (.59,.715);
    
    \draw (1.7,1.2) node{\small{$\mathcal A_{\beta,\mu_0}$}};
    \draw (-.24,.08) node{\small{$\mathcal A_{\beta,\mu}$}};
    \draw (1.33,.3) node[below] {\small{$\gamma_{\beta,\mathrm e}$}};
    \draw (.59,.715) node[above] {\small{$\gamma_{u,v}$}};
    
    \draw (0,0) node {\includegraphics[angle=180, scale=.5, trim={0 1.8cm 0 0}, clip]{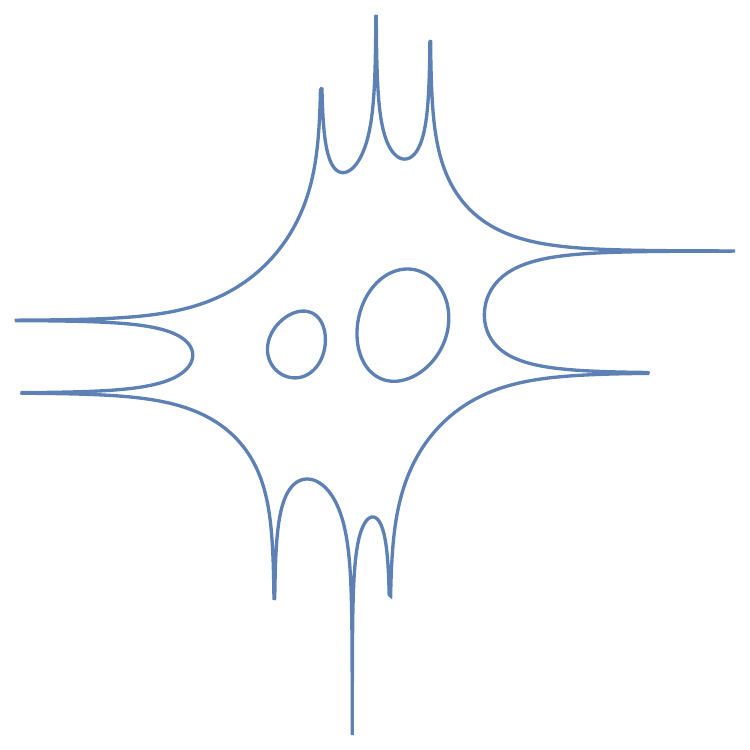}};
  \end{tikzpicture}
 \end{subfigure}
 \end{center}
\caption{Left: An example of the curve~$\gamma_\mu$ (gray) with the edges contained in~$\Gamma_\mu$ (red). Right: The curves~$\gamma_{\beta,\mathrm e}$ (red) with~$\mathrm e\in \Gamma_\mu$ from the left figure, and the curve~$\gamma_{u,v}$ (blue). For the sake of readability, the scales in the two pictures are different. 
\label{fig:discrete_curve}}
\end{figure}

\begin{corollary}\label{cor:tropical_limit_shape}
Let~$(u,v)\in R_\mu\subset D_\text{Az}$ for some~$\mu \in \mathcal F\cup \mathcal Q\cup \mathcal S$ and set
\begin{equation}
\bar h_t(u,v)=\frac{1}{k\ell}\sum_{e\in \Gamma_\mu}\d F_t(\eta(e);u,v)+1,
\end{equation}
where~$\Gamma_\mu$ is as in Definition~\ref{def:discrete_curve}. 
Then
\begin{equation}
\lim_{\beta\to\infty} \bar h_\beta(u,v)=\bar h_t(u,v),
\end{equation}
where~$\bar h_\beta$ is the limit shape for finite~$\beta$ from Theorem~\ref{thm:finite_limit_shape}.
\end{corollary}
\begin{proof}
For each~$e\in \Gamma_\mu$, let~$\gamma_{\beta,e}$ be a curve in~$\mathcal R_\beta$ going around the tube approximating the edge~$e$, as defined in the end of Section~\ref{sec:tropical_limit}. Recall that the orientation of~$e$ imposes an orientation of~$\gamma_{\beta,e}$. If~$\beta>\beta_0$ where~$\beta_0=\beta_0(u,v)$ is as in Theorem~\ref{thm:tropical_limit_arctic_curve}, then~$(u,v)\in R_{\beta,\mu}$. Hence,~$\gamma_{u,v}$ in~\eqref{eq:limit_shape} can be deformed to the union~$\cup \gamma_{\beta,e}$ where the union runs over all edges in~$\Gamma_\mu$, see Figure~\ref{fig:discrete_curve}. Hence, cf.~\eqref{eq:limit_shape},
\begin{equation}
\bar h_\beta(u,v)=\frac{1}{k\ell}\sum_{e\in \Gamma_\mu}\frac{1}{2\pi \i}\int_{\gamma_{\beta,e}}\d F_\beta+1.
\end{equation}  
The corollary now follows from~\eqref{eq:limit_b-cycles}.
\end{proof}
\begin{remark}
The definition of~$\Gamma_\mu$ depends on the path~$\gamma_\mu$. However, the balancing condition~\eqref{eq:balancing_differential} implies that the tropical limit shape~$\bar h_t$ is independent of the choice of~$\gamma_\mu$. 
\end{remark}

The function~$\bar h_t$ is piecewise linear, and Corollary~\ref{cor:tropical_limit_shape} together with~\eqref{eq:action_function_tropical} implies that for~$(u,v)\in R_\mu$ with~$\mu\in \mathcal F\cup \mathcal Q\cup \mathcal S$,
\begin{equation}\label{eq:tropical_height_function}
\bar h_t(u,v)=(u,v)\cdot\nabla h_t(u,v)+H(u,v),
\end{equation}
with slope
\begin{equation}
\nabla h_t(u,v)=\left(\sum_{e\in \Gamma_\mu} \d y(\eta(e)),-\sum_{e\in \Gamma_\mu}\d x(\eta(e))\right)=\mu-\mu_0
\end{equation}
and intercept 
\begin{equation}
H(u,v)=\frac{1}{k\ell}(k,\ell)\cdot(\mu-\mu_0)+1-\frac{1}{k\ell}\sum_{e\in \Gamma_\mu}\d f_t(\eta(e)).
\end{equation}
We will see an alternative (but closely related) representation of the height function in Proposition~\ref{prop:vertex_map_gradient} below. 
\begin{remark}
For~$(u,v)$ in the boundary between two regions, say,~$R_\mu$ and~$R_{\mu'}$, Theorem~\ref{thm:arctic_curve_tropical} tells us that~$\d F_t(\eta(e);u,v)=0$ where~$e\in \mathcal A_{t,\mu}\cap \mathcal A_{t,\mu'}$. This shows that~$\bar h_t$, as given in Corollary~\ref{cor:tropical_limit_shape}, is continuous. Conversely, we know that the limit shape in the pre-limit, as given in~\eqref{eq:limit_shape}, is continuous with uniformly bounded gradient, uniform both in~$(u,v)$ and~$\beta$. Hence, its limit~$\bar h_t$ must be continuous, and we recover, from Corollary~\ref{cor:tropical_limit_shape}, the fact that~$\d F_t(\eta(e);u,v)=0$ for~$(u,v)$ on the boundary of~$R_\mu$ and~$R_{\mu'}$, and~$e\in \mathcal A_{t,\mu}\cap \mathcal A_{t,\mu'}$.
\end{remark}

The local correlations for the Aztec diamond dimer model were obtained in~\cite{BB23}. Applying those results to our setting gives us the following corollary. Recall that for~$\beta>0$ and~$(x,y)\in \RR^2$,~$\PP_{\text{Az},\beta}$ is the probability measure for the Aztec diamond and~$\PP_{(x,y),\beta}$ is the ergodic translation-invariant Gibbs measure, see Section~\ref{sec:measures}. 
\begin{corollary}\label{cor:local_limit}
Let~$(u,v)\in D_\text{Az}$ be in~$R_\mu$ for some slope~$\mu \in \mathcal F\cup \mathcal Q\cup \mathcal S$. There is a~$\beta_0=\beta_0(u,v)$, such that if~$\beta>\beta_0$ then
\begin{equation}
\PP_{\text{Az},\beta}\to \PP_{(x,y),\beta},
\end{equation}
weakly as~$N\to \infty$, where~$(x,y)\in \RR^2$ is any point in~$\mathcal A_{\beta,\mu}$.   
\end{corollary}
\begin{proof}
Let~$(u,v)\in R_\mu$ and~$\beta_0$ be as in Theorem~\ref{thm:tropical_limit_arctic_curve}. Then~$(u,v)\in \mathcal R_{\beta,\mu}$ if~$\beta>\beta_0$ and Theorem~\ref{thm:finite_local_limit} proves the result.
\end{proof}


\subsection{Further properties of the tropical arctic curve}\label{sec:arctic_curve_properties}
In this section, we explore the properties of the tropical arctic curve. Similar features of the arctic curve for finite~$\beta$ were derived in~\cite{BB23}. It is therefore likely that some of the statements below could be derived using the knowledge from the finite temperature regime. However, it is not necessary and in this section, we will only rely on the definitions given in Section~\ref{sec:tropical_action}.

The following claim readily follows from Proposition~\ref{prop:vertex_map_well_defined}.
\begin{corollary}
The map~$\map_t$ maps the vertices connected with leaves to the boundary of~$D_\text{Az}$. More precisely, let~$\mathrm v_i\in V(\mathcal A_t)$ be connected to an edge~$e_i\in L_i(\mathcal A_t)$,~$i=1,2,3,4$. Then~$\map_t(\mathrm v_i)$ is on the top, right, bottom, left boundary of~$D_\text{Az}$ for~$i=1,2,3,4$, respectively.
\end{corollary} 
\begin{proof}
If~$e\in L_1(\mathcal A_t)$ is adjacent to~$\mathrm v\in V(\mathcal A_t)$ and directed towards~$\mathrm v$, then, by definition of~$\d f_t$ and~\eqref{eq:direction_derivative},~$\d_x f_t(\mathrm v)=\d f_t(\eta(e))=-\ell$. Thus, by Proposition~\ref{prop:vertex_map_well_defined},
\begin{equation}
\map_t(\mathrm v)=\frac{1}{k\ell}(\d_y f_t(\mathrm v) ,\ell)-\frac{1}{k\ell}(k,\ell)=\frac{1}{k\ell}(\d_y f_t(\mathrm v)-k,0),
\end{equation}
where the right hand side is on the top boundary of~$D_\text{Az}$, recall~\eqref{eq:aztec_scaled}. The other parts of~$\partial D_\text{Az}$ follow similarly.
\end{proof}

 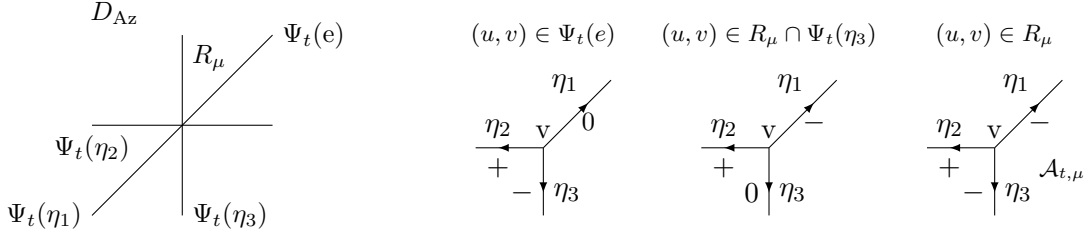
\begin{figure}[t]
 \begin{center}
\begin{tikzpicture}[scale=.6]
\begin{scope}[yshift=.5cm]
\draw (-2,-2) -- (2,2);
\draw (2,2) node[right] {\small{$\map_t(\mathrm e)$}};
\draw (-2,-2) node[left] {\small{$\map_t(\eta_1)$}};
\draw (2,0) -- (-2,0);
\draw (-2,0) node[below] {\small{$\map_t(\eta_2)$}};
\draw (0,2) -- (0,-2);
\draw (0,-2) node[right] {\small{$\map_t(\eta_3)$}};
\draw (-1.5,2.5) node {\small{$D_\text{Az}$}};
\draw (.6,1.5) node {\small{$R_\mu$}};
\end{scope}

\begin{scope}[xshift=8cm]
\draw[-latex] (0,0) -- (1,1) node[below] {\small{$0$}};
\draw (1,1) node[above left] {$\eta_1$};
\draw (.9,.9) -- (1.5,1.5);
\draw[-latex] (0,0) -- (-1,0) node[below] {$+$};
\draw (-1,0) node[above] {$\eta_2$};
\draw (-.9,0) -- (-1.5,0);
\draw[-latex] (0,0) -- (0,-1.) node[left] {$-$};
\draw (0,-1) node[right] {$\eta_3$};
\draw (0,-.9) -- (0,-1.5);
\draw (0,0) node[above] {$\mathrm v$};
\draw (0,2.5) node {\footnotesize{$(u,v)\in \map_t(e)$}};
\end{scope}

\begin{scope}[xshift=13cm]
\draw[-latex] (0,0) -- (1,1) node[below] {$-$};
\draw (1,1) node[above left] {$\eta_1$};
\draw (.9,.9) -- (1.5,1.5);
\draw[-latex] (0,0) -- (-1,0) node[below] {$+$};
\draw (-1,0) node[above] {$\eta_2$};
\draw (-.9,0) -- (-1.5,0);
\draw[-latex] (0,0) -- (0,-1.) node[left] {$0$};
\draw (0,-1) node[right] {$\eta_3$};
\draw (0,-.9) -- (0,-1.5);
\draw (0,0) node[above] {$\mathrm v$};
\draw (0,2.5) node {\footnotesize{$(u,v)\in R_\mu \cap \map_t(\eta_3)$}};
\end{scope}

\begin{scope}[xshift=18cm]
\draw[-latex] (0,0) -- (1,1) node[below] {$-$};
\draw (1,1) node[above left] {$\eta_1$};
\draw (.9,.9) -- (1.5,1.5);
\draw[-latex] (0,0) -- (-1,0) node[below] {$+$};
\draw (-1,0) node[above] {$\eta_2$};
\draw (-.9,0) -- (-1.5,0);
\draw[-latex] (0,0) -- (0,-1.) node[left] {$-$};
\draw (0,-1) node[right] {$\eta_3$};
\draw (0,-.9) -- (0,-1.5);
\draw (0,0) node[above] {$\mathrm v$};
\draw (0,2.5) node {\footnotesize{$(u,v)\in R_\mu$}};
\draw (1.5,-.5) node {\footnotesize{$\mathcal A_{t,\mu}$}};
\end{scope}
\end{tikzpicture}
 \end{center}
\caption{Moving~$(u,v)$ in~$R_\mu$ between the lines~$\map_t(\mathrm e)$ and~$\map_t(\eta_3)$ (see Lemma~\ref{lem:local_orientation} and its proof for definitions) only changes the sign of~$\d F_t$ at one edge connected to~$\mathrm v$ at a time. 
\label{fig:swapping_orientation}}
\end{figure} 

Theorem~\ref{thm:arctic_curve_tropical} together with Lemma~\ref{lem:line_e-zeros} show that~$\map_t$ preserves the tangent lines, in the sense that~$\mathrm v'-\mathrm v$ and~$\map_t(\mathrm v')-\map_t(\mathrm v)$ are parallel. At the same time, the map~$\map_t$ reverses the orientation in the following way.
\begin{lemma}\label{lem:local_orientation}
Let~$\mathrm v,\mathrm v'\in V(\mathcal A_t)$ be connected by an edge, and consider the oriented edge~$e$ going from~$\mathrm v$ to~$\mathrm v'$. Assume~$\map_t(\mathrm v)\neq \map_t(\mathrm v')$ and let~$\map_t(e)$ be the oriented line segment going from~$\map_t(\mathrm v)$ to~$\map_t(\mathrm v')$. If~$R_\mu$ for some~$\mu\in \mathcal F\cup \mathcal Q\cup \mathcal S$ lies to the left of~$\map_t(e)$, then the interior of~$\mathcal A_{t,\mu}$ lies to the right of~$e$. 
\end{lemma}
\begin{proof}
Recall that~$e$ is a double e-zero of~$\d F_t$ if~$(u,v)$ is in the interior of the line segment~$\map_t(e)$, see the proof of Theorem~\ref{thm:arctic_curve_tropical}. This implies, as discussed just after Lemma~\ref{lem:line_e-zeros}, that if we move~$(u,v)$ away from~$\map_t(e)$, then~$\mathrm v$ and~$\mathrm v'$ are simple v-zeros of~$\d F_t$ with respect to one of the adjacent components of~$e$ (and if we vary~$(u,v)$ to the other side of~$\map_t(e)$ they are simple v-zeros with respect to the other component). Lemma~\ref{lem:possible_zeros}, Lemma~\ref{lem:critical_points} and Definition~\ref{def:regions} imply that if~$(u,v)\in R_\mu$, then~$\mathrm v$ and~$\mathrm v'$ are simple v-zeros of~$\d F_t$ with respect to~$\mathcal A_{t,\mu}$. In particular, the interior of~$\mathcal A_{t,\mu}$ is directly to the right or directly to the left of~$e$. 

Let~$\eta_1$,~$\eta_2$ and~$\eta_3$ be the outward pointing primitive vectors at~$\mathrm v$ with~$\eta_1=\eta(e)$ and so that the order, as we go around~$\mathrm v$ in positive direction, is~$\eta_1$,~$\eta_2$,~$\eta_3$. Let~$\map_t(\eta_i)$ be the lines in the~$(u,v)$-plane defined by the equations~$\d F_t(\eta_i;u,v)=0$,~$i=1,2,3$. These lines have, according to Lemma~\ref{lem:line_e-zeros}, tangent vectors~$\eta_i$. The balancing condition~\eqref{eq:balancing_primitive} implies that~$\eta_{i+1}$ and~$\eta_{i+2}$ (with subscripts taken modulo~$3$) lie on different sides of the line~$\map_t(\eta_i)$,~$i=1,2,3$. Consequently, as we circle around~$\map_t(\mathrm v)$ in positive direction, the lines appear in the order~$\map_t(\eta_1)$,~$\map_t(\eta_3)$,~$\map_t(\eta_2)$,~$\map_t(\eta_1)$,~$\map_t(\eta_3)$,~$\map_t(\eta_2)$.  

Hence, we conclude that for~$(u,v)\in R_\mu$ in a neighborhood of~$\map_t(\mathrm v)$ between~$\map_t(e)$ (recall that~$\map_t(e)$ is a subset of the line~$\map_t(\eta_1)$) and~$\map_t(\eta_3)$, the vertex~$\mathrm v$ is a simple v-zero with respect to~$\mathcal A_{t,\mu}$ and the interior of~$\mathcal A_{t,\mu}$ lies just to the right or left of~$e$. If the interior of~$\mathcal A_{t,\mu}$ lies to the left of~$e$, then, by definition of simple v-zeros,~$\d F_t(\eta_1;u,v)$ and~$\d F_t(\eta_2;u,v)$ have the same sign, and if~$\mathcal A_{t,\mu}$ lies to the right, then~$\d F_t(\eta_1;u,v)$ and~$\d F_t(\eta_3;u,v)$ have the same sign. Since~$(u,v)$ lies between~$\map_t(e)$ and~$\map_t(\eta_3)$, taking the ordering of the lines into account, we see that taking~$(u,v)\in \map_t(\eta_3)$ does not change the sign of~$\d F_t(\eta_1;u,v)$ or~$\d F_t(\eta_2;u,v)$, while~$\d F_t(\eta_3;u,v)=0$. The balancing condition~\eqref{eq:balancing_primitive} then shows that~$\d F_t(\eta_1;u,v)$ and~$\d F_t(\eta_2;u,v)$ must have different signs, and~$\mathcal A_{t,\mu}$ cannot be to the left of~$e$, hence, it is on right. See Figure~\ref{fig:swapping_orientation}. 
\end{proof}

 \begin{figure}[t]
 \begin{center}
\begin{tikzpicture}[scale=1]
\begin{scope}[yshift=.8cm]
\draw (0,0) -- (1,1) node[left] {$\mathrm v_1$};
\draw (0,0) -- (-1,0) node[above] {$\mathrm v_2$};
\draw (0,0) -- (0,-1) node[left] {$\mathrm v_3$};
\draw (0,0) node[above] {$\mathrm v$};
\draw (0,-.3) arc[start angle=-90,end angle=45,radius=.3] node[midway, right] {\footnotesize{$\theta_{\mathrm v}'$}};
\draw (1,-.7) node {\footnotesize{$\mathcal A_{t,\mu}$}};
\end{scope}

\begin{scope}[xshift=5cm]
\draw (0,0) -- (1.5,1.5) node[below right] {\footnotesize{$\map_t(\mathrm v_1)$}};
\draw (0,0) -- (-1.5,0) node[above] {\footnotesize{$\map_t(\mathrm v_2)$}};
\draw (0,0) -- (0,1.5) node[left] {\footnotesize{$\map_t(\mathrm v_3)$}};
\draw (0,0) node[below] {\footnotesize{$\map_t(\mathrm v)$}};
\draw (0,.5) arc[start angle=90,end angle=45,radius=.5];
\draw (.3,.7) node {\footnotesize{$\theta_{\mathrm v}$}};
\draw (.6,1.5) node {\footnotesize{$R_{t,\mu}$}};
\end{scope}

\begin{scope}[xshift=9cm, yshift=1cm]
\draw (0,0) -- (1.5,1.5) node[below right] {\footnotesize{$\map_t(\mathrm v_1)$}};
\draw (0,0) -- (1.5,0) node[above] {\footnotesize{$\map_t(\mathrm v_2)$}};
\draw (0,0) -- (0,-1.5) node[left] {\footnotesize{$\map_t(\mathrm v_3)$}};
\draw (0,0) node[below right] {\footnotesize{$\map_t(\mathrm v)$}};
\draw (0,-.3) arc[start angle=270,end angle=45,radius=.3] node[midway, left] {\footnotesize{$\theta_{\mathrm v}$}};
\draw (.3,1) node {\footnotesize{$R_{t,\mu}$}};
\end{scope}

\begin{scope}[xshift=13cm]
\draw (0,0) -- (1.5,1.5) node[below right] {\footnotesize{$\map_t(\mathrm v_1)$}};
\draw (0,0) -- (1.5,0) node[above] {\footnotesize{$\map_t(\mathrm v_2)$}};
\draw (0,0) -- (0,1.5) node[left] {\footnotesize{$\map_t(\mathrm v_3)$}};
\draw (0,0) node[below] {\footnotesize{$\map_t(\mathrm v)$}};
\draw (0,.5) arc[start angle=90,end angle=45,radius=.5];
\draw (.3,.7) node {\footnotesize{$\theta_{\mathrm v}$}};
\draw (.6,1.5) node {\footnotesize{$R_{t,\mu}$}};
\end{scope}
\end{tikzpicture}
 \end{center}
\caption{The possible, up to a rotation by~$\pi$, directions of~$\map_t(\mathrm v_i)-\map_t(\mathrm v)$,~$i=1,2,3$, relative to the direction of~$\mathrm v_i-\mathrm v$. The angle at~$\mathrm v$ in~$\mathcal A_{t,\mu}$ and the angle at~$\map_t(\mathrm v)$ in~$R_\mu$ are related by~$\theta_{\mathrm v}=n_{\mathrm v}\pi-\theta_{\mathrm v}'$ with~$n_{\mathrm v}\in \{1,2\}$.
\label{fig:local_orientation}}
\end{figure}
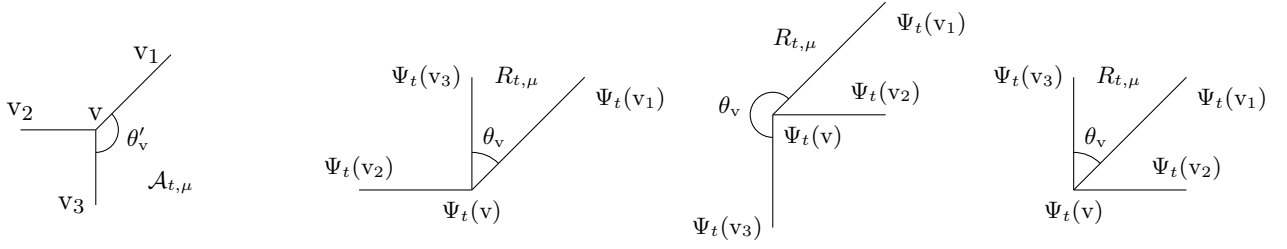

Below we present two consequences of the above observation. 
\begin{corollary}\label{cor:orientation}
For a vertex~$\mathrm v\in V(\mathcal A_t)$, let~$\mathrm v_i\in V(\mathcal A_t)$,~$i=1,2,3$, be adjacent to~$\mathrm v$ and let~$\eta_i$ be the respective outward pointing primitive vectors. We label them so that~$i\mapsto \eta_i$ is rotating in positive direction as~$i$ is increasing. Assume that~$\map_t(\mathrm v)\neq \map_t(\mathrm v_i)$ for~$i=1,2,3$. Then~$i\mapsto \map_t(\mathrm v_i)-\map_t(\mathrm v)$ is rotating in negative direction. In particular, there exists a line going through~$\map_t(\mathrm v)$ so that~$\map_t(\mathrm v_i)$,~$i=1,2,3$, are all on the same side of that line.
\end{corollary}
\begin{proof}
Let the interior of~$\mathcal A_{t,\mu}$ be to the right of~$\eta_1$, then it is to the left of~$\eta_2$. By Lemma~\ref{lem:local_orientation},~$R_\mu$ is to the right of~$\map_t(\mathrm v_1)-\map_t(\mathrm v)$ and to the left of~$\map_t(\mathrm v_2)-\map_t(\mathrm v)$. Hence, the orientation is swapped. The existence of the line now follows from the balancing condition~\eqref{eq:balancing_primitive}. See Figure~\ref{fig:local_orientation}.
\end{proof}

The following corollary has its counterpart before the large~$\beta$ limit, see~\cite[Corollary 4.17]{BB23}.
\begin{corollary}\label{cor:inner_angles}
Let~$\mu\in \mathcal F\cup \mathcal Q\cup \mathcal S$ be such that~$R_\mu$ is non-empty and consider the interior angle~$\theta_{\mathrm v}$ at~$\map_t(\mathrm v)$ in~$R_\mu$ as well as the interior angle~$\theta_{\mathrm v}'$ at~$\mathrm v$ in~$\mathcal A_{t,\mu}$. If~$\map_t(\mathrm v)=\map_t(\mathrm v')$, we take~$\theta_{\mathrm v}'$ to be the sum of the angles at~$\mathrm v$ and~$\mathrm v'$ in~$\mathcal A_{t,\mu}$. Then~$\theta_{\mathrm v}=n_{\mathrm v}\pi-\theta_{\mathrm v}'$ where~$n_{\mathrm v}\in \{1,2\}$. Moreover, if~$\mu \in \mathcal S$, then~$n_{\mathrm  v}=1$ for four vertices~$\mathrm v\in\mathcal A_{t,\mu}$ and~$n_{\mathrm v}=2$ otherwise. Similarly, if~$\mu \in \mathcal Q$ and~$\mu\in \mathcal F$, then~$n_{\mathrm v}=1$ for three and two vertices~$\mathrm v$, respectively, and~$n_{\mathrm v}=2$ otherwise. See Figures~\ref{fig:local_orientation} and~\ref{fig:inner_angles}.
\end{corollary}
\begin{remark}
By the balancing condition~\eqref{eq:balancing_primitive},~$\theta_{\mathrm v}'\in (0,\pi)$, so~$\theta_{\mathrm v}\in (0,\pi)$ if~$n_{\mathrm v}=1$, and~$\theta_{\mathrm v}\in (\pi,2\pi)$ if~$n_{\mathrm v}=2$.
\end{remark}
\begin{proof}
The first part of the statement~$\theta_{\mathrm v}=n_{\mathrm v}\pi-\theta_{\mathrm v}'$ with~$n_{\mathrm v}\in \{1,2\}$ follows from Lemma~\ref{lem:local_orientation} and the fact that~$\mathrm v'-\mathrm v$ and~$\map_t(\mathrm v')-\map_t(\mathrm v)$ are parallel, see Figure~\ref{fig:local_orientation}.

For the second part of the statement we first assume that~$\mu\in \mathcal S$. The sum of the angles of the polygons~$\mathcal A_{t,\mu}$ and~$R_\mu$ satisfy 
\begin{equation}\label{eq:inner_angles}
\sum_{\mathrm v}(\pi-\theta_{\mathrm v})=2\pi, \quad \text{and} \quad \sum_{\mathrm v}(\pi-\theta_{\mathrm v}')=2\pi.
\end{equation} 
Combining these equations with~$\theta_{\mathrm v}=n_{\mathrm v}\pi-\theta_{\mathrm v}'$ implies 
\begin{equation}
2\pi=-2\pi+\pi \sum_{\mathrm v}(2-n_\mathrm v),
\end{equation}
which proves the statement for~$\mu\in \mathcal S$.

If~$\mu\in \mathcal Q$ instead, the angles~$\theta_{\mathrm v}'$ are in the unbounded component~$\mathcal A_{t,\mu}$ and the right hand side of the equality involving~$\theta_{\mathrm v}'$ in~\eqref{eq:inner_angles} is~$\pi$ instead of~$2\pi$. The right hand side of the equality for~$\theta_{\mathrm v}$ is still~$2\pi$. Similarly, if~$\mu\in \mathcal F$, then the right hand side of the equality for~$\theta_{\mathrm v}'$ in~\eqref{eq:inner_angles} is~$\pi/2$ instead of~$2\pi$ and there is an additional term~$\pi/2$ on the left hand side of the second equality, to compensate for the corner of the Aztec diamond. The result then follows as in the case of~$\mu\in \mathcal S$.
\end{proof}

 \begin{figure}[t]
 \begin{center}
\begin{subfigure}[c]{0.45\textwidth}
    \begin{tikzpicture}[scale=1]
    \draw (0,0) node {\includegraphics[scale=.3]{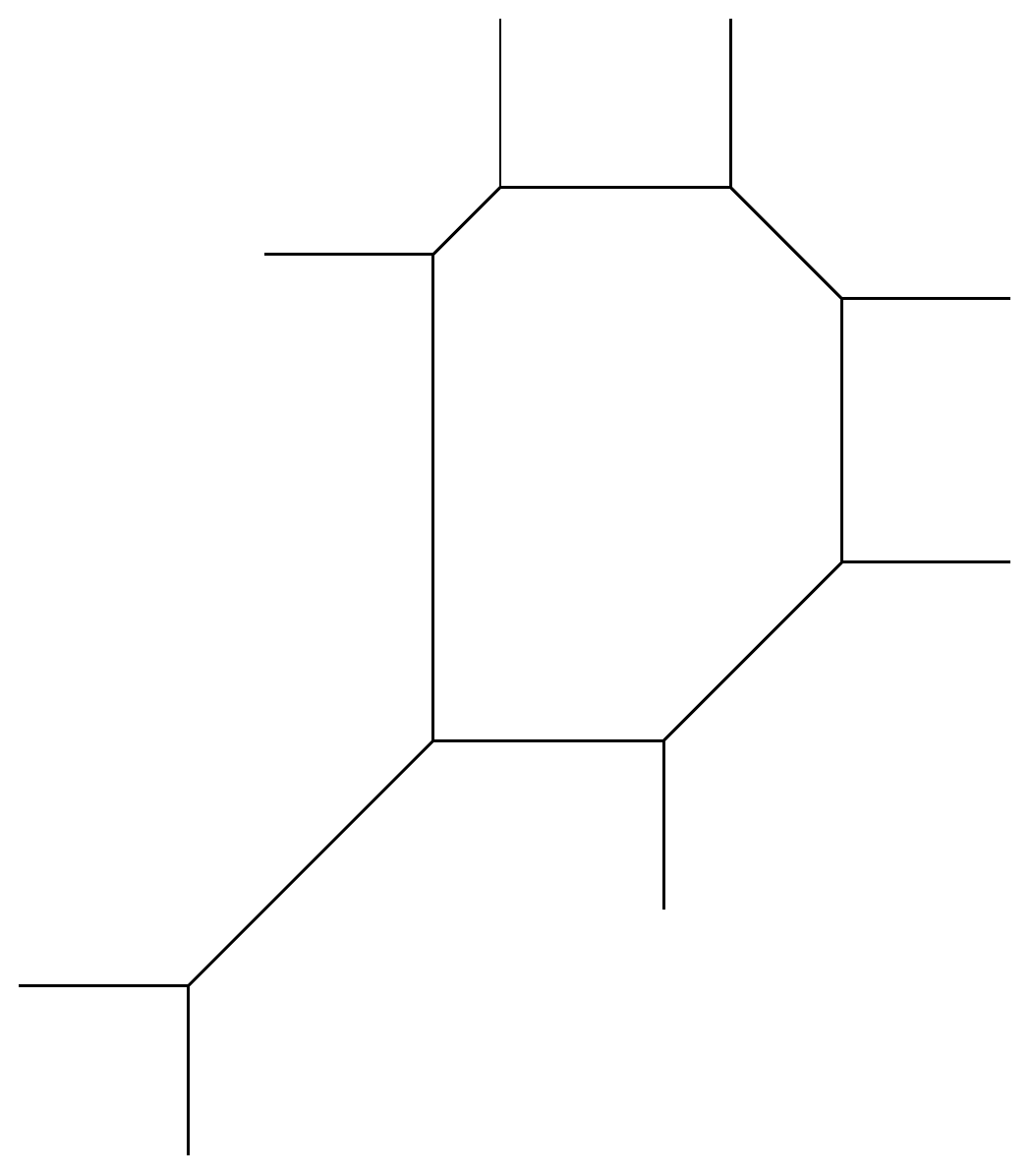}};
    \draw (1.11,2.065) circle (5pt);
    \draw (1.11,2.065) node[below left] {\tiny{$\theta_{61}'$}};
    \draw (1.11,2.065) node[above right] {\tiny{$\theta_{11}'$}};
    
    \draw (-.075,2.065) circle (5pt);
    \draw (-.075,2.065) node[above left] {\tiny{$\theta_{52}'$}};
    \draw (-.075,2.065) node[below right] {\tiny{$\theta_{61}'$}};
    
    \draw (-.425,1.715) circle (5pt);
    \draw (-.425,1.715) node[above left] {\tiny{$\theta_{51}'$}};
    \draw (-.425,1.715) node[below left] {\tiny{$\theta_{43}'$}};
    \draw (-.425,1.715) node[below right] {\tiny{$\theta_{65}'$}};
    
    \draw (-.425,-.795) circle (5pt);
    \draw (-.425,-.795) node[above left] {\tiny{$\theta_{42}'$}};
    \draw (-.425,-.795) node[below right ] {\tiny{$\theta_{32}'$}};
    \draw (-.425,-.795) node[above right] {\tiny{$\theta_{64}'$}};

    \draw (-1.7,-2.07) circle (5pt);
    \draw (-1.7,-2.07) node[above left ] {\tiny{$\theta_{41}'$}};
    \draw (-1.7,-2.07) node[below right] {\tiny{$\theta_{33}'$}};

    \draw (.76,-.795) circle (5pt);
    \draw (.76,-.795) node[above left] {\tiny{$\theta_{63}'$}};
    \draw (.76,-.795) node[below left] {\tiny{$\theta_{31}'$}};
    \draw (.76,-.795) node[below right] {\tiny{$\theta_{22}'$}};

    \draw (1.695,.14) circle (5pt);
    \draw (1.695,.14) node[above left] {\tiny{$\theta_{62}'$}};
    \draw (1.695,.14) node[below right] {\tiny{$\theta_{21}'$}};

    \draw (1.695,1.49) circle (5pt);
    \draw (1.695,1.49) node[above right] {\tiny{$\theta_{12}'$}};
    \draw (1.695,1.49) node[below left] {\tiny{$\theta_{62}'$}};
 	\pic[scale=.2, rotate=-45] at (3.5,2.5) {compassOp};
  \end{tikzpicture}
 \end{subfigure}
 \quad
\begin{subfigure}[c]{0.45\textwidth}
    \begin{tikzpicture}[scale=1]
    \draw (0,0) node {\includegraphics[scale=.3]{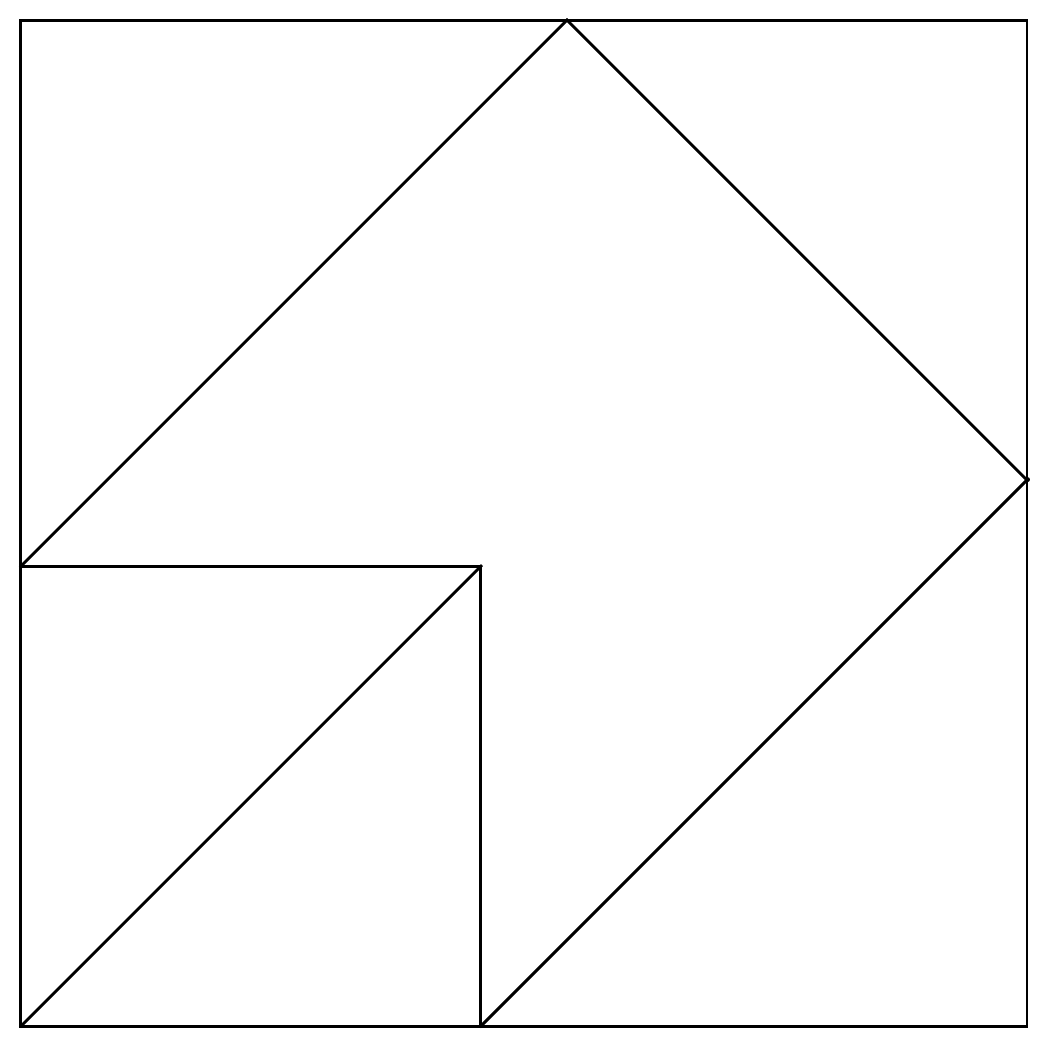}};
    
    \draw (2.05,2.55) arc[start angle=180,end angle=270,radius=.5];
                
    \draw (-.29,2.55) arc[start angle=180,end angle=360,radius=.5];
    \draw (-.25,2.55) node[below left] {\footnotesize{$\theta_{21}$}};
    \draw (.21,2.1) node[below] {\footnotesize{$\theta_{62}$}};
    \draw (.7,2.55) node[below right] {\footnotesize{$\theta_{12}$}};

    \draw (-2.05,2.55) arc[start angle=360,end angle=270,radius=.5];
                
    \draw (-2.55,.29) arc[start angle=90,end angle=-90,radius=.5];
    \draw (-2.55,.25) node[above right] {\footnotesize{$\theta_{22}$}};
    \draw (-2.1,-.21) node[above right] {\footnotesize{$\theta_{63}$}};
    \draw (-2.2,-.7) node[right] {\footnotesize{$\theta_{31}$}};

    \draw (-2.05,-2.55) arc[start angle=0,end angle=90,radius=.5];
    \draw (-2.05,-2.55) node[above right] {\footnotesize{$\theta_{41}$}};
    \draw (-2.55,-2.05) node[above right] {\footnotesize{$\theta_{33}$}};

    \draw (.29,-2.55) arc[start angle=0,end angle=180,radius=.5];
    \draw (.25,-2.55) node[above right] {\footnotesize{$\theta_{51}$}};
    \draw (-.21,-2.1) node[above right] {\footnotesize{$\theta_{65}$}};
    \draw (-.7,-2.55) node[above left] {\footnotesize{$\theta_{43}$}};

    \draw (2.05,-2.55) arc[start angle=180,end angle=90,radius=.5];

    \draw (2.55, -.29) arc[start angle=270,end angle=90,radius=.5];
    \draw (2.55,-.25) node[below left] {\footnotesize{$\theta_{52}$}};
    \draw (2.1,.21) node[left] {\footnotesize{$\theta_{61}$}};
    \draw (2.55,.7) node[above left] {\footnotesize{$\theta_{11}$}};
    
    \draw (-.215,-.215) circle (10pt);
    \draw (.2,.2) node {\footnotesize{$\theta_{64}$}};
    \draw (-.8,-.4) node {\footnotesize{$\theta_{32}$}};
    \draw (-.45,-.8) node {\footnotesize{$\theta_{42}$}};

 	\pic[scale=.2, rotate=-45] at (4,2.5) {compass};
  \end{tikzpicture}
 \end{subfigure}
 \end{center}
\caption{The relation between the angles in the tropical curve~$\theta_{\mathrm v}'$ and the angles in the tropical arctic curve~$\theta_{\mathrm v}$. In the figure, we write~$\theta_{\mathrm v}'=\theta_{ij}'$, where~$i$ indicates the component, and~$j$ the angle in the component. If~$\theta_{ij}'$ appears more than once, it should be interpreted as the sum of the angles that carry it. Cf. Figure~\ref{fig:local_orientation}. 
\label{fig:inner_angles}}
\end{figure}

\subsection{A dual representation of the limit shape}\label{sec:dual_representation}
In Proposition~\ref{prop:vertex_map_well_defined} and Corollary~\ref{cor:tropical_limit_shape} we expressed~$\map_t$ and~$\bar h_t$ in terms of the function~$f_t$, which is a function on~$\mathcal A_t$. In this section, we consider a type of a dual function~$f_t^*$ of~$f_t$ with which we give alternative expressions of~$\map_t$ and~$\bar h_t$.

Recall the duality between~$N_S(P_t)$ and~$\mathcal A_t$ from Section~\ref{sec:tropical_amoeba}. For a vertex~$\mathrm v$, an edge~$e$ and a face~$\mathrm f$ in~$\mathcal A_t$ we denote the corresponding face, edge and vertex by~$\mathrm v^*$,~$e^*$ and~$\mathrm f^*$, respectively. If~$e\in LE(\mathcal A_t)$ is an oriented edge, then we set~$\eta(e^*)$ to be the vector obtained by rotating~$\eta(e)$ by~$\pi/2$ in the positive direction. Given a regular function~$g_t$ on~$\mathcal A_t$, we define its dual 1-form~$(\d g_t)^*$ on the edges of~$N_S(P_t)$ by~$(\d g_t)^*(\eta(e^*))=\d g_t(\eta(e))$. The balancing condition~\eqref{eq:balancing_function} implies that 
\begin{equation}
\sum_{i=1}^3 (\d g_t)^*(\eta(e_i^*))=0,
\end{equation}
where~$e_i^*$ are the edges around a face in~$N_S(P_t)$. The 1-form~$(\d g_t)^*$ is therefore exact, and it defines, up to an additive constant, a function~$g_t^*$ on the vertices of~$N_S(P_t)$ so that~$\d (g_t^*)(\eta(e^*))=(\d g_t)^*(\eta(e^*))$.

We assign a weight~$l^*$ to the interior edges of~$N_S(P_t)$ defined by~$l^*(e^*)=l(e)$, where~$l(e)$ is the length of the edge~$e\in E(\mathcal A_t)$ as defined in Section~\ref{sec:tropical_amoeba}. The weighted discrete Laplacian~$\Delta_l$ is defined by
\begin{equation}\label{eq:balancing_dual}
(\Delta_lg_t^*)(\mathrm f^*)=\sum_{e^*\sim \mathrm f^*}l^*(e^*)\d g_t^*(\eta(e^*)),
\end{equation}
where the sum runs over all edges~$e^*$ adjacent to the vertex~$\mathrm f^*$ with~$\eta(e^*)$ oriented away from~$\mathrm f^*$. Since~$\d g_t\in \Omega_0(\mathcal A_t)$, the function~$g_t^*$ satisfies the condition~$\Delta_l(g_t^*)(\mathrm f^*)=0$ for all inner vertices~$\mathrm f^*$, cf.~\eqref{eq:line_integral}.

Let~$f_t$ be the function defined in Section~\ref{sec:tropical_action} and let~$f_t^*$ be its dual function as defined above. The boundary conditions for~$f_t$ imply the boundary conditions for~$f_t^*$, namely,~$\d f_t^*(\eta(e^*))=-\ell$ if~$e^*\in (L_1(\mathcal A_t))^*$,~$\d f_t^*(\eta(e^*))=k$ if~$e^*\in (L_2(\mathcal A_t))^*$, and~$\d f_t(\eta(e^*))=0$ if~$e^*\in (L_i(\mathcal A_t))^*$,~$i=3,4$. Here~$(L_i(\mathcal A_t))^*$ is defined so that~$e^*\in (L_i(\mathcal A_t))^*$ if and only if~$e\in L_i(\mathcal A_t)$.
\begin{remark}
Instead of starting with~$f_t$, we could have defined~$f_t^*$ (up to an additive constant) by~\eqref{eq:balancing_dual} together with its boundary conditions. By the above duality, we then recover~$f_t$.
\end{remark}

The benefit of considering~$f_t^*$ instead of~$f_t$ is that~$f_t^*$ naturally extends to a piecewise linear continuous function on~$N(P)$, see Figure~\ref{fig:newton_harmonic}. For any point~$(s,t)$ lying within a face~$\mathrm v^*$ of~$N_S(P_t)$, the gradient~$\nabla f^*(s,t)$ of (the extended function)~$f_t^*$ is constant. We denote this constant value by~$\nabla f_t^*(\mathrm v^*)$. Note that
\begin{equation}
\nabla f_t^*(\mathrm v^*)=(\d_y f_t(\mathrm v),-\d_x f_t(\mathrm v))
\end{equation}
due to the~$\pi/2$ rotation between~$e$ and~$e^*$. With this dual perspective, Proposition~\ref{prop:vertex_map_well_defined} turns into the following statement. 
\begin{proposition}\label{prop:vertex_map_gradient}
Let~$f_t^*$ be as defined above. The function~$\map_t$ from Definition~\ref{def:vertex_map} is given by
\begin{equation}
\map_t(\mathrm v)=\frac{1}{k\ell}\nabla f_t^*(\mathrm v^*)-\frac{1}{k\ell}(k,\ell).
\end{equation}  
\end{proposition}

 \begin{figure}[t]
 \begin{center}
\begin{subfigure}[c]{0.35\textwidth}
\includegraphics[scale=.3]{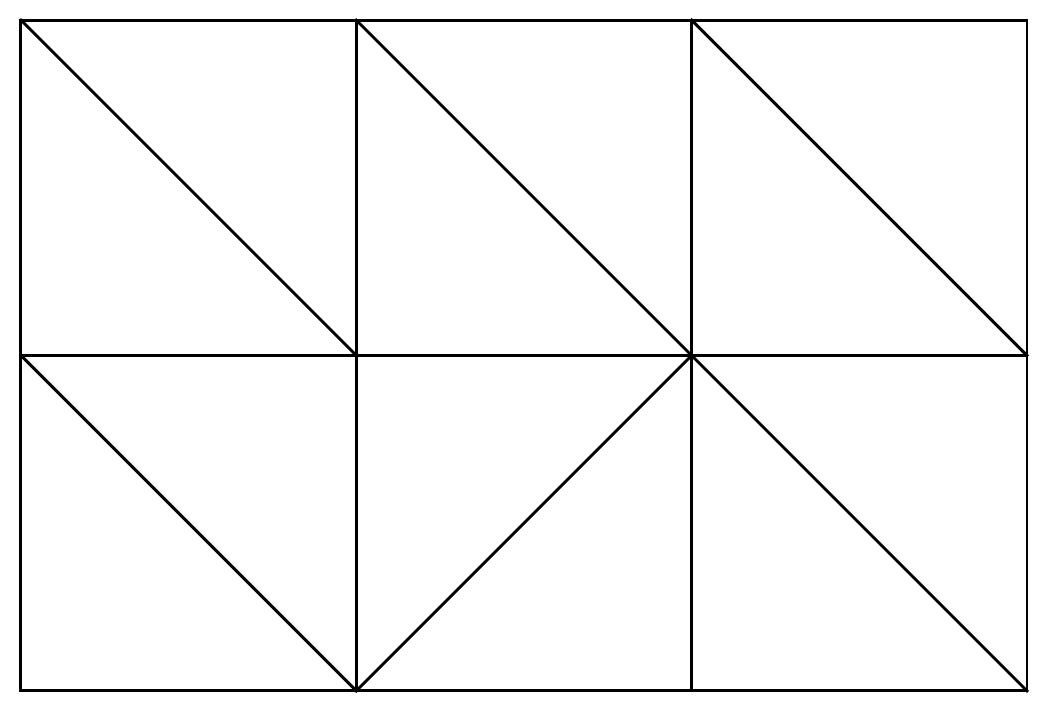}
 \end{subfigure}
\qquad \qquad
\begin{subfigure}[c]{0.35\textwidth}
\includegraphics[scale=.3,trim={10cm 0cm 10cm 0cm}, clip]{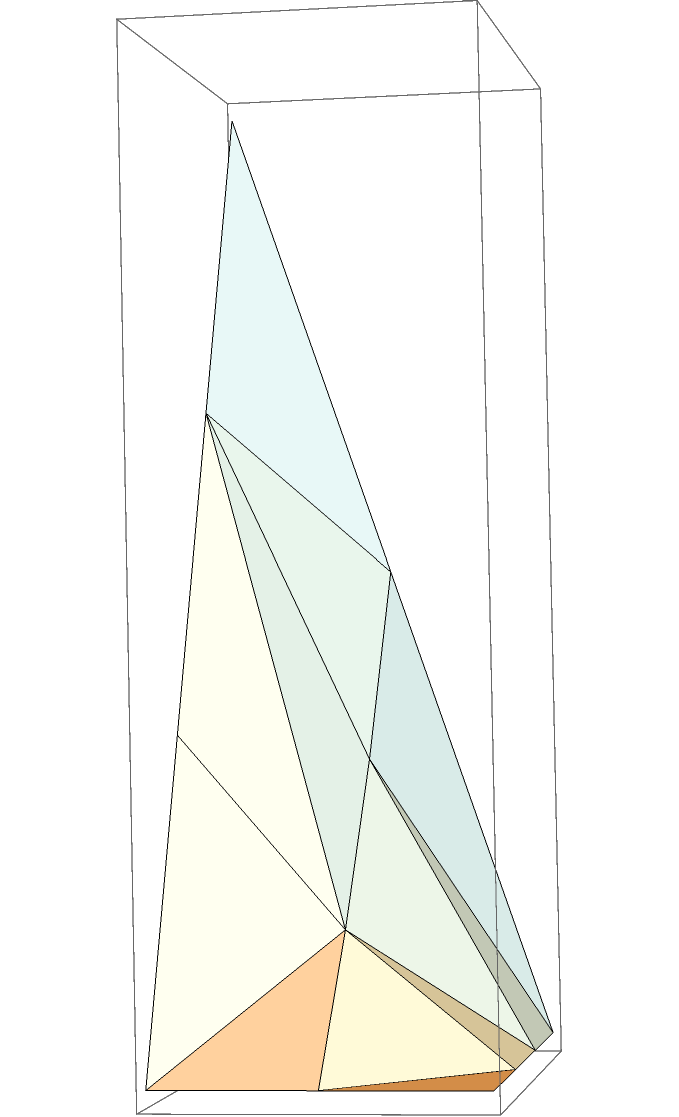}
 \end{subfigure}
 \end{center}
\caption{A subdivision~$N_S(P_t)$ of the Newton polygon (left) and the plot of the function~$f_t^*$ (right). The periodicity is~$k=2$ and~$\ell=3$.
\label{fig:newton_harmonic}}
\end{figure}

By extending the notion of the gradient, we can express not just the vertices of the tropical arctic curve as the gradient of~$f_t^*$, but the entire tropical arctic curve itself. The Clarke subdifferential~$\partial f^*(s,t)$ of the function~$f^*$ at~$(s,t)\in N(P)\subset \RR^2$ is the set
\begin{equation}
\partial f^*(s,t)=\convHull\left\{\lim_{i\to \infty}\nabla f^*(s_i,t_i):N(P)\ni(s_i,t_i)\to(s,t), \, f^* \text{ is differentiable at } (s_i,t_i)\right\},
\end{equation} 
see~\cite[Definition 1.1]{Cla75}. If~$(s,t)$ is in the face~$\mathrm v^*$, we recover the gradient of~$f^*$, namely,~$\partial f^* (s,t)=\{\nabla f^*(\mathrm v^*)\}$. If~$(s,t)$ is in the interior of an edge~$e^*$, with adjacent faces~$\mathrm v_1^*$ and~$\mathrm v_2^*$, then
\begin{equation}
\partial f^*(s,t)=\convHull\left\{\nabla f^*(v_1^*),\nabla f^*(v_2^*)\right\},
\end{equation} 
that is,~$\partial f^*(s,t)$ is the line segment between~$\nabla f^*(v_1^*)$ and~$\nabla f^*(v_2^*)$. Combining Proposition~\ref{prop:vertex_map_gradient} and Theorem~\ref{thm:arctic_curve_tropical} yields the following description of the tropical arctic curve.  
\begin{corollary}\label{cor:tropical_arctic_gradient}
The tropical arctic curve is the set
\begin{equation}
\bigcup_{(s,t)\in N(P)\backslash \mathcal N}\left(\frac{1}{k\ell} \partial f_t^*(s,t)-\frac{1}{k\ell}(k,\ell)\right).
\end{equation}  
\end{corollary}

\begin{remark}\label{rem:arctic_tropical_curve}
It is worth noting the similarities between how we obtain the tropical arctic curve and the tropical curve~$\mathcal A_t$ from~$f_t^*$ and~$-\mathcal E^*$, respectively. Let us extend~$\mathcal E^*$ to the piecewise linear function whose graph coincides with the top boundary of~$\tilde N(P_t)$, making it smooth on the faces of~$N_S(P_t)$. Explicitly, for~$(s,t)$ in a face~$\mathrm v^*$ of~$N_S(P_t)$ corresponding via the duality to a vertex~$\mathrm v=(x,y)$ in~$\mathcal A_t$, we have
\begin{equation}\label{eq:tropical_surface_tension_extended}
-\mathcal E^*(s,t)=xs+yt-P_t(x,y).
\end{equation}
Indeed, each vertex~$\mu_i$,~$i=1,2,3$, of the face~$\mathrm v^*$ attains the maximum of the right hand side of~\eqref{eq:characteristic_polynomial_tropical}, and the values of~$\mathcal E^*$ at those vertices determine the linear function~\eqref{eq:tropical_surface_tension_extended} on~$\mathrm v^*$. Thus, the vertices~$(x,y)$ of~$\mathcal A_t$ are the images of~$-\nabla \mathcal E^*$ evaluated at the faces of~$N_S(P_t)$, and the bounded edges of~$\mathcal A_t$ form the set 
\begin{equation}
\bigcup_{(s,t)\in N(P)\backslash \mathcal N}\left(\partial(-\mathcal E^*(s,t))\right),
\end{equation}  
where~$\partial(-\mathcal E^*)$ is the Clarke subdifferential of~$-\mathcal E^*$ (which, since~$-\mathcal E^*$ is convex, is the set of subgradients,~\cite[Proposition 1.2]{Cla75}).
\end{remark}

We can also express the tropical limit shape~$\bar h_t$, described by Corollary~\ref{cor:tropical_limit_shape} and below, in terms of~$f_t^*$ instead of~$f_t$.
\begin{corollary}\label{cor:limit_shape_newton}
Let~$(u,v)\in R_\mu\subset D_\text{Az}$ for some~$\mu \in \mathcal F\cup \mathcal Q\cup \mathcal S$ and let~$\mu_0=(0,k)\in \mathcal F$. Then
\begin{equation}
\bar h_t(u,v)=\left(u+\frac{1}{\ell},v+\frac{1}{k}\right)\cdot (\mu-\mu_0)+\frac{1}{k\ell}(f_t^*(\mu)-f_t^*(\mu_0))+1.
\end{equation}
\end{corollary}
\begin{proof}
This follows from~\eqref{eq:tropical_height_function} noting also that the duality together with Definition~\ref{def:discrete_curve} implies that~$\sum_{e\in \Gamma_\mu}\d f_t(\eta(e))=f_t^*(\mu_0)-f_t^*(\mu)$.
\end{proof}

It is natural to ask if~$\bar h_t$ is determined by the values~$f_t^*(\mu)$ for~$\mu\in \mathcal N$. In other words, do the equations for the planes from which~$\bar h_t$ is constructed, fully determine~$\bar h_t$? The following simple example shows that they are not sufficient. We also need the subdivision~$N_S(P_t)$, which informs us how the planes are connected in the graph of~$\bar h_t$.
\begin{example}
Let~$k=1$ and~$\ell=1$. In this case, there are no smooth regions, so~$f_t^*$ is immediately determined (up to an additive constant) from its boundary conditions. We have
\begin{equation}
f_t^*((0,1))=0, \quad f_t^*((0,0))=0, \quad f_t^*((-1,0))=1, \quad f_t^*((-1,1))=0.
\end{equation}
The function~$\bar h_t$ is continuous, and its graph consists of pieces of the planes from Corollary~\ref{cor:limit_shape_newton}. There are two choices: The graph may consist of the planes corresponding to~$\mu=(0,1)$ and~$\mu=(-1,0)$, with the tropical arctic curve being the line segment between the west and east corners of the Aztec diamond. Alternatively, the graph may also consist of the other two planes, with the tropical arctic curve being the line segment between the north and south corners of the Aztec diamond. These different cases correspond to the two possible triangulations of~$N(P)$. 

The first choice is obtained, for instance, if~$\nu(e)=a>1$ for~$e$ being any South edge and~$\nu(e)=1$ otherwise. The second choice is obtained, for instance, if~$\nu(e)=a>1$ on the West edges and~$\nu(e)=1$ otherwise.
\end{example}

\section{The zero-temperature limit of the Gibbs measures}\label{sec:zero_gibbs}
This section shifts focus from the Aztec diamond to translation-invariant Gibbs measures defined in~\eqref{eq:gibbs_beta}. Intuitively, for large values of~$\beta$, the randomness in these measures should vanish in the limit. As we will see, this is true generically. However, for specific edge weights, randomness can persist in the limit. Our primary interest lies in the limits of the measures appearing in Corollary~\ref{cor:local_limit}, even though our assumption in the previous section is slightly stronger than what is needed here (see Lemma~\ref{lem:concavity_subdivision} below).

For the purpose of this section, we extend the definition of~$\mathcal E(\mathcal D)$ and~$\mu(\mathcal D)$ from Section~\ref{sec:characteristic_polynomial}, to include any subset of~$E_1$, not just dimer covers. Given a set of edges~$D\subset E_1$, we let 
\begin{equation}\label{eq:energy_slope_2}
\mathcal E(D)=\sum_{e\in D} \log \nu(e)\in \RR_{\geq 0} \quad \text{and} \quad \mu(D)=\sum_{e\in D}(-e\wedge \gamma_u,e\wedge \gamma_v)\in \ZZ^2,
\end{equation}
cf.~\eqref{eq:energy_slope}.

We start by discussing the assumption we will impose in this section, namely that the tropical surface tension~$\mathcal E^*$, defined by~\eqref{eq:tropical_surface_tension}, is concave.
\begin{definition}\label{def:concave}
The tropical surface tension~$\mathcal E^*$ is said to be \emph{concave} at~$\mu\in\mathcal N$ if
\begin{equation}
\mathcal E^*(\mu)\geq \sum_{i=1}^n t_i\mathcal E^*(\mu_i),
\end{equation}
for all~$t_i\in (0,1)$ with~$\sum_{i=1}^nt_i=1$ and~$\sum_{i=1}^n t_i\mu_i=\mu$. If the inequality is strict, we say that~$\mathcal E^*$ is \emph{strictly concave} at~$\mu\in \mathcal N$. 
\end{definition}
In the previous section, we considered the case when~$\mathcal A_t$ is a smooth tropical curve, that is, when the subdivision of the Newton polygon~$N_S(P_t)$ is a triangulation with triangles of area~$1/2$. In particular, that means that all points in~$\mathcal N$ are vertices of~$N_S(P_t)$. The following lemma tells us that the relevant setting from the point of view of the previous section is covered under the assumption that~$\mathcal E^*$ is strictly concave.
\begin{lemma}\label{lem:concavity_subdivision}
The tropical surface tension~$\mathcal E^*$ is strictly concave at~$\mu\in\mathcal N$ if and only if~$\mu$ is a vertex of~$N_S(P_t)$. 
\end{lemma}
\begin{proof}
This follows from the definition of~$N_S(P_t)$. Indeed, the tropical surface tension~$\mathcal E^*$ is strictly concave at~$\mu$ if and only if~$(\mu,\mathcal E^*(\mu))$ is a vertex of the extended polyhedral domain~\eqref{eq:extended_polyhedral_domain}. 
\end{proof}

The limit of the Gibbs measures will depend on the most likely dimer covers on the torus. We define a subgraph of~$G_1$ consisting of those dimer covers, and we will see (Theorem~\ref{thm:gibbs_limit}) that the dimer model on the corresponding graph on the plane is the limit of the Gibbs measure.
\begin{definition}\label{def:maximizer}
A dimer cover~$\mathcal D$ of~$G_1$ is a \emph{maximizer} with slope~$\mu$ if~$\mu(\mathcal D)=\mu$ and~$\mathcal E(\mathcal D)=\mathcal E^*(\mu)$. Given~$\mu\in\mathcal N$, the un-weighted subgraph~$G_{1,\mu}$ is defined as~$G_{1,\mu}=(B_{1,\mu}, W_{1,\mu},E_{1,\mu})\subset G_1$ with~$B_{1,\mu}=B_1$,~$W_{1,\mu}=W_1$ and~$E_{1,\mu}$ consisting of all edges in~$E_1$ that lie in at least one maximizer with slope~$\mu$.
\end{definition}

It is not difficult to see (compare with the discussion in~\cite[Section 4.3.1]{BB23}) that if~$\mu\in\mathcal Q$, then~$\mathcal E^*$ is strictly concave at~$\mu$ if and only if there is a unique maximizer with slope~$\mu$. If~$\mu\in \mathcal S$, however, this is more subtle. See Section~\ref{sec:subdivision} below for a detailed discussion.

Before we define the limiting dimer models, we need two lemmas. 

The set of maximizers with slope~$\mu$ can be viewed as an \emph{edge-$d$-coloring} of a \emph{$d$-multiweb} of~$G_{1,\mu}$, where~$d$ is the number of maximizers. A~$d$-multiweb of~$G_{1,\mu}$ is a multiset of edges with degree~$d$ at each vertex of~$G_{1,\mu}$. An edge-$d$-coloring of a~$d$-multiweb is a partition of the edges such that each part is a dimer cover of~$G_{1,\mu}$. See, \emph{e.g.},~\cite{DKS24}.
\begin{lemma}\label{lem:maximizers_all}
If~$\mathcal E^*$ is strictly concave at~$\mu\in \mathcal N$, then all dimer covers of~$G_{1,\mu}$ are maximizers with slope~$\mu$. 
\end{lemma}
\begin{proof}
We begin by noting that any dimer cover~$\mathcal D$ of~$G_{1,\mu}$, is contained in an edge-$d$-coloring. Indeed, we remove the edges in~$\mathcal D$ from the~$d$-muliweb and obtain a~$(d-1)$-multiweb. This can be colored, since any~$(d-1)$-multiweb admits an edge-$(d-1)$-coloring, see~\cite[Section 3.4.1]{DKS24}. 

Let~$D_i$,~$i=1,\dots,d$, be an edge-$d$-coloring of~$G_{1,\mu}$. Then, since all edges in the~$d$-multiweb are used once, we have 
\begin{equation}\label{eq:energy_slope_coloring}
\sum_{i=1}^d \mathcal E(\mathcal D_i)=d \!\cdot\!\mathcal E^*(\mu), \quad \text{and} \quad \sum_{i=1}^d \mu(\mathcal D_i)=d\!\cdot\! \mu.
\end{equation}
This contradicts the assumption that~$\mathcal E^*$ is concave, with~$t_i=1/d$, unless,~$\mu(\mathcal D_i)=\mu$ for all~$i$. Furthermore, since~$\mu(\mathcal D_i)=\mu$, we get by definition of~$\mathcal E^*$ that~$\mathcal E(\mathcal D_i)\leq \mathcal E^*(\mu)$. We then conclude by~\eqref{eq:energy_slope_coloring} that~$\mathcal E(\mathcal D_i)=\mathcal E^*(\mu)$ for all~$i$. Hence,~$\mathcal D_i$ is a maximizer with slope~$\mu$ for all~$i$. 
\end{proof}

For~$\mathrm b,\mathrm w\in G_1$, set~$G_{\mathrm b,\mathrm w}=G_1\backslash \{\mathrm b,\mathrm w\}$, that is, this is the weighted graph obtained by removing~$\mathrm b$ and~$\mathrm w$ and their adjacent edges from~$G_1$. In a similar way, we define~$G_{\mu,\mathrm b,\mathrm w}=G_{1,\mu}\backslash\{\mathrm b,\mathrm w\}$.

\begin{lemma}\label{lem:sum_dimers}
Let~$e=\mathrm b\mathrm w\in E_1$,~$\mathrm b_0\in B_1$ and~$\mathrm w_0'\in W_1$. For any dimer covers~$\mathcal D_{\mathrm b_0,\mathrm w}$ and~$\mathcal D_{\mathrm b,\mathrm w_0'}$ of~$G_{\mathrm b_0,\mathrm w}$ and~$G_{\mathrm b,\mathrm w_0'}$, respectively, there exists dimer covers~$\mathcal D_1$ and~$\mathcal D_{\mathrm b_0,\mathrm w_0'}$ of~$G_1$ and~$G_{\mathrm b,\mathrm w_0'}$ such that
\begin{equation}
\mathcal D_{\mathrm b_0,\mathrm w}\cup\mathcal D_{\mathrm b,\mathrm w_0'}\cup \{e\}=\mathcal D_1\cup\mathcal D_{\mathrm b_0,\mathrm w_0'},
\end{equation}
where the equality is an equality of multisets, that is, it includes multiplicity. In particular, if~$\mathrm b_0\mathrm w_0'=e'\in E_1$, then there is a dimer cover~$\mathcal D_1'$ of~$G_1$ such that 
\begin{equation}
\mathcal D_{\mathrm b_0,\mathrm w}\cup\mathcal D_{\mathrm b,\mathrm w_0'}\cup \{e,e'\}=\mathcal D_1\cup\mathcal D_1'.
\end{equation}
\end{lemma}
\begin{proof}
The union of the left hand side of the first equality is a set of edges that covers~$\mathrm b_0$ and~$\mathrm w_0'$ once, and all other vertices twice. This means that the union forms double edges and disjoint simple paths where all except one are closed. The endpoints of the open path are~$\mathrm b_0$ and~$\mathrm w_0'$, and the path therefore contains an odd number of edges. We may construct~$\mathcal D_1$ and~$\mathcal D_{\mathrm b_0,\mathrm w_0'}$ as follows. Pick an orientation on each closed path, including double edges, and orient the open path from~$\mathrm b_0$ and~$\mathrm w_0'$. Let~$\mathcal D_1$ be the set of edges oriented from a black vertex to a white vertex, and let~$\mathcal D_{\mathrm b_0,\mathrm w_0'}$ be all edges oriented from white vertices to black vertices. 

The second equality follows by noting that~$\mathcal D'=\mathcal D_{\mathrm b_0,\mathrm w_0'}\cup \{e'\}$ is a dimer cover of~$G_1$. 
\end{proof}

Equipped with these two lemmas, we can now define the limiting dimer model.

We denote the universal cover of~$G_{1,\mu}$ by~$G_\mu$, that is,~$G_\mu\subset G$ is the doubly periodic subgraph of~$G$ constructed by periodically extending~$G_{1,\mu}$. Going forward, we denote elements in~$G$ by~$\tilde e$,~$\tilde{\mathrm b}$ and~$\tilde{\mathrm w}$, and their projections in~$G_1$ by~$e$,~$\mathrm b$ and~$\mathrm w$.
\begin{remark}\label{rem:connected_components}
While we could not find a counterexample exhibiting an unbounded connected component of~$G_\mu$ when~$\mathcal E^*$ is strictly concave at~$\mu$, we likewise could not establish a proof that such unbounded components are precluded under the strict concavity condition. See Corollary~\ref{cor:bounded_components} below for the case when there are two maximizers.
\end{remark}

To define the Kasteleyn matrix on the graphs~$G_{1,\mu}$ and~$G_\mu$ with uniform edge weights, we need to introduce a Kasteleyn sign. It turns out that the Kasteleyn sign~$\sigma$ form Section~\ref{sec:measures} is sufficient.  
\begin{lemma}
The Kasteleyn sign~$\sigma$ defined in Section~\ref{sec:measures} is a Kasteleyn sign for the graph~$G_\mu$.
\end{lemma}
\begin{proof}
Let~$e_1,\dots,e_{2n}$ be edges forming a simple loop in~$G_{1,\mu}$. Assume the interior of the loop contains~$m$ vertices when viewed as a subset of~$G$. Since~$\sigma$ is a Kasteleyn sign of the graph~$G$, the alternating product
\begin{equation}
\frac{\sigma(e_1)\dots\sigma(e_{2n-1})}{\sigma(e_2)\dots\sigma(e_{2n})}=(-1)^{n+m+1}.
\end{equation}
See \emph{e.g.},~\cite[Lemma 3.2]{Joh17} or~\cite[Lemma 1]{Ken09}. Since the interior of the loop can be covered by dimers,~$m$ must be even, which proves the statement.
\end{proof}
\begin{remark}
The previous lemma simply says that a Kasteleyn sign of a graph is still a Kasteleyn sign if we remove a bounded subset of the graph which admits a dimer cover.
\end{remark}

We define the Kasteleyn matrix~$K_{G_\mu}:\CC^{B_\mu} \to \CC^{W_\mu}$ of the graph~$G_\mu$ by
\begin{equation}\label{eq:kasteleyn_gibbs_zero}
\left(K_{G_\mu}\right)_{\tilde{\mathrm w}\tilde{\mathrm b}}=\one_{\tilde{\mathrm w}\tilde{\mathrm b}\in E_\mu} \sigma(\tilde{\mathrm w}\tilde{\mathrm b}),
\end{equation}
and the Kasteleyn matrix~$K_{G_{1,\mu}}:\CC^{B_{1,\mu}} \to \CC^{W_{1,\mu}}$ of~$G_{1,\mu}$ by
\begin{equation}\label{eq:kasteleyn_component}
\left(K_{G_{1,\mu}}(z,w)\right)_{\mathrm w\mathrm b}=\one_{\mathrm w\mathrm b\in E_{1,\mu}} \sigma(\mathrm w\mathrm b)\frac{w^{\mathrm w\mathrm b\wedge \gamma_u}}{z^{\mathrm w\mathrm b\wedge \gamma_v}},
\end{equation}
where~$\sigma$ is the Kasteleyn sign from Section~\ref{sec:measures}, cf.~\eqref{eq:kasteleyn_aztec} and~\eqref{eq:kasteleyn_gibbs}. The characteristic polynomial~$P_\mu$ is defined as
\begin{equation}\label{eq:characteristic_polynomial_mu}
P_\mu(z,w)=\det K_{G_{1,\mu}}(z,w),
\end{equation}
cf.~\eqref{eq:characteristic_polynomial_def}.
\begin{proposition}\label{prop:non-zero_polynomial}
Let~$\mu=(\mu_1,\mu_2)\in\mathcal N$ and let~$P_\mu$ be the characteristic polynomial~\eqref{eq:characteristic_polynomial_mu}. If~$\mathcal E^*$ is strictly concave at~$\mu$, then
\begin{equation}
P_\mu(z,w)=\tau Z_{1,\mu} z^{\mu_1}w^{\mu_2},
\end{equation}
for some~$\tau\in \{\pm 1\}$, where~$Z_{1,\mu}$ is the partition function of the dimer covers of~$G_{1,\mu}$.
\end{proposition}
\begin{proof}
For a loop in~$G_1$, we express its homology class in the basis~$\{[\gamma_u],[\gamma_v]\}$, where~$[\gamma_u]$ and~$[\gamma_v]$ are the homology classes of the loops~$\gamma_u$ and~$\gamma_v$, respectively. Let~$\mathcal D_1$ and~$\mathcal D_2$ be dimer covers of~$G_1$, and orient the edges in~$\mathcal D_1$ from a black vertex to a white vertex and the edges in~$\mathcal D_2$ from a white vertex to a black vertex. Their union consists of, say,~$d$ oriented loops. Note that if~$(m,n)$ is the homology class of a loop~$\gamma$, then~$(-n,m)=\mu(\gamma_1)-\mu(\gamma_2)$, where~$\gamma_i=\gamma\cap \mathcal D_i$,~$i=1,2$. 

Assume now that~$\mathcal D_1$ and~$\mathcal D_2$ are dimer covers of~$G_{1,\mu}$, by Lemma~\ref{lem:maximizers_all} they have slope~$\mu$. If their union contains a loop~$\gamma$ in the homology class~$(m,n)\neq (0,0)$, then we can construct a dimer cover of~$G_{1,\mu}$ with slope different from~$\mu$, which proves that all loops lie in the homology class~$(0,0)$. Indeed, change the orientation of the loop~$\gamma$, and let~$\mathcal D_1'$ be the dimer cover consisting of all edges in the union of~$\mathcal D_1$ and~$\mathcal D_2$ oriented from a black vertex to a white vertex, that is, we swap the edges in~$\gamma$. Then~$\mu(\mathcal D_1')=\mu-(-n,m)\neq \mu$.

The fact that the homology class of all loops constructed by taking the union of two dimer covers of~$G_{1,\mu}$ is~$(0,0)$, implies that the loops lift to loops in~$G_\mu$ and we can compute~$P_\mu$ using the standard approach for planar graphs. See, \emph{e.g.}, the proof of~\cite[Theorem 3.1]{Joh17}. More concretely, expanding the determinant~\eqref{eq:characteristic_polynomial_mu}, we obtain, using that~$\mu(\mathcal D)=\mu=(\mu_1,\mu_2)$ for all dimer covers~$\mathcal D$ of~$G_{1,\mu}$,
\begin{equation}
\PP_\mu(z,w)=\sum_{\mathcal D}\sgn(s(\mathcal D))\prod_{e\in \mathcal D}\sigma(e)\frac{w^{e\wedge \gamma_u}}{z^{e\wedge \gamma_v}}
=\sum_{\mathcal D}\tau(\mathcal D)z^{\mu_1}w^{\mu_2},
\end{equation}
where the sum runs over all dimer covers of~$G_{1,\mu}$,~$s(\mathcal D)\in S_{k\ell}$ is the permutation corresponding to the dimer cover~$\mathcal D$, and
\begin{equation}
\tau(\mathcal D)=\sgn(s(\mathcal D))\prod_{e\in \mathcal D}\sigma(e).
\end{equation}
What remains is to show that~$\tau(\mathcal D)=\tau(\mathcal D')$ for any dimer covers~$\mathcal D$ and~$\mathcal D'$ of~$G_{1,\mu}$. This follows from the corresponding argument for planer graphs, as given, \emph{e.g.}, in the proof of~\cite[Theorem 3.1]{Joh17}.
\end{proof}
In the proof of the previous proposition, we saw that there are no loops in the union of two dimer covers of~$G_{1,\mu}$ that wind around the torus. In particular, this immediately implies the following statement, cf. Remark~\ref{rem:connected_components}.
\begin{corollary}\label{cor:bounded_components}
If~$\mathcal E^*$ is strictly concave at~$\mu\in \mathcal N$ and the number of maximizers with slope~$\mu$ is two, then all connected components of~$G_\mu$ are bounded. 
\end{corollary}

For~$\tilde{\mathrm b}=\tilde{\mathrm b}_{\ell m+i,k n+j}$ and~$\tilde{\mathrm w}=\tilde{\mathrm w}_{\ell m'+i',k n'+j'}$ in~$G_\mu$ and with~$\mathrm b$ and~$\mathrm w$ as their projections in~$G_{1,\mu}$, we define
\begin{equation}\label{eq:inverse_kasteleyn_zero}
\left(K_{G_\mu}^{-1}\right)_{\tilde{\mathrm b}\tilde{\mathrm w}}=\frac{1}{(2\pi\i)^2}\int_{|z|=1}\int_{|w|=1}\left(K_{G_{1,\mu}}(z,w)^{-1}\right)_{\mathrm b\mathrm w}\frac{z^{n'-n}}{w^{m'-m}}\frac{\d w}{w}\frac{\d z}{z}.
\end{equation}
If the connected components of~$G_\mu$ are all bounded, see Remark~\ref{rem:connected_components} and Corollary~\ref{cor:bounded_components}, then
\begin{equation}
\left(K_{G_\mu}^{-1}\right)_{\tilde{\mathrm b}\tilde{\mathrm w}}
=
\begin{cases}
\left(K_{G_{1,\mu}'}^{-1}\right)_{\tilde{\mathrm b}\tilde{\mathrm w}} & \text{if } \tilde{\mathrm b} \text{ and } \tilde{\mathrm w} \text{ are in the same component of } G_\mu, \\
0 & \text{if } \tilde{\mathrm b} \text{ and } \tilde{\mathrm w} \text{ are not in the same component of } G_\mu,
\end{cases}
\end{equation}
where~$K_{G_{1,\mu}'}^{-1}$ is the inverse of the finite Kasteleyn matrix~$K_{G_{1,\mu}'}$ of the connected component of~$G_\mu$ that contains~$\tilde{\mathrm b}$ and~$\tilde{\mathrm w}$ . 

Let us define a probability measure~$\PP_\mu$ on dimer covers~$\mathcal D$ of~$G_\mu$ as follows. Given edges~$\tilde e_s=\tilde{\mathrm w}_s\tilde{\mathrm b}_s\in E_\mu$,~$s=1,\dots,p$, the edge probabilities are given by
\begin{equation}\label{eq:gibbs_zero}
\PP_\mu\left[\tilde e_1,\dots,\tilde e_p\in \mathcal D\right]=\det \left(\left(K_{G_\mu}\right)_{\tilde{\mathrm w}_{s'}\tilde{\mathrm b}_{s'}}\left(K_{G_\mu}^{-1}\right)_{\tilde{\mathrm b}_s\tilde{\mathrm w}_{s'}}\right)_{1\leq s,s'\leq p},
\end{equation}
cf.~\eqref{eq:prob_beta} and~\eqref{eq:gibbs_beta}. The fact that~\eqref{eq:gibbs_zero} indeed defines a random point process is a part of the theorem below. 

We are now ready to state and prove the main result of this section.
\begin{theorem}\label{thm:gibbs_limit}
Let~$\mu\in\mathcal N$ and assume~$\mathcal E^*$ is strictly concave at~$\mu$. For~$(x,y)$ in the interior of~$\mathcal A_{t,\mu}$ let~$\PP_{\beta,(x,y)}$ be the probability measure~\eqref{eq:gibbs_beta}. Equation~\eqref{eq:gibbs_zero} (uniquely) defines a probability measure~$\PP_\mu$, and for any~$\tilde e_s=\tilde{\mathrm b}_s\tilde{\mathrm w}_s\in G$, for~$s=1,\dots,p$, there is an~$\eps>0$ such that,
\begin{equation}\label{eq:gibbs_limit}
\PP_{\beta,(x,y)}\left[\tilde e_1,\dots,\tilde e_p\in \mathcal D\right]=\PP_\mu\left[\tilde e_1,\dots,\tilde e_p\in \mathcal D\right]+\Ordo\left(\e^{-\eps\beta}\right),
\end{equation}
as~$\beta \to \infty$.
\end{theorem}
Informally, Theorem~\ref{thm:gibbs_limit} says that the~$\beta \to \infty$ limit of the Gibbs measures is the product of the Gibbs measures corresponding to uniformly weighted dimer models on the (finite or infinite) connected components of~$G_\mu$.
\begin{proof}
It is enough to prove the limiting relation. The fact that~$\PP_\mu$ indeed defines a probability measure follows from~\eqref{eq:gibbs_limit}, because these correlations uniquely determine probabilities of any dimer configuration in a finite window (which is a cylindrical subset for our point process), and the set of probability measures on a finite set is compact. Coherency of probabilities of cylindrical subsets is also clearly preserved by the limit.\footnote{Uniqueness of point process with given correlations in the case of a discrete state space as we have here, is trivial.}

The probabilities are expressed in terms of determinants, so by Leibniz's formula for determinants,
\begin{equation}
\PP_{\beta,(x,y)}\left[\tilde e_1,\dots,\tilde e_p\in \mathcal D\right]=\sum_{\sigma\in S_p}\sgn(\sigma)\prod_{s=1}^p\left(K_{G,\beta,(x,y)}\right)_{\tilde{\mathrm w}_s\tilde{\mathrm b}_s}\left(K_{G,\beta,(x,y)}^{-1}\right)_{\tilde{\mathrm b}_{\sigma(s)}\tilde{\mathrm w}_s}.
\end{equation}
It is therefore enough to compute the limit of finite products of the cyclic form
\begin{equation}\label{eq:correlation_kernel_product}
\prod_{s=1}^p\left(K_{G,\beta,(x,y)}\right)_{\tilde{\mathrm w}_s\tilde{\mathrm b}_s}\left(K_{G,\beta,(x,y)}^{-1}\right)_{\tilde{\mathrm b}_{s+1}\tilde{\mathrm w}_s},
\end{equation}
where all vertices are different and we take~$\tilde{\mathrm b}_{p+1}=\tilde{\mathrm b}_1$. 

By definition of~$K_{G,\beta,(x,y)}^{-1}$, we have
\begin{multline}\label{eq:correltaion_kernal_integral}
\left(K_{G,\beta,(x,y)}\right)_{\tilde{\mathrm w}_s\tilde{\mathrm b}_s}\left(K_{G,\beta,(x,y)}^{-1}\right)_{\tilde{\mathrm b}_{s+1}\tilde{\mathrm w}_s} \\
=\frac{1}{(2\pi\i)^2}\int_{|z|=\e^{\beta x}}\int_{|w|=\e^{\beta y}}\frac{K_{G_1}(z,w)_{\mathrm w_s\mathrm b_s}\adj K_{G_1}(z,w)_{\mathrm b_{s+1} \mathrm w_s}}{P(z,w)}\frac{z^{n_s-n_{s+1}}}{w^{m_s-m_{s+1}}}\frac{\d w}{w}\frac{\d z}{z},
\end{multline}
where~$\tilde{\mathrm b}_s=\tilde{\mathrm b}_{\ell m_s+i_s,k n_s+j_s}$. Recall that the~$(i,j)$-entry of the adjugate of a matrix~$A$ is~$(-1)^{i+j}M_{ji}$, where~$M_{ij}$ is the determinant of the submatrix constructed from~$A$ by deleting its~$i$th row and~$j$th column. By Leibniz's formula for determinants, 
\begin{equation}\label{eq:adjugate_expansion}
\adj K_{G_1}(z,w)_{\mathrm b_{s+1}\mathrm w_s}=\sum_{\mathcal D_{\mathrm b_{s+1},\mathrm w_s}}\tau(\mathcal D_{\mathrm b_{s+1},\mathrm w_s})\e^{\beta \mathcal E(\mathcal D_{\mathrm b_{s+1},\mathrm w_s})} z^{\mu_1(\mathcal D_{\mathrm b_{s+1},\mathrm w_s})}w^{\mu_2(\mathcal D_{\mathrm b_{s+1},\mathrm w_s})},
\end{equation}
where the sum runs over all dimer covers~$\mathcal D_{\mathrm b_{s+1},\mathrm w_s}$ of the graph~$G_{\mathrm b_{s+1},\mathrm w_s}$ and~$\tau(\mathcal D_{\mathrm b_{s+1},\mathrm w_s})\in \{\pm 1\}$ is the combined sign, which we will not make explicit. It follows from iteratively  applying Lemma~\ref{lem:sum_dimers} that if~$\mathcal D_{\mathrm b_{s+1},\mathrm w_s}$,~$s=1,\dots,p$, are dimer covers of~$G_{\mathrm b_{s+1},\mathrm w_s}$ and~$e_s=\mathrm b_s\mathrm w_s\in E_1$, then
\begin{equation}
\sum_{s=1}^p\left(\mathcal E(\mathcal D_{\mathrm b_{s+1},\mathrm w_s})+\log \nu(e_s)\right)=\sum_{s=1}^p\mathcal E(\mathcal D_s),
\end{equation}
for some dimer covers~$\mathcal D_s$,~$s=1,\dots,p$, of~$G_1$. Hence, the product 
\begin{equation}\label{eq:adjugate_product}
\prod_{s=1}^pK_{G_1}(z_s,w_s)_{\mathrm w_s\mathrm b_s}\adj K_{G_1}(z_s,w_s)_{\mathrm b_{s+1} \mathrm w_s}\frac{z_s^{n_s-n_{s+1}}}{w_s^{m_s-m_{s+1}}}
\end{equation}
is a signed sum of terms of the form
\begin{equation}\label{eq:terms_product_first}
\left(\prod_{s=1}^p\e^{\beta \mathcal E(\mathcal D_s)}\right) \prod_{s=1}^pz_s^{\mu_1(\mathcal D_{\mathrm b_{s+1},\mathrm w_s})+\mu_1(e_s)}w_s^{\mu_2(\mathcal D_{\mathrm b_{s+1},\mathrm w_s})+\mu_2(e_s)}\frac{z_s^{n_s-n_{s+1}}}{w_s^{m_s-m_{s+1}}}.
\end{equation}
With the variable change~$(z_s,w_s)\mapsto (z_s\e^{\beta x},w_s\e^{\beta y})$, so that~$|z_s|=|w_s|=1$ in ($p$ instances of)~\eqref{eq:correltaion_kernal_integral}, the quantity~\eqref{eq:terms_product_first} becomes
\begin{equation}\label{eq:terms_product}
\left(\prod_{s=1}^p\e^{\beta \mathcal E(\mathcal D_s)+\beta x\mu_1(\mathcal D_s)+\beta y\mu_2(\mathcal D_s)}\right) \prod_{s=1}^pz_s^{\mu_1(\mathcal D_{\mathrm b_{s+1},\mathrm w_s})+\mu_1(e_s)}w_s^{\mu_2(\mathcal D_{\mathrm b_{s+1},\mathrm w_s})+\mu_2(e_s)}\frac{z_s^{n_s-n_{s+1}}}{w_s^{m_s-m_{s+1}}}.
\end{equation}

When~$|z_s|=|w_s|=1$, the absolute value of the second product of~\eqref{eq:terms_product} is equal to~$1$, and the first factor is bounded above by
\begin{equation}\label{eq:leading_term}
\e^{p\beta \mathcal E^*(\mu)+p\beta x\mu_1+p\beta y\mu_2}.
\end{equation}
Indeed, by definition of~$\mathcal A_{t,\mu}$,
\begin{equation}
\mathcal E(\mathcal D_s)+x\mu_1(\mathcal D_s)+y\mu_2(\mathcal D_s)\leq \mathcal E^*(\mu)+x\mu_1+y\mu_2,
\end{equation}
if~$(x,y)$ is in the interior of~$\mathcal A_{t,\mu}$, cf.~\eqref{eq:characteristic_polynomial_tropical}, and the inequality is strict unless~$\mathcal D_s$ is a maximizer with slope~$\mu$. In fact, the quotient of the first product in~\eqref{eq:terms_product} and~\eqref{eq:leading_term} does not tend to zero as~$\beta\to\infty$ if and only if~$\mathcal D_s$ is a maximizer with slope~$\mu$ for all~$s=1,\dots,p$, which is equivalent to~$\mathcal D_s$ being a dimer cover of~$G_{1,\mu}$, by Lemma~\ref{lem:maximizers_all}. Moreover,~$\mathcal D_s$ is a dimer cover of~$G_{1,\mu}$ for all~$s=1,\dots,p$ if and only if~$\mathcal D_{\mathrm b_{s+1},\mathrm w_s}$ is a dimer cover of~$G_{\mu,\mathrm b_{s+1},\mathrm w_s}$ and~$\tilde{\mathrm e}_s\in E_{1,\mu}$ for all~$s=1,\dots,p$. This implies, via term-by-term convergence, that summing over the terms of~\eqref{eq:terms_product}, including the signs, that do not tend to zero as~$\beta \to \infty$ when divided by~\eqref{eq:leading_term}, after the variable change~$(z_s,w_s)\mapsto (z_s\e^{\beta x},w_s\e^{\beta y})$, is given by
\begin{equation}\label{eq:adjugate_product_expansion}
\e^{p\beta \mathcal E^*(\mu)+p\beta x\mu_1+p\beta y\mu_2}\prod_{s=1}^pK_{G_{1,\mu}}(z_s,w_s)_{\mathrm w_s\mathrm b_s}\adj K_{G_{1,\mu}}(z_s,w_s)_{\mathrm b_{s+1} \mathrm w_s}\frac{z_s^{n_s-n_{s+1}}}{w_s^{m_s-m_{s+1}}}.
\end{equation}

We continue by considering the characteristic polynomial~$P_\beta$. By expanding the determinant, 
\begin{equation}\label{eq:charateristic_polynomial_sum}
P_\beta(z\e^{\beta x},w\e^{\beta y})=\sum_{\mathcal D}\tau(\mathcal D)\e^{\beta \mathcal E(\mathcal D)+\beta x\mu_1(\mathcal D)+\beta y\mu_2(\mathcal D)}z^{\mu_1(\mathcal D)}w^{\mu_2(\mathcal D)},
\end{equation} 
where the sum runs over all dimer covers of~$G_1$ and~$\tau$ is the sign, not the same as in~\eqref{eq:adjugate_expansion}, which we will not make explicit here. Since~$(x,y)$ is in the interior of~$\mathcal A_{t,\mu}$, the leading term of~\eqref{eq:charateristic_polynomial_sum} comes from the terms with~$\mathcal D$ being a maximizer with slope~$\mu$, that is, a dimer cover of~$G_{1,\mu}$. Hence, as~$\beta \to \infty$, 
\begin{equation}\label{eq:charateristic_polynomial_expansion}
P_\beta(z\e^{\beta x},w\e^{\beta y})=\e^{\beta \mathcal E^*(\mu)+\beta x\mu_1+\beta y\mu_2}\left(P_\mu(z,w)+\Ordo\left(\e^{-\eps' \beta}\right)\right),
\end{equation}
for some~$\eps'>0$. Here~$P_\mu$ is the polynomial~\eqref{eq:characteristic_polynomial_mu} and it is non-zero for~$|z|=|w|=1$ by Proposition~\ref{prop:non-zero_polynomial}.

Combining~\eqref{eq:inverse_kasteleyn_zero},~\eqref{eq:correltaion_kernal_integral},~\eqref{eq:adjugate_product_expansion} and~\eqref{eq:charateristic_polynomial_expansion} yields
\begin{equation}
\prod_{s=1}^p\left(K_{G,\beta,(x,y)}\right)_{\tilde{\mathrm w}_s\tilde{\mathrm b}_s}\left(K_{G,\beta,(x,y)}^{-1}\right)_{\tilde{\mathrm b}_{s+1}\tilde{\mathrm w}_s}=\prod_{s=1}^p\left(K_{G,\beta,(x,y)}\right)_{\tilde{\mathrm w}_s\tilde{\mathrm b}_s}\left(K_{G,\beta,(x,y)}^{-1}\right)_{\tilde{\mathrm b}_{s+1}\tilde{\mathrm w}_s}+\Ordo\left(\e^{-\eps \beta}\right),
\end{equation}
as~$\beta \to \infty$, for some~$\eps>0$, which proves the theorem.
\end{proof}

\begin{remark}\label{rem:multiple_maximizers}
As we will see in Section~\ref{sec:subdivision}, the set of maximizers defining the graph~$G_{1,\mu}$ in Definition~\ref{def:maximizer} consists generically of only one dimer cover. However, it is not difficult to construct examples by hand where there are multiple maximizers, and the graph~$G_{1,\mu}$ is constructed from more than one dimer cover. For instance, choose your graph~$G_{1,\mu}$ and define the edge weight function~$\nu$ so that it is small for all edges except those in~$G_{1,\mu}$, which are set to~$1$. Figure~\ref{fig:multiple_maximizers} illustrates two examples constructed in this way. 
\end{remark}

 \begin{figure}[t]
 \begin{center}
  \begin{subfigure}[c]{0.22\textwidth}
\includegraphics[scale=.3]{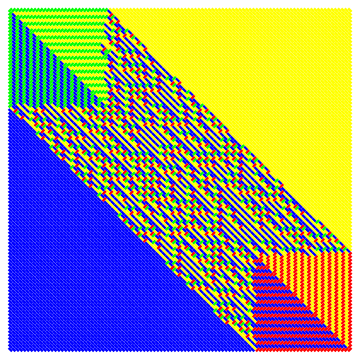}
\end{subfigure}
\quad
  \begin{subfigure}[c]{0.22\textwidth}
\includegraphics[scale=.3]{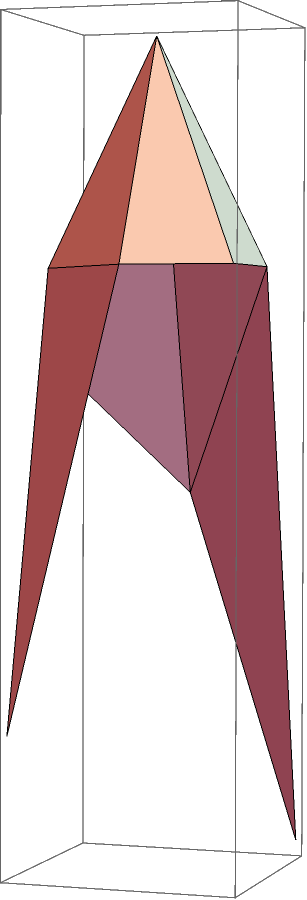}
\end{subfigure}
\quad
  \begin{subfigure}[c]{0.22\textwidth}
\includegraphics[scale=.3]{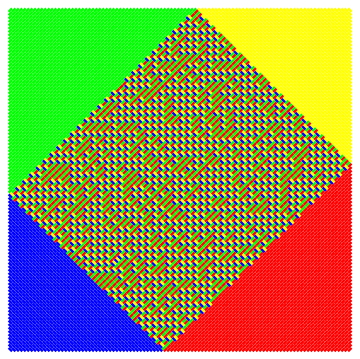}
\end{subfigure}
\quad
  \begin{subfigure}[c]{0.22\textwidth}
\includegraphics[scale=.3]{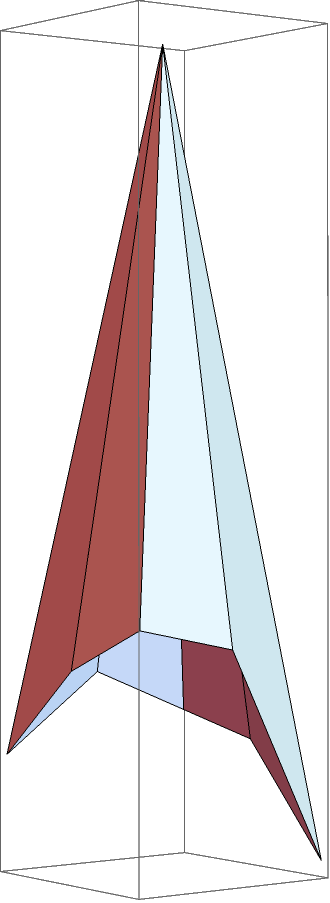}
\end{subfigure}
 \end{center}
\caption{
Two examples of random samples along with (the top boundary of) the associated extended polyhedral domain~$\tilde N(P_t)$. In both examples, the tropical surface tension~$\mathcal E^*$ is strictly convex, yet there are two maximizers with slope~$\mu$ corresponding to the smooth region. Notably, the smooth regions display randomness.
\label{fig:multiple_maximizers}}
\end{figure}

We conclude this section by pointing out that if there is a unique maximizer with slope~$\mu\in \mathcal N$, then~$\mathcal E^*$ is strictly concave at~$\mu$, see Corollary~\ref{cor:concave} below. Consequently, by Theorem~\ref{thm:gibbs_limit},
\begin{equation}
\lim_{\beta\to \infty}\PP_{\beta,(x,y)}\left[\tilde e_1,\dots,\tilde e_p\in \mathcal D\right]=\one_{e_1,\dots,e_p\in \mathcal D},
\end{equation}
where~$\mathcal D$ denotes the unique maximizer with slope~$\mu$. In other words, the probability measures~$\PP_{\beta,(x,y)}$ converge to the delta measure on the graph~$G_\mu$, the graph constructed by periodically extending~$\mathcal D$ to the plane. 

\section{Concavity of the tropical surface tension and generic subdivisions}\label{sec:subdivision}
The goal of this section is to prove that, generically, the tropical curve~$\mathcal A_t$ is a smooth tropical curve, cf. Definition~\ref{def:t-curve}. In the process, we prove that the tropical surface tension~$\mathcal E^*$ is always concave at all~$\mu\in \mathcal N$, cf. Definition~\ref{def:concave}.

Recall the definition of a~$d$-multiweb and an edge-$d$-coloring given above Lemma~\ref{lem:maximizers_all}. Our first statement provides a way to color a~$d$-multiweb into dimer covers with the same slope if such an edge-$d$-coloring exists. This statement is what allows us to prove that~$\mathcal E^*$ is concave.
\begin{proposition}\label{prop:coloring}
Let~$\mu \in\mathcal N$ and, for some~$d\in \ZZ_{>0}$, assume~$\mu_i\in \mathcal N$,~$i=1,\dots,d$, with~$\mu_i\neq \mu$, are such that
\begin{equation}\label{eq:slope_linear_combination}
\sum_{i=1}^d \mu_i=d\!\cdot \! \mu.
\end{equation}
For any dimer covers~$\mathcal D_i$ of~$G_1$ with~$\mu(\mathcal D_i)=\mu_i$, let~$G_{1,d \cdot \mu}\subset G_1$ be the subgraph consisting of all edges contained in the union of~$\mathcal D_i$. Then the~$d$-multiweb of~$G_{1,d\cdot\mu}$ constructed as the union, as multisets, of~$\mathcal D_i$, admits an edge-$d$-coloring of dimer covers~$\mathcal D_i'$ with~$\mu(\mathcal D_i')=\mu$.
\end{proposition}
\begin{proof}
Given the~$d$-multiweb of~$G_{1,d\cdot\mu}$ in the statement, we define a height function~$h_d$ from the faces of~$G_{1,d\cdot\mu}$ to~$\ZZ/d\ZZ$ recursively as follows. Let~$\mathrm f_0$ be the face in~$G_{1,d\cdot\mu}$ adjacent to the black vertices~$\mathrm b_{\ell-1,0}$ and~$\mathrm b_{0,0}$, see Figure~\ref{fig:magnetic_weights}, and set~$h_d(\mathrm f_0)=0$. Let~$\mathrm f$ and~$\mathrm f'$ be two adjacent faces with a common edge of multiplicity~$n$ and such that the black vertex is to the left as we cross from~$\mathrm f'$ to~$\mathrm f$. Then we require~$h_d(\mathrm f)=h_d(\mathrm f')+n$. Since there are~$d$ edges, counting with multiplicity, connected to each vertex,~$h_d$ is well-defined locally. The assumption~\eqref{eq:slope_linear_combination} implies that the function~$h_d$ is well-defined globally as well. Indeed, let~$h_d(\mathrm f_v)$ be the value of~$h_d$ at the face~$\mathrm f_0$ after going along~$\gamma_v$ once. By definition of~$\mu(\mathcal D_i)=(\mu_1(\mathcal D_i),\mu_2(\mathcal D_i))$, see~\eqref{eq:energy_slope}, the curve~$\gamma_v$ intersects~$\mu_2(\mathcal D_i)$ edges of~$\mathcal D_i$. Hence, since~$\gamma_v$ crosses all edges with the black vertex to the left, see Figure~\ref{fig:magnetic_weights},
\begin{equation}
h_d(\mathrm f_v)=\sum_{i=1}^d \mu_2(\mathcal D_i)= d\mu_2\equiv 0 \!\! \mod d.
\end{equation}
Similarly, going along~$\gamma_u$ once does not change the value of~$h_d$. 

Let us color every edge with multiplicity~$n\geq 1$ and with adjacent faces~$\mathrm f$ and~$\mathrm f'$ such that the black vertex of the edge is to the left as we go from~$\mathrm f'$ to~$\mathrm f$, that is,~$h_d(\mathrm f)=h_d(\mathrm f')+n$, by~$h_d(\mathrm f')+1,h_d(\mathrm f')+2,\dots,h_d(\mathrm f')+n=h_d(\mathrm f)$. We denote the collection of edges colored by~$i$ by~$\mathcal D_i'$. By construction,~$\mathcal D_i'$ contains exactly one edge adjacent to each vertex, so it is a dimer cover of~$G_{1,d\cdot\mu}$.

What remains is to show that~$\mu(\mathcal D_i')=\mu=(\mu_1,\mu_2)$ for all~$i$. Since~$\gamma_v$ crosses all edges with the black vertex to the left, the value of~$h_d$ is increasing along~$\gamma_v$. Moreover, by definition of~$h_d$, we cross all colors of the edge-$d$-coloring before repeating any color, and, as observed above,~$\gamma_v$ crosses~$d \mu_2$ edges. In particular, there are~$\mu_2$ edges in~$\mathcal D_i'$ for each~$i$ that crosses~$\gamma_u$. Hence,~$\mu_2(\mathcal D_i')=\mu_2$. Similarly,~$\mu_1(\mathcal D_i')=\mu_1$.
\end{proof}
\begin{remark}
The proof of the previous statement not only establishes the existence of the coloring~$\mathcal D_i'$,~$i=1,\dots,d$, but also provides an explicit construction of the coloring. This is a generalization of C. Frohman's construction (see~\cite[Section 3.6.1]{DKS24}) to the torus. The proof also tells us, if we think of each edge with multiplicity, say,~$n$ as~$n$ edges together with~$n-1$ faces between them, that there is a partition of the faces of~$G_{1,d\cdot\mu}$, constructed so that~$h_d$ is constant on each part of the partition, and the boundary of each part of the partition is equal to the union of two consecutive dimer covers. This is, in fact, how we were led to the proof given above.
\end{remark}

Let~$\mathcal M\simeq \RR^{|E_1|}$ be our parameter space, where the parameters are given by~$\{\log \nu(e)\}_{e\in E_1}\in \RR^{|E_1|}$, and~$|E_1|$ is the number of edges in~$G_1$. We will say that a property holds generically in~$\mathcal M$ if the property is true, outside of a (subset of a) finite number of hyperplanes. We have the following corollary.
\begin{corollary}\label{cor:concave}
The tropical surface tension~$\mathcal E^*$ is concave at~$\mu$ for all~$\mu\in \mathcal N$ and for every point in~$\mathcal M$, and it is strictly concave at~$\mu \in\mathcal N$ generically in~$\mathcal M$. In particular, if there is a unique maximizer with slope~$\mu\in \mathcal N$, then~$\mathcal E^*$ is strictly concave at~$\mu$. 
\end{corollary}
\begin{proof}
Let~$\mathcal D_i$ in the statement of Proposition~\ref{prop:coloring} be a maximizer (Definition~\ref{def:maximizer}) with slope~$\mu_i$,~$i=1,\dots,d$. Then, since~$\mathcal D_i$ and~$\mathcal D_i'$ are edge-$d$-colorings of the same~$d$-multiweb,
\begin{equation}\label{eq:energy_inequality}
\sum_{i=1}^d\mathcal E^*(\mu_i)=\sum_{i=1}^d\mathcal E(\mathcal D_i')\leq d\!\cdot \! \mathcal E^*(\mu).
\end{equation}
Hence,~$\mathcal E^*$ satisfies the inequality in Definition~\ref{def:concave} for all rational~$t_i\in (0,1)$ and hence,~$\mathcal E^*$ is concave. Note that the inequality~\eqref{eq:energy_inequality} is strict unless~$\mathcal D_i'$ is a maximizer for all~$i=1,\dots,d$. Since generically in~$\mathcal M$ there is a unique maximizer for a given slope, the statement follows. 
\end{proof}

We saw in Lemma~\ref{lem:concavity_subdivision} that strict concavity of~$\mathcal E^*$ implies that~$\mu$ is a vertex of the subdivision~$N_S(P_t)$ for all~$\mu\in \mathcal N$. As we will quickly see in the proof below, this implies that all faces in the subdivision are triangles or parallelograms of area~$1/2$ and~$1$, respectively. So the main part of the argument will be to prove that, generically, there are no parallelograms.

\begin{proposition}\label{prop:triangulation_generic}
Generically in~$\mathcal M$, the tropical curve~$\mathcal A_t$ is a smooth tropical curve.
\end{proposition}
\begin{proof}
By Lemma~\ref{lem:concavity_subdivision} and Corollary~\ref{cor:concave}, all points in~$\mathcal N$ are vertices of the subdivision~$N_S(P_t)$ generically in the parameter space~$\mathcal M$. Furthermore, by the definition of~$N_S(P_t)$, see~\eqref{eq:extended_polyhedral_domain}, all faces of~$N_S(P_t)$ are convex. Consequently, each face must be either a triangle with an area of~$1/2$ or a parallelogram with an area of~$1$. 

Indeed, by Pick's theorem, any triangle in~$N_S(P_t)$ has area~$1/2$ and any quadrilateral has area~$1$. Moreover, any convex quadrilateral with the property that drawing any of the diagonals results in two triangles with the same area has to be a parallelogram. To see this, pick two opposite vertices, and let~$d_1$ be the diagonal going between them and~$d_2$ be the other diagonal. The two vertices have to have the same distance to~$d_2$ since the areas of the corresponding triangles are the same. That implies that~$d_2$ divides~$d_1$ in the middle. The same is true for~$d_2$, and, hence, the quadrilateral is a parallelogram. The fact that there cannot be any~$n$-gon with~$n\geq 5$ is readily reduced to the fact that there cannot be any pentagons. If the boundary of a face is a pentagon, we divide it into a triangle and a quadrilateral (which has to be a parallelogram by the above) in two different ways and conclude that two adjacent edges have to be parallel, which cannot happen.

Let~$\mu_i$,~$i=1,\dots,4$ be vertices of a parallelogram in~$N_S(P_t)$, with~$\mu_1$ and~$\mu_2$ being opposite to each other and~$\mu_3$ and~$\mu_4$ being opposite to each other. Generically, we may assume that there is a unique maximizer~$\mathcal D_i$ with slope~$\mu_i$ for~$i=1,\dots,4$. Then
\begin{equation}\label{eq:parallelogram}
\mathcal E(\mathcal D_1)+\mathcal E(\mathcal D_2)=\mathcal E(\mathcal D_3)+\mathcal E(\mathcal D_4).
\end{equation}
We will see that the set where this occurs is contained in a hyperplane in~$\mathcal M$. 

As in the proof of Proposition~\ref{prop:non-zero_polynomial}, we orient the edges in~$\mathcal D_1$ from black to white vertices and the edges in~$\mathcal D_2$ from white to black vertices. The union of~$\mathcal D_1$ and~$\mathcal D_2$ then consists of double edges and oriented loops in~$G_1$. Let~$\gamma_i$,~$i=1,\dots,d$, be all such loops. Note that~$\gamma_i$ belongs to the homology class~$(m_i,n_i)$ in the basis~$\{[\gamma_u],[\gamma_v]\}$ where
\begin{equation}\label{eq:homology_orthogonal}
(-n_i,m_i)=\mu(\gamma_{i,1})-\mu(\gamma_{i,2}),
\end{equation} 
and~$\gamma_{i,j}=\gamma_i\cap \mathcal D_j$,~$j=1,2$, with the notation~\eqref{eq:energy_slope_2}. Compare with the proof of Proposition~\ref{prop:non-zero_polynomial}. Moreover,
\begin{equation}\label{eq:homology_total}
\sum_{i=1}^d\left(\mu(\gamma_{i,1})-\mu(\gamma_{i,2})\right)=\mu(\mathcal D_1)-\mu(\mathcal D_2).
\end{equation}
In particular, since the loops are disjoint, the loops in a non-zero homology class are parallel to each other, and by~\eqref{eq:homology_orthogonal} and~\eqref{eq:homology_total} their homology classes are orthogonal to~$\mu(\mathcal D_1)-\mu(\mathcal D_2)$. Moreover, there exists at least one loop in a non-zero homology class orthogonal to~$\mu(\mathcal D_1)-\mu(\mathcal D_2)$.

Similarly, there is a loop in the union of~$\mathcal D_3$ and~$\mathcal D_4$ in a non-zero homology class orthogonal to~$\mu(\mathcal D_3)-\mu(\mathcal D_4)$. Since~$\mu(\mathcal D_1)-\mu(\mathcal D_2)=\mu_1-\mu_2$ and~$\mu(\mathcal D_3)-\mu(\mathcal D_4)=\mu_3-\mu_4$ are not parallel, we conclude that there are edges in the union of~$\mathcal D_1$ and~$\mathcal D_2$ that are not in the union of~$\mathcal D_3$ and~$\mathcal D_4$, and \emph{vise versa}. Hence,~\eqref{eq:parallelogram} becomes a linear equation on the edge weighs which are not in both unions, counting with multiplicity,
\begin{equation}\label{eq:parallelogram_equation}
\sum_{e\in \mathcal D_{1,2}}\log \nu(e)=\sum_{e\in \mathcal D_{3,4}}\log \nu(e),
\end{equation}
where~$\mathcal D_{1,2}\neq \emptyset$ are the edges in~$\mathcal D_1$ or~$\mathcal D_2$ that are not in~$\mathcal D_3$ or~$\mathcal D_4$, and~$\mathcal D_{3,4}\neq \emptyset$ are the ones in~$\mathcal D_3$ or~$\mathcal D_4$ but not in~$\mathcal D_1$ or~$\mathcal D_2$. In both unions, we are counting with multiplicity. As~\eqref{eq:parallelogram_equation} is an equation of a hyperplane in~$\mathcal M$, the proof is complete.
\end{proof}

\appendix

\section{The zero-temperature limit of the surface tension}\label{appendix:surface_tension}
We prove in this section that the limit of the surface tension as~$\beta \to \infty$ indeed is equal to (minus) the tropical surface tension. We begin by recalling the definition of the \emph{Ronkin function} and the \emph{surface tension}. For further properties and details, we refer to~\cite{KOS06, Mik04a} and references therein. 

The Ronkin function~$R_\beta:\RR^2\to \RR$ of the polynomial~$P_\beta$ is defined by
\begin{equation}
R_\beta(x,y)=\frac{1}{(2\pi\i)^2}\int_{|z|=1}\int_{|w|=1}\log\left|P_\beta\left(\e^{\beta x} z,\e^{\beta y} w\right)\right|\frac{\d w}{w}\frac{\d z}{z},
\end{equation}
and the surface tension~$\sigma_\beta:N(P)\to \RR$ is the Legendre transform of the Ronkin function,
\begin{equation}\label{eq:surface_tension_legendre}
\sigma_\beta(\mu)=\max_{(x,y)\in \RR^2}(-R_\beta(x,y)+\mu_1\beta x+\mu_2 \beta y).
\end{equation}
The scaling~$\beta x$ and~$\beta y$ is done to match the scaling we used in the definition of the amoeba~$\mathcal A_\beta$, and it is the appropriate scaling in the zero-temperature limit.  

\begin{proposition}
If the tropical surface tension~$\mathcal E^*$ is strictly concave at~$\mu\in \mathcal N$, then
\begin{equation}
\beta^{-1}\sigma_\beta(\mu)\to-\mathcal E^*(\mu)
\end{equation}
as~$\beta \to \infty$.
\end{proposition}
\begin{remark}
If we extend~$\mathcal E^*$ to a piecewise linear continuous function on~$N(P)\subset \RR^2$ so that its graph coincides with the top boundary of the extended polyhedral domain~$\tilde N(P_t)$ (cf. Remark~\ref{rem:arctic_tropical_curve}), we expect that
\begin{equation}
\beta^{-1}\sigma_\beta\to-\mathcal E^*
\end{equation}
uniformly as~$\beta\to\infty$, and this limit holds without any restrictions on the edge weights. While a stronger result of this nature could be of interest, it is not essential for this paper, and we will therefore not provide a proof.  
\end{remark}
\begin{proof}
Let~$(x,y)\in \RR^2$ be in the interior of~$\mathcal A_{t,\mu}$. By definition of~$P_t$~\eqref{eq:characteristic_polynomial_tropical}, 
\begin{equation}
P_t(x,y)=\mathcal E^*(\mu)+x\mu_1+y\mu_2,
\end{equation}
where~$\mu=(\mu_1,\mu_2)$, and the maximum in the definition is uniquely attained. If~$\beta$ is large enough, then, by~\eqref{eq:limit_amoeba},
\begin{equation}
\sigma_\beta(\mu)=-R_\beta(x,y)+\mu_1\beta x+\mu_2 \beta y,
\end{equation}
since the maximum in~\eqref{eq:surface_tension_legendre} is attained by all~$(x,y)$ in the interior of~$\mathcal A_{\beta,\mu}$. Moreover, from~\eqref{eq:characteristic_polynomial_finite}, we have
\begin{equation}
\log|P_\beta(\e^{\beta x} z,\e^{\beta y} w)|=\beta P_t(x,y)+n_\mu + \Ordo\left(\e^{-\beta \eps}\right),
\end{equation}
as~$\beta\to \infty$, for some~$\eps>0$ uniformly for~$|z|=|w|=1$, and where~$n_\mu$ is the number of maximizers of~$G_1$ with slope~$\mu$. Hence,
\begin{equation}
\beta^{-1}R_\beta(x_\beta,y_\beta)\to P_t(x,y).
\end{equation}
We conclude that
\begin{equation}
\beta^{-1}\sigma_\beta(\mu)\to-P_t(x,y)+\mu_1 x+\mu_2 y=-\mathcal E^*(\mu).
\end{equation}
as~$\beta \to \infty$. 
\end{proof}

\bibliographystyle{plain}
\bibliography{bibliotek}

\end{document}